\documentclass[11pt]{article}

\usepackage{amsmath}
\usepackage{amsthm}
\usepackage{amssymb}
\usepackage{enumerate}
\usepackage{amssymb}
\usepackage{tikz}
\usepackage{hyperref}

\usepackage[top=1.4in, bottom=1.4in, left=1.2in, right=1.2in]{geometry}

\theoremstyle{definition} \newtheorem{definition}{Definition}[subsection]
\theoremstyle{definition} \newtheorem{remark}[definition]{Remark}
\theoremstyle{definition} 
\theoremstyle{definition} 
\theoremstyle{definition} \newtheorem{notation}[definition]{Notation}
\theoremstyle{definition} \newtheorem{fact}[definition]{Fact}
\theoremstyle{definition} \newtheorem{convention}[definition]{Convention}
\theoremstyle{plain} \newtheorem{theorem}[definition]{Theorem}
\theoremstyle{plain} \newtheorem{claim}[definition]{Claim}
\theoremstyle{plain} 
\theoremstyle{plain} 
\theoremstyle{plain} \newtheorem{lemma}[definition]{Lemma}
\theoremstyle{plain} 
\theoremstyle{plain} \newtheorem{proposition}[definition]{Proposition}
\theoremstyle{plain} \newtheorem{corollary}[definition]{Corollary}
\theoremstyle{definition} \newtheorem{observation}[definition]{Observation}
\theoremstyle{plain}

\newcommand{\vu}{\vec u}
\newcommand{\vv}{\vec v}
\newcommand{\vx}{\vec x}

\newcommand{\vz}{\vec z}

\newcommand{\aaa}{\mathfrak{A}}
\newcommand{\sss}{\mathfrak{S}}

\newcommand{\bbb}{\mathfrak B}
\renewcommand{\ggg}{\mathfrak{G}}
\newcommand{\hhh}{\mathfrak{H}}
\newcommand{\jjj}{\mathfrak J}
\renewcommand{\lll}{\mathfrak L}
\newcommand{\ppp}{\mathfrak{P}}
\newcommand{\xxx}{\mathfrak X}
\newcommand{\yyy}{\mathfrak Y}
\newcommand{\zzz}{\mathfrak Z}

\newcommand{\xx}{\mathfrak x}
\newcommand{\yy}{\mathfrak y}
\newcommand{\zz}{\mathfrak z}

\newcommand{\E}{\mathbb E}

\renewcommand{\P}{\mathbb P}
\newcommand{\Z}{\mathbb Z}

\newcommand{\new}{\mathrm{new}}
\newcommand{\old}{\mathrm{old}}
\newcommand{\trans}{\mathrm{trans}}
\newcommand{\jh}{\mathrm{Jh}}

\newcommand{\calB}{\mathcal{B}}
\newcommand{\calC}{\mathcal{C}}
\newcommand{\calD}{\mathcal{D}}
\newcommand{\calE}{\mathcal{E}}
\newcommand{\calF}{\mathcal{F}}
\newcommand{\calFbar}{\overline{\mathcal{F}}}
\newcommand{\calH}{\mathcal{H}}
\newcommand{\calK}{\mathcal{K}}
\newcommand{\calL}{\mathcal{L}}
\newcommand{\calP}{\mathcal{P}}
\newcommand{\calR}{\mathcal{R}}

\newcommand{\sets}{\mathsf{Sets}}
\newcommand{\coloredsets}{\mathsf{ColoredSets}}
\newcommand{\partitionedsets}{\mathsf{PartitionedSets}}

\newcommand{\ellk}{k}
\newcommand{\vf}{\varphi}

\newcommand{\Cbar}{\overline C}

\newcommand{\Ebar}{\overline E}
\newcommand{\Gbar}{\overline G}
\newcommand{\Mbar}{\overline M}
\newcommand{\Rbar}{\overline R}

\newcommand{\Wbar}{\overline W}

\newcommand{\pibar}{\overline{\pi}}
\newcommand{\psibar}{\overline{\psi}}
\newcommand{\sigmabar}{\overline{\sigma}}
\newcommand{\normal}{\triangleleft}

\newcommand{\ie}{i.\,e.}

\newcommand{\wto}{\widetilde O}
\newcommand{\wtw}{\widetilde W}

\newcommand{\mn}{\medskip\noindent}
\newcommand{\sn}{\smallskip\noindent}

\DeclareMathOperator{\DL}{DL}
\DeclareMathOperator{\hyp}{Hyp}

\DeclareMathOperator{\aff}{Aff}

\DeclareMathOperator{\supp}{supp}
\DeclareMathOperator{\aut}{Aut}
\DeclareMathOperator{\out}{Out}

\DeclareMathOperator{\iso}{Iso}
\newcommand{\sym}{\sss}
\newcommand{\alt}{\aaa}

\DeclareMathOperator{\diag}{diag}

\DeclareMathOperator{\poly}{poly}

\DeclareMathOperator{\soc}{Soc}
\DeclareMathOperator{\pker}{Pker}

\DeclareMathOperator{\gl}{GL}

\def\bfor{{\bf for}}
\def\bwhile{{\bf while}}

\def\bendwhile{{\bf end(while)}}

\def\bendfor{{\bf end(for)}}
\def\bif{{\bf if}}

\def\bthen{{\bf then}}

\def\belse{{\bf else}}
\def\breturn{{\bf return}}
\def\bexit{{\bf exit}}

\begin{document}

\begin{center}
\LARGE{Graph Isomorphism in Quasipolynomial Time}
\end{center}

\begin{center}
{\large L\'aszl\'o Babai \\
        University of Chicago}
\end{center}

\begin{center}
{\large 2nd preliminary version} \\ 
January 19, 2016 \\
\end{center}

\begin{abstract}
We show that the Graph Isomorphism (GI) problem and the related problems
of String Isomorphism (under group action) (SI) and Coset Intersection (CI)
can be solved in\ quasipolynomial\ \ ($\exp\left((\log
n)^{O(1)}\right)$)\ \ time.\quad  The\ best\ previous\ bound\ for\ GI\
was\ \ $\exp(O(\sqrt{n\log n}))$, where $n$ is the number
of vertices (Luks, 1983); for the other two problems, the bound was
similar, $\exp(\wto(\sqrt{n}))$, where $n$ is the size of the
permutation domain (Babai, 1983).

The algorithm builds on Luks's SI framework and attacks the barrier
configurations for Luks's algorithm by group theoretic ``local
certificates'' and combinatorial canonical partitioning techniques.
We show that in a well--defined sense, Johnson graphs are the only
obstructions to effective canonical partitioning.

Luks's barrier situation is characterized by a homomorphism $\vf$
that maps a given permutation group $G$ onto $S_k$ or $A_k$,
the symmetric or alternating group of degree $k$, where $k$
is not too small.  We say that an element $x$ in the 
permutation domain on which $G$ acts is \emph{affected} by $\vf$
if the $\vf$-image of the stabilizer of $x$ does not
contain $A_k$.  The affected/unaffected dichotomy 
underlies the core ``local certificates'' routine and
is the central divide-and-conquer tool of the algorithm.
\end{abstract}

\vspace{2cm}
\noindent
For a list of updates compared to the first arXiv version,
see the last paragraph of the Acknowledgments.

\vspace{2.5cm}

\mn
$^*$ Research supported in part by NSF Grants CCF-7443327 (2014-current),
     CCF-1017781 (2010-2014), and CCF-0830370 (2008--2010).
     Any opinions, findings, and 
     conclusions or recommendations expressed in this paper are
     those of the author and do not necessarily reflect the
     views of the National Science Foundation (NSF).

\newpage
\tableofcontents

\section{Introduction}

\subsection{Results and philosophy}
\subsubsection{Results: the String Isomorphism problem}
Let $G$ be a group of permutations of the set $[n]=\{1,\dots,n\}$
and let $\xx,\yy$ be strings of length $n$ over a finite alphabet.
The \emph{String Isomorphism (SI) problem} asks, given $G$,
$\xx$, and $\yy$, does there exist an element of $G$ that
transforms $\xx$ into $\yy$.  (See the precise definition in 
Def.~\ref{def:string}.  Permutation groups are given by a list 
of generators.) A function $f(n)$ is \emph{quasipolynomially
bounded} if there exist constants $c,C$ such that 
$f(n)\le \exp(C(\log n)^c)$ for all sufficiently large $n$.
``Quasipolynomial time'' refers to quasipolynomially bounded time.

We prove the following result.
\begin{theorem}
The String Isomorphism problem can be solved in quasipolynomial
time.
\end{theorem}

The Graph Isomorphism (GI) Problem asks to decide whether
or not two given graphs are isomorphic.  The Coset Intersection (CI)
problem asks, given cosets of two permutation groups over
the same finite domain, do they have a nonempty intersection.

\begin{corollary}
The Graph Isomorphism problem and the Coset Intersection problem
can be solved in quasipolynomial time.
\end{corollary}

The SI and CI problems were introduced by Luks~\cite{luks-bded}
(cf.~\cite{luks-dimacs})
who also pointed out that these problems are polynomial-time
equivalent (under Karp reductions) and GI easily reduces to either.
For instance, GI for graphs with $n$ vertices is identical, under
obvious encoding, with SI for binary strings of length
$\binom{n}{2}$ with respect to the induced action of the symmetric
group of degree $n$ on the set of $\binom{n}{2}$ unordered pairs.

The previous best bound for each of these three problems was
$\exp(\wto(n^{1/2}))$ (the tilde hides polylogarithmic factors%
\footnote{Accounting for those logs, the best bound for GI for three
decades was $\exp(O(\sqrt{n\log n}))$, established by Luks in 1983,
cf.~\cite{bkl}.}), where for GI, $n$ is the number of vertices, for
the two other problems, $n$ is the size of the permutation domain.
For GI, this bound was obtained in 1983 by combining Luks's
group-theoretic algorithm~\cite{luks-bded} with a combinatorial
partitioning lemma by Zemlyachenko (see~\cite{zemlyachenko, canonical,
  bkl}).  For SI and CI, additional group-theoretic observations were
used (\cite{doktori}, cf.~\cite{bkl}).  No improvement over either of
these results was found in the intervening decades.

The actual results are slightly stronger: only the length of the largest
orbit of $G$ matters.

\begin{theorem}
The SI problem can be solved in time, polynomial in $n$ (the length of
the strings) and quasipolynomial in $n_0(G)$,
the length of the largest orbit of $G$.
\end{theorem}

The first class of graphs studied using group theory was that of
vertex-colored graphs (isomorphisms preserve color by 
definition)~\cite{lasvegas} (1979).

\begin{corollary}
The GI problem for vertex-colored graphs can be solved in time, 
polynomial in $n$ (the number of vertices) and quasipolynomial 
in the largest color multiplicity.
\end{corollary}

\subsubsection{Quasipolynomial complexity analysis, multiplicative cost}
The analysis will be guided by the observation that if $f(x)$ and $q(x)$ 
are positive, monotone increasing functions and
\begin{equation}   \label{eq:recurrence0}
        f(n) \le q(n)f(9n/10)
\end{equation}
then $f(n) \le q(n)^{O(\log n)}$.  In  particular, if $q(x)$
is quasipolynomially bounded then so is $f(x)$.  Here $f(n)$ stands
for the worst-case cost (number of group operations)
on instances of size $\le n$ and $q(n)$
is the branching factor in the algorithm to which we refer
as the ``multiplicative cost'' in reference to 
Eq.~\eqref{eq:recurrence0}.  So our goal will be
to achieve a \emph{significant reduction} in the problem size,
say $n\leftarrow 9n/10$, at a quasipolynomial multiplicative cost.
There is also an additive cost to the reduction,
but this will typically be absorbed by the multiplicative cost.

Both in the group-theoretic and in the combinatorial arguments, we
shall actually use double recursion.  In addition to the input domain
$\Omega$ of size $n=|\Omega|$, we shall also build an auxiliary set
$\Gamma$ and track its size $m=|\Gamma|\le n$.  Most of the action
will occur on $\Gamma$.  Significant progress will be deemed to have
occurred if we significantly reduce $\Gamma$, say $m\leftarrow 9m/10$,
while not increasing $n$.  When $m$ drops below a threshold $\ell(n)$
that is polylogarithmic in $n$, we perform brute force enumeration
over the symmetric group $\sym(\Gamma)$ 
at a multiplicative cost of $\ell(n)!$.  This
eliminates the current $\Gamma$ and significantly reduces $n$.
Subsequently a new $\Gamma$, of size $m\le n$, is introduced, and the
game starts over.  If $q_1(x)$ is the multiplicative cost of
significantly reducing $\Gamma$ then the overall cost estimate becomes
$f(n)\le q_1(x)^{\log^2 n}$.

\subsubsection{Philosophy: local to global}
We follow Luks's general framework~\cite{luks-bded}, developed for his
celebrated polynomial-time algorithm to test isomorphism of graphs of
bounded valence.  

Luks's method meets a barrier when it encounters large primitive
permutation groups without well--behaved subgroups of small index.  By
a 1981 result of Cameron~\cite{cameron}
these barrier groups are the ``Johnson groups'' (symmetric or alternating
groups in their induced action on $t$-tuples).  This much has been
known for more than 35 years.
Our contribution is in breaking this symmetry.

We achieve this goal via a new group-theoretic divide-and-conquer
algorithm that handles Luks's barrier situation.  The tools we develop
include two combinatorial partitioning algorithms and a
group-theoretic lemma.  The latter serves as the main
divide-and-conquer tool for the procedure at the heart of the
algorithm.

Our strategy is an
interplay between local and global symmetry, formalized
through a technique we call \emph{``local certificates.''}
We shall certify both the presence and the absence of local symmetry.
Locality in our context refers to logarithmic-size subdomains.
If we don't find local symmetry, we infer global irregularity; if
we do find local symmetry, we either build it up to global symmetry
(find automorphisms)
or use any obstacles to this attempt to infer local irregularity
which we then synthesize to global irregularity.
Building up symmetry means approximating the automorphism group
from below; when we discover irregularity,
we turn that into an effective upper bound on the automorphism
group (by reducing the ambient group $G$ using combinatorial
partitioning techniques, including polylog-dimensional
Weisfeiler--Leman refinement\footnote{As far as I know,
this paper is the first to derive analyzable
gain from employing the $k$-dimensional WL method for 
unbounded values of $k$.  We
use it in the proof of the Design Lemma (Thm.~\ref{thm:design}).
In our applications of the Design Lemma, the value of 
$k$ is polylogarithmic (see Secs. 
\ref{sec:localguide},~\ref{sec:aggregate}).}).  
When the lower and the upper bounds on the automorphism group meet,
we have found the automorphism group and solved the isomorphism
problem.

The critical contribution of the group-theoretic {\sf LocalCertificates}
routine (Theorem~\ref{thm:local-certificates}) is the
following, somewhat implausible, dichotomy: either the algorithm finds 
\emph{global} automorphisms (automorphisms
of the entire input string) that certify high \emph{local symmetry},
or it finds a \emph{local} obstruction to high local symmetry.
We shall explain this in more specific terms.

\subsection{Strategy}

\subsubsection{Notation, terminology. Giants, Johnson groups}
For groups $G, H$ we write $H\le G$ to indicate that $H$ is a subgroup
of $G$.

For a set $\Gamma$ we write $\sym(\Gamma)$ to denote the symmetric
group acting on $\Gamma$, and $\alt(\Gamma)$ for the alternating group.
We refer to these two subgroups of $\sym(\Gamma)$ as the \emph{giants}.
If $|\Gamma|=m$ then we also generically write $\sss_m$ and $\aaa_m$ 
for the giants acting on $m$ elements.   We say that a homomorphism
$\vf : G\to\sym(\Gamma)$ is a \emph{giant representation} if
the image $G^{\vf}$ is a giant (\ie, $G^{\vf}\ge \aaa(\Gamma)$).

We write $\sym^{(t)}(\Gamma)$ for the induced action of $\sym(\Gamma)$
on the set $\binom{\Gamma}{t}$ of $t$-tuples of elements of $\Gamma$.
We define $\alt^{(t)}(\Gamma)$ analogously.  We call the groups
$\sym^{(t)}(\Gamma)$ and $\aaa^{(t)}(\Gamma)$ \emph{Johnson groups}
and also denote them by $\sss^{(t)}_m$ and $\aaa^{(t)}_m$ if $|\Gamma|=m$.
Here we assume $1\le t\le m/2$.

The input to the String Isomorphism problem is a permutation group
$G\le\sym(\Omega)$ and two strings $\xx,\yy:\Omega\to\Sigma$
where $\Sigma$ is a finite alphabet.   For $\sigma\in \sym(\Omega)$, 
the string $x^{\sigma}$
is obtained from $\xx$ by permuting the positions of the entries via 
$\sigma$.  The set $\iso_G(\xx,\yy)$ of $G$-isomorphisms of the
strings $\xx,\yy$ consist of those $\sigma\in G$ that satisfy
$\xx^{\sigma}=\yy$, and the $G$-automorphism group of $\xx$
is the group $\aut_G(\xx)=\iso_G(\xx,\xx)$.  The set
$\iso_G(\xx,\yy)$ is either empty or a right coset
of $\aut_G(\xx)$.

\subsubsection{Local certificates}
Luks's SI algorithm proceeds by processing the permutation group
$G\le\sym(\Omega)$ orbit by orbit.  If $G$ is transitive, it
finds a minimal system $\Phi$ of imprimitivity ($\Phi$ is a $G$-invariant
partition of the permutation domain $\Omega$ into maximal blocks),
so the action $\ggg\le\sym(\Phi)$ is a primitive permutation group.
The naive approach then is to enumerate all elements of $\ggg$,
each time reducing to the kernel of the $G\to\ggg$ epimorphism.

By Cameron's cited result, the barrier to efficient application of 
this method occurs when $\ggg$ is a Johnson group, $\sss^{(t)}_m$
or $\aaa^{(t)}_m$, for some value $m$ deemed too large to
permit full enumeration of $\ggg$.  (Under brute force 
enumeration, the number $m$ will go into the exponent of 
the complexity.)   We shall set this threshold at $c\log n$
for some constant $c$.

It is easy to verify recursively whether or not $\aut_G(\xx)$ maps
onto $\ggg$, or a small-index subgroup of $\ggg$; and if the answer is
positive, we can also find the $G$-isomorphisms of $\xx$ and $\yy$ via
efficient recursion.

So our goal is to significantly reduce $\ggg$ unless $\aut_G(\xx)$
maps onto a large portion of $\ggg$.

First we note that our Johnson group $\ggg$, a quotient of $G$, is
isomorphic to a symmetric or alternating group, $\sss_m$ or $\aaa_m$,
so $G$ has a giant representation $\vf : G\to \sym(\Gamma)$ for some
domain $\Gamma$ of size $|\Gamma|=m$.

Virtually all the action in our algorithm will occur on the set
$\Gamma$.  By ``locality'' we shall refer to logarithmic-size subsets
of $\Gamma$ which we call \emph{test sets}.  If $A\subseteq \Gamma$ is
a test set, we say that $A$ is \emph{full} if all permutations in
$\alt(A)$ lift to $G$-automorphisms of the input string $\xx$.

One of our main technical result says, somewhat surprisingly, that
we can efficiently decide local symmetry (fullness of a test set):
in quasipolynomial time we are able to reduce the question
of fullness to $N$-isomorphism questions, where the groups $N$
have no orbit larger than $n/m$.

This is a significant reduction, it permits efficient application
of Luks's recurrence.

Our test sets (small subsets of $\Gamma$)
have corresponding portions of the input
(defined on subsets of $\Omega$); we refer to elements of $G$
that respect such a portion of the input as \emph{local
  automorphisms}.  The obstacles to local symmetry we find are local
(they already block local automorphisms from being sufficiently rich),
so our certificates to non-fullness will be local.  Fullness
certificates, however, by definition cannot be local: we need global
automorphisms (that respect the full input) to certify fullness.  The
surprising part is that from a sufficiently rich set of local
automorphisms we are able to infer global automorophisms
(automorphisms of the full string).

The nature of these certificates is outlined in
Sec.~\ref{sec:local-certificates-explained}.

\subsubsection{Aggregating the local certificates}
The next phase is that we aggregate these $\binom{m}{k}$ local
certificates (where $k=|A|$ is the size of the test sets; we shall
choose $k$ to be $O(\log n)$) into global information.  In fact, not
only do we study test sets $A$ but compare pairs $A, A'$ of test sets,
and we also compare test sets for the input $\xx$ and for the input
$\yy$, so our data for the aggregation procedures take about
$\binom{m}{k}^2$ items of local information as input.

Aggregating the positive certificates is rather simple; these
are subgroups of the automorphism group, so we study 
the group $F$ they generate, and the structure of its
projection $F^{\Gamma}$ into $\sym(\Gamma)$.   If this group is
all of $\sym(\Gamma)$ then $\xx$ and $\yy$ are $G$-isomorphic
if and only if they are $N$-isomorphic where $N=\ker(\vf)$.  
The situation is not much more difficult when
$F^{\Gamma}$ acts as a giant on a large portion of $\Gamma$
(Section~\ref{sec:topaction}).

Otherwise, if $F^{\Gamma}$ has a large support in $\Gamma$
but is not a giant on a large orbit of this support, then
we can take advantage of the structure of $F^{\Gamma}$ 
(orbits, domains of imprimitivity) to obtain the desired
split of $\Gamma$ (Section~\ref{sec:aggregate}).

The aggregate of the negative certificates will be a 
canonical $k$-ary relational structure on $\Gamma$
and the subject of our combinatorial reduction techniques
(Design Lemma, Sec.~\ref{sec:tools1-designlemma},
and Split-or-Johnson algorithm, Sec.~\ref{sec:johnson}) which,
in combination, will achieve the desired reduction of $\Gamma$.

\subsubsection{Group theory required}
The algorithm heavily depends on the Classification of Finite Simple
Groups (CFSG) through Cameron's classification of large primitive
permutation groups.  Another instance where we rely on CFSG occurs in
the proof of Lemma~\ref{lem:prim-topquotient} that depends on
``Schreier's Hypothesis'' (that the outer automorphisms group of a
finite simple group is solvable), a consequence of CFSG.

No deep knowledge of group theory is required for reading this
paper.  The cited consequences of the CFSG are simple to state,
and aside from these, we only use elementary group theory.

We should also note that we are able to dispense with
Cameron's result using our combinatorial partitioning technique,
significantly reducing the dependence of our analysis on the CFSG.
We comment on this in Section~\ref{sec:CFSG}.

\subsection{The ingredients}

The algorithm is based on Luks's classical framework.  It has four
principal new ingredients: a group-theoretic result, a group-theoretic
divide-and-conquer algorithm, and two combinatorial partitioning
algorithms.  Both the group-theoretic algorithm and part of one of
the combinatorial partitioning algorithms implement the idea
of ``local certificates.''

\subsubsection{The group-theoretic divide-and-conquer tool}
In this section we describe our main group theoric tool.

Let $G\le\sym(\Omega)$ be a permutation group.
Recall that we say that a homomorphism $\vf : G\to\sss_k$
is a \emph{giant representation} of $G$ if $G^{\,\vf}$ (the
image of $G$ under $\vf$) contains $\aaa_k$.
We say that an element $x\in\Omega$ is \emph{affected} by $\vf$ 
if $G_x^{\,\vf}\ngeq \aaa_k$,
where $G_x$ denotes the stabilizer of $x$ in $G$.  Note that
if $x$ is affected then every element of the orbit $x^G$ is affected.
So we can speak of affected orbits.

\begin{theorem}
\label{thm:unaffected1}
Let $G\le\sym(\Omega)$ be a permutation group and 
let $n_0$ denote the length of the largest orbit of $G$.  
Let $\vf : G\to\sss_k$ be a giant representation.  
Let $U\subseteq\Omega$ denote the set of
elements of $\Omega$ not affected by $\vf$.
Then the following hold.
\begin{enumerate}[(a)] 
\item   \emph{(Unaffected Stabilizer Theorem)}  \label{item:unaffected}
Assume $k >\max\{8, 2+\log_2 n_0\}$.  Then
$\vf$ maps $G_{(U)}$, the pointwise stabilizer of\, $U$,
onto $\aaa_k$ or $\sss_k$ (so $\vf : G_{(U)}\to\sss_k$ is
still a giant representation).
In particular, $U\neq \Omega$ (at least one element is affected).
\item  \emph{(Affected Orbits Lemma)} 
Assume $k \ge 5$.
If $\Delta$ is an affected $G$-orbit, \ie,
$\Delta\cap U=\emptyset$, then $\ker(\vf)$ is not
transitive on $\Delta$; in fact, each orbit of
$\ker(\vf)$ in $\Delta$ has length $\le |\Delta|/k$.
\end{enumerate}
\end{theorem}     
This result is a combination of Theorem~\ref{thm:unaffected}
and Corollary~\ref{cor:affected} proved in 
Section~\ref{sec:altquotient}.
\begin{remark}
We note that part~\eqref{item:unaffected}
becomes false if we relax the condition
$k > 2+\log_2 n_0$ to $k\ge 2+\log_2 n_0$.  In Remark~\ref{rem:zeroweight}
we exhibit infinitely many transitive groups with giant 
actions with $k=2+\log_2 n$ where
none of the elements is affected (and the kernel is transitive).
\end{remark}

The affected/unaffected dichotomy is our 
\emph{principal divide-and-conquer} tool\footnote{The
discovery of this tool was the turning point of this project, 
cf. footnote~\ref{fn:eureka} in
Sec.~\ref{subsec:localcertificates}}.

These results are employed in {\sf Procedure LocalCertificates},
the heart of the entire algorithm,
in Section~\ref{sec:localcertificates}.
It is Theorem~\ref{thm:unaffected1} that allows us to build up local
symmetry to global automorphism unless an explicit 
obstruction is found.

\subsubsection{The group-theoretic divide-and-conquer algorithm}
\label{sec:local-certificates-explained}

We sketch the {\sf LocalCertificates} procedure, to be formally described
in Section~\ref{sec:localcertificates}.  The procedure takes a
logarithmic size ``test set'' $A\subseteq \Gamma$, $|A|=k$, and 
tries to build automorphisms corresponding to arbitrary (even)
permutations of $A$.  This is an iterative process: first we ignore
the input string $\xx$, so we have the group $G_A$ (setwise stabilizer
of $A$).  Then we begin ``growing a beard'' which at first consists of
those elements of $\Omega$ that are affected by $\vf$.  Now we take
the segment of $\xx$ that falls in the ``beard'' into account, so the
automorphism group shrinks, more points will be affected; we include
them in the next layer of the growing beard, etc.  

We stop when either the action of the current automorphism group (of
the segment of the input that belongs to the beard) on $A$ is no
longer giant, or the affected set (the beard) stops growing.  

In the former case we found a canonical $k$-ary relation on $A$; after
aggregating these over all test sets, we hand over the process to
the combinatorial partitioning techniques of the next section 
(Design Lemma, Split-or-Johnson algorithm).

In the latter case we pointwise stabilize all non-affected points; we
still have a giant action on $A$ and this time the reduced group
consists of automorphisms of the entire string $\xx$ (since it does
not matter what letter of the alphabet is assigned to fixed points of
the group).  Then we analyze the action on $\Gamma$ of the
group of automorphisms obtained from the
positive local certificates (Sec.~\ref{sec:aggregate}).

The procedure described is to be used when $G$ is transitive and
imprimitive.  If $G$ is intransitive, we apply Luks's Chain Rule
(orbit-by-orbit processing, Prop.~\ref{prop:chainrule}).  
If $G$ is primitive, we may assume $G$
is a Johnson group $\sss_m^{(t)}$ or $\aaa_m^{(t)}$, so the situation
is that of an edge-colored $t$-uniform hypergraph on $m$ vertices.  We
use the Extended Design Lemma (Theorem~\ref{thm:extended-design})
to either canonically partition $k$ or to reduce $\aaa_m$
to a much smaller Johnson group acting on $\le m$ points,
see Sec.~\ref{sec:master}.

\subsubsection{Combinatorial partitioning; discovery of a 
  canonically embedded large Johnson graph}   \label{sec:partitioning}

The partitioning algorithms take as input a set $\Omega$ related in some
way to a structure $X$.  The goal is either to establish high symmetry of
$X$ or to find a canonical structure on $\Omega$ that represents
an explicit obstruction to such high symmetry.

Significant partitioning is expected at modest ``multiplicative
cost'' (explained below).
Favorable outcomes of the partitioning algorithms are (a) a canonical
coloring of $\Omega$ where each color-class has size $\le 0.9n$
($n=|\Omega|$), or (b) a canonical equipartition of a canonical subset
of $\Omega$ of size $\ge 0.9n$.

A Johnson graph $J(v,t)$ has $n=\binom{v}{t}$ vertices labeled by the
$t$-subsets $T\subseteq [v]$.  The $t$-subsets $T_1,T_2$ are adjacent
if $|T_1\setminus T_2|=1$.  Johnson graphs do not admit a
coloring/partition as described, even at quasipolynomial
multiplicative cost, if $t$ is subpolynomial in $v$ (\ie,
$t=v^{o(1)}$).  (Johnson graphs with $t=2$ have been the most
notorious obstacles to breaking the $\exp(\wto(\sqrt{n}))$ bound on GI.)
One of the main results of the paper is that in a
well--defined sense, Johnson graphs are the \emph{only} obstructions
to effective partitioning:
either partitioning succeeds as desired or a canonically embedded
Johnson graph on a subset of size $\ge 0.9n$ is found.
Here is a corollary to the result.

\begin{theorem}
Let $X=(V,E)$ be a nontrivial regular graph (neither complete, nor empty)
with $n$ vertices.  At a quasipolynomial multiplicative cost we can
find one of the following structures.  We call the structure
found $Y$.
\begin{enumerate}[(a)]
\item A coloring of $V$ with no color-class larger than $0.9n$;
\item A coloring of $V$ with a color-class $C$ of size $\ge 0.9n$ and
  a nontrivial equipartition of $C$ (the blocks of the partition are of
  equal size $\ge 2$ and there are at least two blocks);
\item A coloring of $V$ with a color-class $C$ of size $\ge 0.9n$ and 
a Johnson graph $J(v,t)$ $(t\ge 2)$ with vertex-set $C$,
\end{enumerate}
such that the index of the subgroup
$\aut(X)\cap\aut(Y)$ in $\aut(X)$ is quasipolynomially bounded.
\end{theorem} 
The index in question (and its natural extension to isomorphisms)
represents the multiplicative cost incurred.
The full statement can be found in Theorem~\ref{thm:UPCC}.


The same is true if $X$ is a $k$-ary relational structure that
does not admit the action of a symmetric group of degree $\ge 0.9n$
on its vertex set (has ``symmetry defect'' $\ge 0.1n$,
see Def.~\ref{def:defect}) assuming $k$
is polylogarithmically bounded.  The reduction from $k$-ary
relations ($k\ge 3$) to regular graphs (and to highly regular
binary relational structures called ``uniprimitive coherent
configurations'' or UPCCs) is the content of the Design Lemma
(Theorem~\ref{thm:design}).

Note that the Johnson graph will not be a subgraph of $X$; but it will
be ``canonically embedded'' relative to an arbitrary choice from a
quasipolynomial number of possibilities, with the consequence of not
reducing the number of automorphisms/isomorphisms by more than a
quasipolynomial factor.

The number $0.9$ is arbitrary; the result would remain valid
for any constant $0.5<\alpha <1$ in place of $0.9$.

We note that the \emph{existence} of such a structure $Y$ can be deduced
from the Classification of Finite Simple Groups.  We not only assert
the existence but also find such a structure in quasipolynomial time, and
the analysis is almost entirely combinatorial, with a modest use of
elementary group theory.

The structure $Y$ is ``canonical relative to an arbitrary choice'' 
from a quasipolynomial number of possibilities.  These arise
by individualizing a polylogarithmic number of 
``ideal points'' of $Y$.  An ``ideal point'' of $X$ is a 
point of a structure $X'$ canonically constructed from $X$,
much like ``ideal points'' of an affine plane are the ``points at infinity.''
Individualizing a point at infinity means individualizing a
parallel class of lines in the affine plane.

Canonicity means being preserved under isomorphisms in a category of
interest.  This category is always very small, it 
often has just two objects (the two graphs or
strings of which we wish to decide isomorphism); sometimes it has a
quasipolynomial number of objects (when checking local symmetry, we
need to compare every pair of polylogarithmic size subsets of the
domain).  In any case, this notion of canonicity does not require
canonical forms for the class of all graphs or strings, a problem we
do not address in this paper.  We say that we incur a ``multiplicative
cost'' $\tau$ if a choice is made from $\tau$ possibilities.  This
indeed makes the algorithm branch $\tau$ ways, giving rise to
a factor of $\tau$ in the recurrence.

Canonicity and ``relative canonicity at a multiplicative cost''
are formalized in the language of functors in Section~\ref{sec:functor}.

\section{Preliminaries} 

\subsection{Fraktur}
We list the Roman equivalents of the letters in Fraktur we use:

\noindent
$\xx$ -- x, \ $\yy$ -- y, \ $\zz$ -- z, \\
$\aaa$ -- A, \ $\bbb$ -- B, \ $\ggg$ -- G, \
$\hhh$ -- H, \ $\jjj$ -- J, $\lll$ -- L, $\ppp$ -- P, \
$\sss$ -- S, \ $\xxx$ -- X, \ $\yyy$ -- Y, \ $\zzz$ -- Z

\subsection{Permutation groups}
All groups in this paper are finite.
Our principal reference for permutation groups is the monograph
by Dixon and Mortimer~\cite{dixon}.  Wielandt's classic~\cite{wielandtbook}
is a sweet introduction.  Cameron's article~\cite{cameron}
is very informative.
For the basics of permutation group algorithms we refer the reader to
Seress's monograph~\cite{seress-book}.  
Even though we summarize Luks's method in our language in
Sec.~\ref{sec:luks}, Luks's seminal paper~\cite{luks-bded} 
is a prerequisite for this one.

For a set $\Omega$ we write $\sym(\Omega)$ for the symmetric group
consisting of all permutations of $\Omega$ and $\alt(\Omega)$ for
the alternating group on $\Omega$ (set of even permutations of
$\Omega)$.  
We write $\sss_n$ for $\sym([n])$ and $\aaa_n$ for $\alt([n])$
where $[n]=\{1,\dots,n\}$.  We also use the symbols $\sss_n$ and
$\aaa_n$ when the permutation domain is not specified
(only its size).  For a function $f$ we usually write $x^f$ for
$f(x)$.  In particular, for $\sigma\in\sym(\Omega)$ and
$x\in\Omega$ we denote the image of $x$ under $\sigma$ by
$x^{\sigma}$.  For $x\in\Omega$, $\sigma\in\sym(\Omega)$,
$\Delta\subseteq\Omega$, and $H\subseteq\sym(\Omega)$ 
we write 
\begin{equation}   \label{eq:action}
     x^H =\{x^{\sigma}\mid \sigma\in H\} \text{\ and \ }
     \Delta^{\sigma} =\{y^{\sigma}\mid y\in\Delta\} \text{\ and \ }
     \Delta^{H} =\{\Delta^{\sigma}\mid \sigma\in H\}.
\end{equation}

For groups $G,H$ we write $H\le G$ to indicate that $H$ is a
subgroup of $G$.  The expression $|G : H|$ denotes the \emph{index}
of $H$ in $G$.  Subgroups $G\le \sym(\Omega)$ are the
\emph{permutation groups} on the domain $\Omega$.  The size of the
permutation domain, $|\Omega|$, is called the \emph{degree} of $G$
while $|G|$ is the \emph{order} of $G$.  We refer to $\sym(\Omega)$
and $\alt(\Omega)$, the two largest permutation groups on $\Omega$,
as the \emph{giants}.

By a \emph{representation} of a group $G$ we shall always mean a
\emph{permutation representation}, \ie, a homomorphism $\vf:G\to
\sym(\Omega)$.  We also say in this case that $G$ \emph{acts on
  $\Omega$} (via $\vf$).  We say that $\Omega$ is the \emph{domain}
of the representation and $|\Omega|$ is the \emph{degree} of the
representation.  If $\vf$ is evident from the context, we write
$x^{\pi}$ for $x^{\pi^{\,\vf}}$.  For $x\in\Omega$,
$\sigma\in G$, $\Delta\subseteq\Omega$, and $H\subseteq G$, we
define $x^{H}$ and $\Delta^{\sigma}$ and $\Delta^H$
by Eq.~\eqref{eq:action}.

We denote the image of $G$ under $\vf$ by $G^{\,\vf}$, so
$G^{\,\vf}\cong G/\ker(\vf)$.  If $G^{\,\vf}\ge\alt(\Omega)$ we say
$\vf$ is a \emph{giant action} and $G$ acts on $\Omega$ ``as a
giant.''

A subset $\Delta\subseteq\Omega$ is \emph{$G$-invariant}
if $\Delta^G=\Delta$.  
\begin{notation}  \label{not:restriction}
If $\Delta\subseteq\Omega$ is $G$-invariant then
$G^{\Delta}$ denotes the image of the representation
$G\to\sym(\Delta)$ defined by restriction to $\Delta$.
So $G^{\Delta}\le\sym(\Delta)$.
\end{notation}

The \emph{stabilizer} of $x\in \Omega$ is the subgroup
$G_x=\{\sigma\in G\mid x^{\sigma}=x\}$.  The \emph{orbit}
of $x\in\Omega$ is the set $x^G=\{x^{\sigma}\mid\sigma\in G\}$.
The orbits partition $\Omega$.  A simple bijection shows that
\begin{equation}
   |x^G| = |G:G_x| .
\end{equation}

For $T\subseteq\Omega$ and $G\le\sym(\Omega)$ we write
$G_T$ for the \emph{setwise stabilizer} of $T$ and
$G_{(T)}$ for the \emph{pointwise stabilizer} of $T$, \ie,
\begin{equation}
  G_T = \{\alpha\in G \mid T^{\alpha} = T\}
\end{equation}
and
\begin{equation}
  G_{(T)} = \{\alpha\in G \mid (\forall x\in T)(x^{\alpha} = x)\}.
\end{equation}

So $G_{(T)}$ is the kernel of the $G_T\to\sym(T)$ homomorphism
obtained by restriction to $T$; in particular, $G_{(T)}\normal G_T$.

For $t\ge 0$ we write $\binom{\Omega}{t}$ to denote the set of
$t$-subsets of $\Omega$.  So if $|\Omega|=k$ then
$\left|\binom{\Omega}{t}\right|=\binom{k}{t}$.  A permutation group
$G\le\sym(\Omega)$ naturally acts on $\binom{\Omega}{t}$; we refer to
this as the \emph{induced action on $t$-sets} and denote the resulting
subgroup of $\sym\binom{\Omega}{t}$ by $G^{(t)}$.  This in particular
defines the notation $\sss_k^{(t)}$ and $\aaa_k^{(t)}$;
these are subgroups of $\sss_{\binom{k}{t}}$.   We refer to
$\sss_k^{(t)}$ and $\aaa_k^{(t)}$ as \emph{Johnson groups}
since they act on the ``Johnson schemes''
(see below)\footnote{``Johnson schemes'' is a standard term; we
introduce the term ``Johnson groups'' for convenience.}.

The group $G$ is \emph{transitive} if it has only one orbit,
\ie, $x^G=\Omega$ for some (and therefore any) $x\in\Omega$.
The $G$-invariant sets are the unions of orbits.

A \emph{$G$-invariant partition} of $\Omega$ is a partition
$\{B_1,\dots,B_m\}$ where the $B_i$ are nonempty, pairwise
disjoint subsets of which the union is $\Omega$ such that $G$
permutes these subsets, \ie, 
$(\forall \sigma\in G)(\forall i)(\exists j)(B_i^{\sigma}=B_j)$.
The $B_i$ are the \emph{blocks} of this partition.

A nonempty subset $B\subseteq \Omega$
is a \emph{block of imprimitivity} for $G$ if
$(\forall g\in G)(B^g=B$ or $B^g\cap B=\emptyset)$.
A subset $B\subseteq \Omega$ is a block of imprimitivity
if and only if it is a block in an invariant partition.

A \emph{system of imprimitivity} for $G$ is
a $G$-invariant partition $\calB=\{B_1,\dots,B_m\}$
of a $G$-invariant subset $\Delta\subseteq \Omega$
such that $G$ acts transitively on $\calB$. 
(So $\Delta = \dot\bigcup_i B_i$; we assume here that
$(\forall i)(B_i\neq\emptyset)$).   The $B_i$ are
then blocks of imprimitivity, and every system of
imprimitivity arises as the set of $G$-images of a block
of imprimitivity.  The group $G$ acts on $\calB$ by permuting
the blocks; this defines a representation $G\to\sss_m$.

A \emph{maximal system of imprimitivity} for $G$ is a system of
imprimitivity of blocks of size $\ge 2$ that cannot be refined,
\ie, where the blocks are minimal (do not properly contain
any block of imprimitivity of size $\ge 2$).

$G\le \sym(\Omega)$ is \emph{primitive} if $|G|\ge 2$ and
$G$ has no blocks of imprimitivity other than $\Omega$ and 
the singletons (sets of one element).
In particular, a primitive group is transitive.
Examples of primitive groups include the cyclic group
of prime order $p$ acting naturally on a set of $p$ elements,
and the Johnson groups $\sss_k^{(t)}$ and $\aaa_k^{(t)}$ 
for $t\ge 1$ and $k\ge 2t+1$.

A group $G\le\sym(\Omega)$ is \emph{doubly transitive} 
if its induced action on the set of $n(n-1)$ ordered
pairs is transitive (where $n=|\Omega|$).  

\begin{definition}   \label{def:mindeg}
The \emph{support} of a permutation $\sigma\in\sym(\Omega)$
is the set of elements that $\sigma$ moves: 
$\supp(\sigma)=\{x\in\Omega\mid x^{\sigma}\neq x\}$.  
The \emph{degree} of $\sigma$ is the size of its support.
The \emph{minimal degree} of a permutation group $G$ is
$\min_{\sigma\in G, \sigma\neq 1} |\supp(\sigma)|$.
\end{definition}

The following 19th-century gem will be used
in the proof of Claim 2b1 in Sec.~\ref{sec:aggregate}.
It appears in~\cite{bochert97}.  
We quote if from~\cite[Thm. 5.4A]{dixon}.

\begin{theorem}[Bochert, 1897]   \label{thm:bochert}
If $G$ is a doubly transitive group of degree $n$ other than $\aaa_n$ or
$\sss_n$ then its minimal degree is at least $n/8$.  If $n\ge 217$
then the minimal degree is at least $n/4$.
\end{theorem}

\subsection{Relational structures, $k$-ary coherent configurations}
\label{sec:kary}

A \emph{$k$-ary relation} on the set $\Omega$ is a subset
$R\subseteq \Omega^{\,k}$.  A \emph{relational structure}
$\xxx=(\Omega;\calR)$ consists of $\Omega$, the set of
\emph{vertices}, and $\calR=(R_1,\dots,R_r)$, a list of
relations on $\Omega$.  We write $\Omega=V(\xxx)$.
We say that $\xxx$ is a \emph{$k$-ary relational structure}
if each $R_i$ is $k$-ary.  Let $\xxx'=(\Omega';\calR')$
where $\calR'=(R_1',\dots,R_r')$.   A bijection
$f : \Omega\to\Omega'$ is an $\xxx\to\xxx'$ \emph{isomorphism}
if $(\forall i)(R_i^f=R_i')$, \ie, for $x_i\in\Omega$ we have
$(x_1,\dots,x_k)\in R_i \iff (x_1^f,\dots,x_k^f)\in R_i')$.
We denote the set of $\xxx\to\xxx'$ isomorphisms by
$\iso(\xxx,\xxx')$ and write $\aut(\xxx)=\iso(\xxx,\xxx)$
for the automorphism group of $\xxx$.

\begin{definition}[Induced substructure] \label{def:induced}
Let $\Delta\subseteq\Omega$ and let $\xxx=(\Omega;R_1,\dots,R_r)$
be a $k$-ary relational structure.  Let 
$R_i^{\Delta}=R_i\cap\Delta^k$.  We define the \emph{induced
substructure} $\xxx[\Delta]$ of $\xxx$ on $\Delta$ as
$\xxx[\Delta]=(\Delta;R_1^{\Delta},\dots,R_r^{\Delta})$.
\end{definition}

\begin{definition}[$t$-skeleton]  \label{def:skeleton}
For $R\subseteq \Omega^{\,k}$ and $t\le k$ let
$R^{(t)}= \{(x_1,\dots,x_t)\mid (x_1,\dots,x_t,x_t,\dots,x_t)\in R\}$.
We define the \emph{$t$-skeleton} $\xxx^{(t)}=(\Omega;\calR^{(t)})$
of the $k$-ary relational
structure $\xxx=(\Omega;\calR)=(\Omega;R_1,\dots,R_r)$ by setting
$\calR^{(t)}=(R_1^{(t)},\dots,R_r^{(t)})$.
\end{definition}

The group $\sss_k$ acts naturally on $\Omega^{\,k}$ by permuting the
coordinates.

\begin{notation}[Substitution]
For ${\vec x}=(x_1,\dots,x_k)\in \Omega^{\,k}$ and $y\in\Omega$
let ${\vec x}_i^{\,y}=(x_1',\dots,x_k')$ where $x_j'=x_j$ for
all $j\neq i$ and $x_i'=y$.
\end{notation}

We shall especially be interested in the case when 
the $R_i$ partition $\Omega^{\,k}$.  This is equivalent
to coloring $\Omega^{\,k}$; if ${\vec x}=(x_1,\dots,x_k)\in R_i$ 
then we call $i$ the \emph{color} of the $k$-tuple $\vec x$ 
and write $c(\vec x)=i$.  

\begin{definition}[Configuration]   
We say that the $k$-ary relational structure $\xxx$ 
is a \emph{$k$-ary configuration} if the following hold:
\begin{itemize}
\item[(i)] the $R_i$ partition $\Omega^{\,k}$ and all the $R_i$ are nonempty;
\item[(ii)] if $c(x_1,\dots,x_k)=c(x_1',\dots,x_k')$ then 
       $(\forall i,j\le k)(x_i=x_j \iff x_i'=x_j')$;
 \item[(iii)] $(\forall \pi\in\sss_k)(\forall i\le k)(\exists
         j\le k)(R_i^{\pi}=R_j)$.
\end{itemize}
Here $R^{\pi}$ denotes the relation 
$R^{\pi}=\{(x_{1^{\pi}},\dots,x_{k^{\pi}})\mid (x_1,\dots,x_k)\in R\}$.
\end{definition}
We call $r$ the \emph{rank} of the configuration.
We note that the $t$-skeleton of a configuration 
of rank $r$ is a configuration of rank $\le r$
(we keep only one copy of identical relations).

Vertex colors are the colors
of the diagonal elements: $c(x)=c(x,\dots,x)$. 
We say that the configuration $\xxx$ is \emph{homogeneous}
if all vertices have the same color.  We note that
the $s$-skeleton of a $k$-ary homogeneous 
configuration is homogeneous.

\begin{definition}[$k$-ary coherent configurations]
We call a $k$-ary configuration $\xxx=(\Omega;R_1,\dots,R_r)$
\emph{coherent} if, in addition to items (i)--(iii), the following holds:
\begin{enumerate}
\item[(iv)]
    There exists a family of $r^{k+1}$
    nonnegative integer \emph{structure constants}\quad $p(i_0,\dots,i_k)$
    \quad ($1\le i_0,\dots,i_k\le r$)
    such that for all ${\vec x}\in R_{i_0}$
    we have
\begin{equation}
       |\{y\in\Omega \mid 
         (\forall j\le k)(c({\vec x}_j^{\,y})=i_j)\}|=p(i_0,\dots,i_k).
\end{equation}
\end{enumerate}
\end{definition}
These are the stable configurations under the \emph{$k$-dimensional
Weisfeiler-Leman canonical refinement process} (Sec.~\ref{sec:WL}).

\begin{observation}
For all $t\le k$, the $t$-skeleton of a $k$-ary
coherent configuration is an $t$-ary coherent
configuration.
\end{observation} 

\subsection{Twins, symmetry defect}  \label{sec:twins}

\begin{convention}
Let $\Psi\subseteq \Omega$.  We view $\sym(\Psi)$ as a subgroup
of $\sym(\Omega)$ by extending each element of $\sym(\Psi)$
to act trivially on $\Omega\setminus\Psi$.
\end{convention}

\begin{definition}[Twins]   \label{def:twins}
Let $G\le\sym(\Omega)$ and $x,y\in\Omega$.  
We say that the points $x\neq y$ are \emph{strong twins} if 
the transposition $\tau=(x,y)$ belongs to $G$.
We say that the points $x\neq y$ are \emph{weak twins} if
either they are strong twins or
there exists $z\notin\{x,y\}$ such that the
3-cycle $\sigma=(x,y,z)$ belongs to $G$.  
\end{definition}

\begin{observation}  
Both the ``strong twin or equal'' 
and the ``weak twin or equal'' relations are
equivalence relations on $\Omega$.   
\end{observation}
\begin{proof}
We need to show transitivity.  This is obvious for the strong twin relation.
To see the transitivity of the weak twin relation, let $x,y,z$ be
distinct points, and assume the 3-cycles $\rho=(x,y,u)$ and $\sigma=(y,z,w)$
belong to $G$ for some $u,v$ where $u\notin \{x,y\}$ and $v\notin\{y,z\}$.
Let $S=\{x,y,z,u,v\}$ (so $3\le |S|\le 5$) and let 
$H=\langle \rho,\sigma\rangle \le G$.  So $H$ is generated by a connected
set of 3-cycles and therefore $H=\alt(S)$.  Consequently, the 3-cycle
$(x,z,y)$ belongs to $H$, so $x$ and $z$ are weak twins.
\end{proof}
\begin{definition}
We call the nontrivial (non-singleton) equivalence classes of these
relations the \emph{strong/weak-twin equivalence classes}
of $G$, respectively.
\end{definition}
\begin{definition}[Symmetrical sets]   \label{def:symmetrical}
Let $G\le \sym(\Omega)$ where $|\Omega|=n$.  Let $\Psi\subseteq
\Omega$.  We say that $\Psi$ is a \emph{strongly/weakly symmetrical
  set} for $\Omega$ if $|\Psi|\ge 2$ and all pairs of points
in $\Psi$ are strong/weak twins, resp.
\end{definition}
\begin{observation}  \label{obs:weak-alt}
$\Psi\subseteq\Omega$ is a strongly symmetrical set exactly if
$\sym(\Psi)\le G$.  $\Psi\subseteq\Omega$ is a weakly symmetrical 
set exactly if either $|\Psi|\ge 3$ and $\alt(\Psi)\le G$, or
$|\Psi|=2$ and $\sym(\Psi)\le G$, or  $|\Psi|=2$ and there exists
a proper superset $\Psi'\supset\Psi$ such that $\alt(\Psi')\le G$.
\end{observation}


\begin{definition} \label{def:defect}
Let $T\subseteq\Omega$ be a smallest subset of $\Omega$
such that $\Omega\setminus T$ is a weakly symmetrical set.
We call $|T|$ the 
\emph{symmetry defect} of $G$ and the quotient
$|T|/n$ the \emph{relative symmetry defect} of $G$,
where $n=|\Omega|$.
\end{definition}

For example, let $\Omega=\Omega_1\dot\cup\Omega_2$ where
$|\Omega_i|\ge 3$, and let $G=\alt(\Omega_1)\times\alt(\Omega_2)$.
Then the symmetry defect of $G$ is $\min(|\Omega_1|,|\Omega_2|)$.

\begin{definition}   \label{def:rel-defect}
Let $\xxx=(\Omega;\calR)$ be a relational structure.  We say that 
$x,y\in\Omega$ are strong twins for $\xxx$ if they are strong
twins for $\aut(\xxx)$.  We analogously transfer all concepts
introduced in this section from groups to structures
via their automorphism groups (weak twins, strongly/weakly
symmetrical sets, symmetry defect). For instance, the 
(relative) symmetry defect of $\xxx$ is the (relative)
symmetry defect of $\aut(\xxx)$. 
\end{definition}

\begin{observation}
Given an explicit relational structure $\xxx$ with vertex set $\Omega$,
one can find
the maximal weakly symmetrical subsets of $\Omega$ in polynomial time.
Consequently, the (relative) symmetry defect of $\xxx$ can also be
determined in polynomial time.
\end{observation}

\begin{proof}
Test, for each transposition and 3-cycle $\sigma\in\sym(\Omega)$, 
whether or not $\sigma\in\aut(\xxx)$.  Join two elements $x,y$
of $\Omega$ by an edge if the transposition $(x,y)\in\aut(\xxx)$ 
or there exists $z\in\Omega$ such that the 3-cyle
$(x,y,z)\in\aut(\xxx)$.  The connected components of this
graph are the maximal weakly symmetrical sets.
\end{proof}

\begin{proposition}   \label{prop:alt-sym}
Let $\xxx$ be a $k$-ary relational structure on the vertex set
$\Omega$ and $\Psi\subset\Omega$ such that $|\Psi| \ge k+2$.
If $\Psi$ is weakly symmetrical then $\Psi$ is strongly
symmetrical.  In other words, any $k+2$ vertices that
are pairwise weak twins are pairwise strong twins.
\end{proposition}

\begin{proof}
We need to show that if $\alt(\Psi)\le \aut(\xxx)$ then
$\sym(\Psi)\le \aut(\xxx)$.  Let $x,y\in \Psi$, $x\neq y$, and let
$\tau=(x,y)$ be the corresponding transposition.  Suppose for a
contradiction that $\tau\notin\aut(\xxx)$ and let $i\le r$ and
${\vec x}=(x_1,\dots, x_k)\in R_i$ be a witness of this, \ie,
${\vec x}^{\,\tau}\notin R_i$.  Let 
$u,v\in \Psi\setminus\{x_1,\dots,x_k\}$ where $u\neq v$,
and let $\sigma=(x,y)(u,v)$ (product of two transpositions).  So
$\sigma\in\alt(\Psi)$ and therefore $\sigma\in\aut(\xxx)$.  But
${\vec x}^{\,\sigma}={\vec x}^{\,\tau}\notin R_i$, a contradiction.
\end{proof}

\begin{definition}
A \emph{digraph} is a pair $X=(V,E)$ where $E \subseteq V\times V$ 
is a
binary relation on $V$.
The \emph{out-degree} of vertex $u\in V$ is the number of $v\in V$
such that $(u,v)\in E$.  In-degree is defined analogously.
$X$ is \emph{biregular} if all vertices have the same in-degree
and the same out-degree (so these two numbers are also equal).
The \emph{diagonal} of $V$ is the set $\diag(V)=\{(x,x)\mid x\in V\}$.
$X$ is \emph{irreflexive} if $E\cap\diag(V)=\emptyset$.
The \emph{irreflexive complement} of an irreflexive digraph
$X=(V,E)$ is $X'=(V,E')$ where $E'=V\times V\setminus (\diag(V)\cup E)$.
$X$ is \emph{trivial} if $\aut(X)=\sym(V)$, \ie,
$E$ is one of the following: the empty set, $V\times V$,
$\diag(V)$, or $V\times V\setminus \diag(V)$.
A subset of $A\subseteq V$ is \emph{independent} if
it spans no edges, \ie, $E\cap A\times A=\emptyset$.
Note that an independent set cannot contain a self-loop, \ie,
a vertex $x$ such that $(x,x)\in E$.
\end{definition}

The following observation is well known.  It will
be used directly in Case~3a2 in Section~\ref{sec:blockdesign}
and indirectly through Cor.~\ref{cor:reg-defect} below.
\begin{proposition}   \label{prop:digraph-indep}
Let $X=(V,E)$ be a non-empty biregular digraph.
Then $X$ has no independent set of size greater than
$n/2$ where $n=|V|$.
\end{proposition}
\begin{proof} 
Let $d>0$ be the out-degree (and therefore the in-degree)
of each vertex.
Let $A\subseteq V$ be independent.
Then $V\setminus A$
has to absorb all edges emanating from $A$, so 
$d(n-|A|)\ge d|A|$.
\end{proof}

The following corollary will be used in item 2b2 of the
algorithm described in Section~\ref{sec:aggregate}. 

\begin{corollary}   \label{cor:reg-defect} 
Let $X=(V,E)$ be a nontrivial irreflexive biregular digraph
with $n\ge 4$ vertices.
Then the relative symmetry defect of $X$ is $\ge 1/2$.
\end{corollary}
\begin{proof}
Let $A\subseteq V$ be a (weakly) symmetrical set
with $\ge 3$ vertices.  So
$\aut(X)\ge\alt(A)$.   Then $A$ is either an independent
set in $X$, or independent in the irreflexive complement of
$X$.  In both cases, Prop.~\ref{prop:digraph-indep}
guarantees that $|A|\le n/2$.
\end{proof}

\subsection{Classical coherent configurations}
\label{sec:classical} 
\subsubsection{Definition, constituent digraphs, the clique configuration}
\label{sec:CCdef}
Continuing our discussion of $k$-ary coherent configurations
(Sec.~\ref{sec:kary}), we now turn to the classical case, $k=2$.
A \emph{(classical) coherent configuration} is a binary (2-ary) coherent
configuration.  If we don't specify arity, we mean the classical case and
usually omit the adjective ``classical.''  Coherent configurations
are the stable configurations of the classical Weisfeiler-Leman
canonical refinement process~\cite{weisfeiler-leman,weisfeiler-book}
see Sec.~\ref{sec:WL}.

Let us review the definition, starting with binary configurations.

Recall that
$\diag(\Omega)=\{(x.x)\mid x\in\Omega\}$ denotes the diagonal of the set 
$\Omega$.  For a relation $R\subseteq\Omega\times\Omega$ let
$R^-=\{(y,x)\mid (x,y)\in R\}$.

We call a binary relational structure $\xxx=(\Omega;R_1,\dots,R_r)$
$(R_i\subseteq\Omega\times\Omega)$ a \emph{binary configuration} if 
\begin{itemize}
  \item[(i)] the $R_i$ are nonempty and partition $\Omega\times\Omega$
  \item[(ii)] $(\forall i)(R_i\subseteq \diag(\Omega)$ or
              $R_i\cap \diag(\Omega)=\emptyset$
  \item[(iii)] $(\forall i)(\exists j)(R_i^-=R_j)$
\end{itemize}
(This is the binary case of the $k$-ary configurations defined
in Section~\ref{sec:kary}.)  We write $j=i^-$ if $R_i^-=R_j)$.

The \emph{rank} of $\xxx$ is $r$; so 
if $|\Omega|\ge 2$ then $r\ge 2$.
If $(x,y)\in R_i$ we say that
the \emph{color} of the pair $(x,y)$ is $c(x,y)=i$. 
We designate $c(x,x)$ to be the color of the vertex $x$.

\begin{definition}[Constituents, in- and out-vertices]
We call the digraph $X_i=(\Omega,R_i)$ the \emph{color-$i$ constituent
digraph} of $\xxx$.   We say that $x\in\Omega$ is an \emph{effective
vertex} of $X_i$ if the degree (in-degree plus out-degree) of $x$ in $X_i$
is not zero.  We say that $x$ is an \emph{out-vertex} of $X_i$ if
its out-degree in $X_i$ is positive, and an \emph{in-vertex} of $X_i$
if its in-degree in $X_i$ is positive.
\end{definition}

In accordance with Sec.~\ref{sec:kary},
we say that the configuration $\xxx$ is \emph{coherent} if
    there exists a family of $r^3$
    nonnegative integer \emph{structure constants} $p_{ij}^k$
    ($1\le i,j,k\le r$) such that 
\begin{itemize}
\item[(iv)]  $(\forall i,j,k\le r)(\forall 
        (x,y)\in R_k)(|\{z\mid (x,z)\in R_i\text{\ and\ }
        (z,y)\in R_j\}|=p_{ij}^k)$\,.
\end{itemize}

For the rest of this section, let $\xxx=(\Omega;R_1,\dots,R_r)$
be a coherent configuration. Let $n=|\Omega|$.

Next we show that for $x,y\in\Omega$, the color $c(x,y)$ determines the colors
$c(x)$ and $c(y)$.
\begin{observation}[Vertex-color awareness] \label{obs:vx-color}
All out-vertices of the constituent $X_i$ have the same color.
Analogously, all in-vertices of the constituent $X_i$ have the same color.
\end{observation}
\begin{proof}
Assume $c(x,y)=c(x',y')=i$.  We need to show that $c(x)=c(x')$
(and analogously, $c(y)=c(y')$).  Let $c(x,x)=\ell$.  
$p(\ell,i,i)=1$.  It follows that $(\exists z')(c(x',z')=\ell$
and $c(z',y')=i$.  But $\ell$ is a diagonal color, so $z'=x'$
and therefore $c(x')=c(x',x')=\ell$.  The proof of $c(y)=c(y')$
works analogously.
\end{proof}

Next we show that the color of a vertex $x$ determines its
in- and out-degree in every color.
\begin{observation}[Degree awareness]
\label{obs:deg-aware}
The out-degree of all out-vertices of the constituent $X_i$ is
equal; we denote this quantity by $\deg^+(i)$.  Analogously,
the in-degree of all in-vertices of $X_i$ is equal;
we denote this quantity by $\deg^-(i)$.
\end{observation}
\begin{proof}
Let $x$ be an out-vertex of $X_i$; let $c(x)=\ell$.  By 
Obs.~\ref{obs:vx-color}, $\ell$ does not depend on the choice of $x$.
Now the out-degree of $x$ in $X_i$ is $p(\ell,i,i^-)$.
The proof for in-degrees is analogous.
\end{proof}

A \emph{graph} $X=(V,E)$ can be viewed as a configuration
$\xxx(X)=(V;\diag(V),E,\Ebar)$ where the edge set $E$ is viewed as an 
irreflexive,
symmetric relation and $\Ebar=V\times V\setminus(\diag(V)\cup E)$
is the set of edges of the complement of $X$.

The configuration $\xxx(X)$ has rank 3 unless $X$ is the empty
or the complete graphs (empty relations are omitted), in which
case it has rank~2.

The configuration $\xxx(X)$ is coherent if and only if 
$X$ is a \emph{strongly regular graph}.  

\begin{definition}
The \emph{clique configuration} $\xxx(K_n)$ has $n$ vertices and
rank 2: the constituents are the
diagonal $\diag(\Omega)=\{(x,x)\mid x\in\Omega\}$, 
and the rest: $\Omega\times\Omega\setminus\diag(\Omega)$.
We also refer to the clique configuration as the 
\emph{trivial configuration}.
\end{definition}
Alternatively, the clique configuration can be defined as the
unique (up to naming the relations) configuration of which
$\sss_n$ is the automorphism group.

Since diagonal and off-diagonal colors must be distinct, the rank
is $r\ge 2$ (assuming $n\ge 2$).  The only rank-2 configuration is
the clique configuration.

\subsubsection{Connected components of constituents}
Let $C_1,\dots, C_s$ be the vertex-color classes of the
coherent configuration $\xxx$; they 
partition $\Omega$.  The following observation will be useful.

\begin{proposition}[Induced subconfiguration]
Let $\Delta\subseteq\Omega$ be the union of some of the
vertex-color classes and let $\xxx^{\Delta}$ denote the 
subconfiguration induced in $\Delta$.  Then $\xxx^{\Delta}$
is coherent.
\end{proposition}

It follows from Obs.~\ref{obs:vx-color} that 
either all effective vertices of $X_i$ have the same color 
(``$X_i$ is homogeneous'') or they belong to two color classes,
say $C_j$ and $C_{\ell}$, $j\neq\ell$, and $R_i\subseteq C_j\times C_{\ell}$
(``$X_i$ is bipartite'').

\begin{definition}[Equipartition]
An \emph{equipartition} of a set $\Omega$ is a partition of $\Omega$
into blocks of equal size.
\end{definition}
\begin{proposition}[Bipartite connected components] 
\label{prop:CCbipartite}
If $X_i$ is a constituent digraph with vertices in color classes
$C_j$ and $C_{\ell}$ then $X_i$ is semiregular (all vertices in
$C_j$ have the same degree and the same holds for $C_{\ell}$)
and the weakly connected components of $X_i$ equipartition 
each of the two color classes.  In particular, all 
weakly connected components of $X_i$ have the same number of vertices.
\end{proposition}

\begin{proposition}[Homogeneous connected components] \label{prop:CChomog}
If $X_i$ is a homogeneous constituent in color class $C_j$ then
each weakly connected component of $X_i$ is strongly connected, and
the connected components equipartition $C_j$.
\end{proposition}

Corollary~\ref{cor:contract} below will be used in the
justification of one of our main algorithms, see
Lemma~\ref{lem:contracting}.  We start with three
preliminary lemmas.

\begin{lemma}[Neighborhood of connected component of constituent]
\label{lem:comp-neighborhood}
Let $\xxx=(\Omega;\calR)$ be a coherent configuration.
Let $C_1$ and $C_2$ be two distinct vertex-color classes.
Let $B_1,\dots,B_m$ be the connected components of 
the homogeneous constituent digraph $X_3=(C_1,R_3)$ in $C_1$
and let $X_4=(C_1,C_2;R_4)$ be a bipartite constituent
between $C_1$ and $C_2$ 
(so $R_4\subseteq C_1\times C_2$). 
For $j=1,\dots,m$ let $M_j$ denote
the set of vertices $v\in C_2$ such that there exists 
$w\in B_j$ 
such that $c(w,v)=4$.  Then $|M_1|=\dots=|M_m|$.
\end{lemma}
\begin{proof}
Let $J=\{j\mid (\exists u\in C_1, v\in C_2)(c(u,v)=j)\}$.
For each color $j$, the number $d_j=p(1,j,j^-)$ is the out-degree
of each vertex in $C_1$ in the constituent $X_j=(C_1,C_2;R_j)$.  
Let $x\in B_i$.  Let $N_j(x)=\{y\in C_2\mid c(x,y)=j\}$.
So $|N_j(x)|=d_j$.  For $y\in N_j(x)$, let
$f(k,j)$ denote the number of walks of length $k$ starting from $x$,
ending at $y$, and consisting of $k-1$ steps of color $3$ and one step 
of color $4$.  By coherence, this number does not depend on the choice
of $x$ and $y$ as long as
$c(x,y)=j$, justifying the notation $f(k,j)$.
Such a walk stays in $B_i$ for the first $k-1$ steps, and moves to
$C_2$ along an edge of color $4$ in the last step.  Let 
$K$ be the set of those $j$ for which $(\exists k)(f(k,j)>0)$.
Clearly, $M_i=\bigcup_{j\in K}N_j(x)$ and therefore
$|M_i|=\sum_{j\in K}d_j$.  This number does not depend on $i$.
\end{proof}

\begin{lemma}  \label{lem:comp-otherside}
Using the notation of Lemma~\ref{lem:comp-neighborhood},
let $x\in C_1$ and $y\in C_2$ such that $c(x,y)=4$.
Assume $x\in B_i$.  Let $M(x,y)=\{z\in B_i\mid c(z,y)=4\}$.
Then $|M(x,y)|$ does not depend on the choice of $x$ and $y$.
\end{lemma}
\begin{proof}
Let $L$ be the set of those colors $j$ for which 
$p(j,4,4^-)>0$ and for which there exists a pair $(u,v)$ 
of vertices such that $c(u,v)=j$ and there exists a walk
from $u$ to $v$ solely along steps of color $3$.  
(So for $j\in L$ we have $R_j\subseteq C_1\times C_1$.)
The existence of
such a walk does not depend on the choice of $u$ and $v$,
only on $c(u,v)$. Therefore, 
$M(x,y)=\bigcup\{N_j(x)\mid j\in L\}$
and so $|M(x,y)|=\sum_{j\in L}d_j$, independent of the choice of $x$ and $y$.
\end{proof}

\begin{lemma} \label{lem:comp-otherside2}
Using the notation of Lemma~\ref{lem:comp-neighborhood},
for $y\in C_2$ let $E(y)$ denote the set of those $i$
for which there exists $x\in B_i$ satisfying $c(x,y)=4$.
Then $|E(y)|$ does not depend on $y$.
\end{lemma}
\begin{proof}
Let $q=|M(x,y)|>0$ be the quantity shown not to depend on
$x,y$ in Lemma~\ref{lem:comp-otherside} $(c(x,y)=4)$.  Now
$d_{4^-}=q|E(y)|$ by Lemma~\ref{lem:comp-otherside}. 
\end{proof}
\begin{corollary}[Contracting components] \label{cor:contract}
Using the notation of Lemma~\ref{lem:comp-neighborhood},
let $Y$ be the graph $Y=(C_2,[m];E)$ where 
$(y,i)\in E$ if $(\exists x\in B_i)(c(x,y)=4)$.
Then $Y$ is semiregular.
\end{corollary}  
\begin{proof}
Regularity on the $[m]$ side is the content of 
Lemma~\ref{lem:comp-neighborhood}.  Regularity on the $C_2$
side is the content of Lemma~\ref{lem:comp-otherside2}. 
\end{proof}

\subsubsection{Twins in classical coherent configurations} 
\begin{observation}
If $x,y$ are weak twins then $c(x)=c(y)$.
\end{observation}
\begin{proof}
$x$ and $y$ are by definition in the same orbit of $\aut(\xxx)$.
\end{proof}
\begin{observation}[Weak is strong]    \label{obs:weakstrong}
If four vertices are pairwise weak twins then they are
pairwise strong twins.
\end{observation}
\begin{proof}
This is a special case of Prop.~\ref{prop:alt-sym} (case $k=2$).
\end{proof}

Next we formalize the statement  
that $\xxx$ is ``aware'' of the ``strong twin'' relation.

\begin{lemma}[Strong twin characterization] \label{st-char}
Let $\xxx$ be a coherent configuration and let $x\neq y$ be vertices.
Let $c(x,x)=\ell$ and $c(x,y)=i$; so $i\neq\ell$.  Then $x$ and $y$ are
strong twins if and only if all of the following hold:
\begin{enumerate}[(a)]
\item   $i=i^-$                    \label{it:sa}
\item   $p(i,i,i)=p(i,i,\ell)-1$   \label{it:sb}
\item   $(\forall j\notin\{\ell,i\})(p(j^-,j,\ell)=p(j^-,j,i))$.  \label{it:sc}
\end{enumerate}
\end{lemma}
\begin{proof}
I. Necessity.  By definition, $x,y$ are strong twins
if and only if $c(x,y)=c(y,x)$ and
$(\forall w\notin\{x,y\})(c(w,x)=c(w,y))$.
The first of these is the statement $i=i^-$, proving item~\eqref{it:sa}.

Let now $j\neq\ell$.   If $x$ is not an in-vertex of $X_j$ then $j\neq i$
and $p(j^-,j,\ell)=p(j^-,j,i)=0$ (because $x$ is an in-vertex 
of both $X_{\ell}$ and $X_i$).  Assume now that 
$x$ is an in-vertex of $X_j$, and let 
$J=\{w\in\Omega \mid c(w,x)=j\}$.  Let $W=\{x,y\}$.
So $J\cap W=\emptyset$ if $j\neq i$ and
$J\cap W=\emptyset$ if $j\neq i$ and
$J\cap W=\{y\}$ if $j=i$.
Note that $|J|=\deg^-(j)=p(j^-,j,\ell)$.
But $c(u,y)=c(u,x)$ for all $u\notin W$, so
$J=\{w\in\Omega\setminus W \mid c(w,x)=c(w,y)=j\}$
and therefore $|J|=p(j^-,j,i)$ if $j\neq i$ and
$|J|=p(i,i,i)+1$ if $j=i$.  These two equations
prove items \eqref{it:sc} and \eqref{it:sc}, respectively.

\medskip\noindent
II. Sufficiency.  Assume items \eqref{it:sa}--\eqref{it:sc}.
Let $W=\{x,y\}$.
Given that $c(x,y)=c(y,x)$ (item~\eqref{it:sa}), 
we only need to show that $c(u,x)=c(u,y)$
for all $u\notin W$.  

Let $u\notin W$; let $j=c(u,x)$; so $j\neq \ell$.
Let $J=\{w\in\Omega \mid c(w,x)=j\}$.  So $|J|=\deg^-(j)$.
Furthermore, $J\cap W=\emptyset$ if $j\neq i$ and
$J\cap W=\{y\}$ if $j=i$.
Let $J'=\{w\in\Omega \mid c(w,x)=c(w,y)=j\}$.
Obviously, $J'\subseteq J$.  On the other hand,
$|J'|=p(j^-,j,i)$ which is equal to $p(j^-,j,\ell)=|J|$
if $j\neq i$ by item~\eqref{it:sc}, hence $J=J'$ in this case.   
Since $u\in J$, it follows that $u\in J'$ and
therefore $c(u,y)=j$.  Assume now $j=i$.  In this case
$J'\cup\{y\}\subseteq J$ and 
$|J'\cup\{y\}|=|J'|+1=p(i,i,i)+1=p(i,i,\ell)=|J|$.
It follows that $J'\cup\{y\}= J$.  Since $u\in J\setminus \{y\}$,
we conclude that $u\in J'$ and therefore, $c(u,y)=i=c(u,x)$.
\end{proof}

\begin{corollary}[Strong-twin awareness]   \label{cor:st-aware}
Let $x,y,u,v$ be vertices of the coherent configuration $\xxx$.
Assume $c(x)=c(u)$.  Assume further that $x,y$ are strong twins.  Then $u,v$
are strong twins if and only if $c(x,y)=c(u,v)$.
\end{corollary}
\begin{proof}
Set $\ell=c(x)=c(u)$ and $i=c(x,y)$.  If $c(u,v)=i$ then the
conditions listed in Lemma~\ref{st-char} hold for $(u,v)$ because
they hold for $(x,y)$; therefore $u,v$ are strong twins.

What remains to be shown is that if $u,v$ are strong twins then
$c(u,v)=i$.  Assume $c(u,v)\neq i$.  But $c(u)=\ell$, so there is
a vertex $w$ such that $c(u,w)=i$; so $w\neq v$.  By the result of
the previous paragraph, $u$ and $w$ are strong twins.  So
$u,v,w$ are in the same strong-twin equivalence class; therefore
the transposition $\tau=(v,w)$ is an automorphism of $\xxx$.
It follows that $i=c(u,w)=c(u,v)$, a contradiction, proving
that $c(u,v)=i$.
\end{proof}

\begin{corollary}[Strong-twin equipartition]   \label{cor:st-equi}
Let $\xxx$ be a coherent configuration.  Then
every vertex-color class that includes strong twins
is equipartitioned by its strong-twin equivalence classes.
\end{corollary}
\begin{proof}
By Cor.~\ref{cor:st-aware}, these equivalence classes are the
connected components of a color constituent, and therefore
they have equal size by Prop.~\ref{prop:CChomog}.
\end{proof}

It is easy to see that ternary coherent configurations are ``aware''
of the ``weak twin'' relation.

Next we show the somewhat less obvious fact that (classical)
coherent configurations are also ``aware'' of the ``weak twin'' 
relation.

\begin{lemma}[Weak twin characterization] \label{lem:wt-char}
Let $\xxx$ be a coherent configuration and let $x\neq y$ be vertices.
Let $c(x,x)=\ell$ and $c(x,y)=i$; so $i\neq\ell$.  Then $x$ and $y$ are
weak but not strong twins if and only if all of the following hold:
\begin{enumerate}[(a)]
\item $i^-\neq i$           \label{it:wa}
\item $p(i^-,i^-,i)=1$   \label{it:wb}
\item $p(i,i,i)=0$             \label{it:wc}
\item $\deg^+(i)=\deg^-(i)=1$  \label{it:wd}
\item $(\forall j\notin\{\ell,i,i^-\})(p(j^-,j,\ell)=p(j^-,j,i))$. 
   \label{it:we}
\end{enumerate}
\end{lemma}
\begin{proof}
I. Necessity.  Assume $x,y$ are weak but not strong twins.
Let $W$ denote the weak-twin equivalence class of $x,y$.
If $|W|=2$ then $x,y$ are strong twins by definition;
if $|W|\ge 4$ then $x,y$ are strong twins by 
Obs.~\ref{obs:weakstrong}.
Therefore $|W|=3$.  Let $\sigma=(x,y,z)\in\aut(\xxx)$ be an 
automorphism witnessing that $x,y$ are weak twins; so this $z$ 
is unique and $W=\{x,y,z\}$.  It follows that 
$c(x,y)=c(y,z)=c(z,x)=i$ and
for every $w\notin W$ we have $c(w,x)=c(w,y)=c(w,z)$.

If $i=i^-$ then it follows that $c(z,x)=c(z,y)=i$, so
the transposition $\tau=(x,y)$ is an automorphism, hence
$x,y$ are strong twins, a contradition, proving item~\eqref{it:wa}.

The walk $x\stackrel{i^-}{\to} z\stackrel{i^-}{\to} y$
demonstrates that $p(i^-,i^-,i)\ge 1$.  Suppose
for a contradiction that $p(i^-,i^-,i)\ge 2$.
Then $(\exists w\notin W)(c(x,w)=c(w,y)=i^-)$.
But for $w\notin W$ we have $c(w,x)=c(w,y)$,
\ie, $i=i^-$, a contradiction, proving item~\eqref{it:wb}.

To prove item~\eqref{it:wc}, assume for a contradiction that 
$p(i,i,i)\ge 1$.  Then $(\exists w\notin W)(c(x,w)=c(w,y)=i)$.
But for $w\notin W$ we have $c(w,x)=c(w,y)$,
\ie, $i^-=i$, a contradiction, proving item~\eqref{it:wc}.

To prove item~\eqref{it:wd}, assume for a contradiction that 
$\deg^+(i)\ge 2$.  Then $(\exists w\notin W)(c(x,w)=i)$.
Therefore $c(y,w)=i$.  So the walk
$x\stackrel{i}{\to} y\stackrel{i}{\to} w$ demonstrates
that $p(i,i,i) > 0$, a contradiction, proving $\deg^+(i)=1$.  
Now we have $c(x)=c(y)$ (because $\sigma\in\aut(\xxx)$),
hence $X_i$ is a homogeneous component; therefore 
$\deg^-(i)=\deg^+(i)=1$ by Obs.~\ref{obs:deg-aware}.

Finally to prove item~\eqref{it:we}, let $j\notin\{\ell,i,i^-\}$.
If $x$ is not an in-vertex of $X_j$ then 
$p(j^-,j,\ell)=p(j^-,j,i)=0$ (because $x$ is an in-vertex 
of both $X_{\ell}$ and $X_i$).  Assume now that 
$x$ is an in-vertex of $X_j$, and let 
$J=\{w\in\Omega \mid c(w,x)=j\}$.  So $J\cap W=\emptyset$.
Note that $|J|=\deg^-(j)=p(j^-,j,\ell)$.
But $c(u,y)=c(u,x)$ for all $u\notin W$, so
$J=\{w\in\Omega \mid c(w,x)=c(w,y)=j\}$
and therefore $|J|=p(j^-,j,i)$.

\medskip\noindent
II. Sufficiency.  Assume items \eqref{it:wa}--\eqref{it:we}.
By item~\eqref{it:wb} there is a unique vertex $z$ with
$c(z,x)=c(y,z)=i$.  Let $W=\{x,y,z\}$ and let $\sigma$
denote the 3-cycle $\sigma=(x,y,z)\in\sym(\Omega)$.  We claim that
$\sigma\in\aut(\xxx)$.  

Since $\sigma$ is an automorphism of the 3-cycle (digraph) $(x,y,z)$
in $X_i$, we only need to show that $c(u,x)=c(u,y)=c(u,z)$
for all $u\notin W$.  Let $u\notin W$; let $j=c(u,x)$.
It follows from item~\eqref{it:wd} that $j\notin \{\ell,i,i^-\}$.
Let $J=\{w\in\Omega \mid c(w,x)=j\}$.  So $|J|=\deg^-(j)$ and
$J\cap W=\emptyset$.  Let $J'=\{w\in\Omega \mid c(w,x)=c(w,y)=j\}$.
Obviously, $J'\subseteq J$.  On the other hand,
$|J'|=p(j^-,j,i)=p(j^-,j,\ell)=|J|$ by item~\eqref{it:we},
hence $J=J'$.   Since $u\in J$, it follows that $u\in J'$ and
therefore $c(u,y)=j$.  The same argument, starting from the pair
$\{y,z\}$ instead of $\{x,y\}$, shows that $c(u,z)=j$.
\end{proof}

\begin{corollary}[Weak-twin awareness]   \label{cor:wt-aware}
Let $x,y,u,v$ be vertices of the coherent configuration $\xxx$.
Assume $c(x)=c(u)$.  Assume further that $x,y$ are weak twins.  
Then $u,v$ are weak twins if and only if 
$c(u,v)\in \{c(x,y), c(y,x)\}$.
\end{corollary}
\begin{proof}
Set $\ell=c(x)=c(u)$ and $i=c(x,y)$.  Assume first that $c(u,v)=i$.  
If $x,y$ are strong twins then $c(u,v)$ are also strong twins
by Cor.~\ref{cor:st-aware}.  Moreover, if $u,v$ are strong
twins then $c(u,v)=c(x,y)$, also by Cor.~\ref{cor:st-aware}.

Assume now that $x,y$ are weak but 
not strong twins.  Now conditions \eqref{it:wa}--\eqref{it:we}
hold for $(u,v)$ because they hold for $(x,y)$; therefore $u,v$ are 
weak but not strong twins.

What remains to be shown is that if $u,v$ are weak but not strong twins 
then $c(u,v)\in\{i,i^-\}$.  

Assume $c(u,v)\neq i$.  But $c(u)=\ell$, so there is
a vertex $w$ such that $c(u,w)=i$; so $w\neq v$.  By the result of
the previous paragraph, $u$ and $w$ are weak but not strong twins.  So
$u,v,w$ are in the same weak-twin equivalence class $C$.  Since
$\alt(C)\le \aut(\xxx)$ (Obs.~\ref{obs:weak-alt}), 
the 3-cycle $\sigma=(u,v,w)$ is an automorphism of $\xxx$.
It follows that $i=c(u,w)=c(w,v)=c(v,u)$, so $c(v,u)=c(x,y)$
and therefore $c(u,v)=c(y,x)=i^-$, as desired.
\end{proof}

\begin{corollary}[Weak-twin equipartition]   \label{cor:wt-equi}
Let $\xxx$ be a coherent configuration.  Then
every vertex-color class that includes weak twins
is equipartitioned by its weak-twin equivalence classes.
\end{corollary}
\begin{proof}
By Cor.~\ref{cor:wt-aware}, these equivalence classes are the
connected components of a color constituent, and therefore
they have equal size by Prop.~\ref{prop:CChomog}.
\end{proof}

\subsubsection{Uniprimitive coherent configurations (UPCCs)} 
\label{sec:UPCCdef}
Recall that $\xxx$ is \emph{homogeneous} if
all vertices have the same color.

\begin{definition}   \label{def:UPCC}
$\xxx$ is \emph{primitive} if it is homogeneous and all
constituent digraphs other than the diagonal are connected.
$\xxx$ is \emph{uniprimitive} if it is primitive and has
rank $\ge 3$, \ie, it is not the clique configuration.
\end{definition}
\begin{notation}
We abbreviate ``uniprimitive coherent configuration'' as UPCC.
\end{notation}
Combinatorial properties of UPCCs have been studied by the author
in~\cite{uniprimitive} and in great depth by Sun and Wilmes
in~\cite{sun-wilmes}.  UPCCs play an important
role in the study of GI as the obstacles to a natural
combinatorial partitioning approach.  One of the main
technical contributions of this paper is that we overcome this
obstacle (Section~\ref{sec:johnson}).

\subsubsection{Association schemes, Johnson schemes}
We say that the coherent configuration $\xxx$ is an
\emph{association scheme} if $c(x,y)=c(y,x)$ for
every $x,y\in\Omega$.  It follows that association schemes 
are homogeneous.

A particular class of association schemes will be of great interest to us.

Let $t\ge 2$ and $k\ge 2t+1$.  The \emph{Johnson scheme}
$\jjj(k,t)=(\Omega;R_0,\dots,R_t)$ is an association scheme 
with $\binom{k}{t}$ vertices corresponding to the $t$-subsets 
of an $k$-set $\Gamma$.  We identify $\Omega$ with the set
$\Omega =\binom{\Gamma}{t}$.  The relation
$R_i$ consists of those pairs $(T_1,T_2)$ 
($T_j\subset\Gamma$, $|T_j|=t$) with $|T_1\setminus T_2|=i$.

\mn
An important functor (see Section~\ref{sec:functor})
maps the category of $k$-sets to
the category of Johnson schemes $\jjj(k,t)$.  This
functor is \emph{surjective} (on $\iso(\xxx,\yyy)$ for any
pair $(\xxx,\yyy)$ of objects).  The principal content of this 
nontrivial statement is the following.

\begin{proposition}
If $t\ge 2$ and $m\ge 2t+1$ then
  $\aut(\jjj(m,t))=\sym^{(t)}(\Gamma)$.
\end{proposition}

\subsection{Hypergraphs}

\subsubsection{Basic terminology}
A \emph{hypergraph} $\calH=(V,\calE)$ consists of a vertex set $V$
and a subset $\calE$ of the power-set of $V$.  We say that
$\calH$ is \emph{$d$-uniform} if $|E|=d$ for all $E\in\calE$.

We say that $\calH$ is an \emph{empty hypergraph} if
$\calE=\emptyset$.  The \emph{complete $d$-uniform hypergraph}
has edge set $\calE=\binom{V}{d}$.  The \emph{trivial} $d$-uniform
hypergraphs are the empty and the complete ones.   In other words,
a $d$-uniform hypergraph is trivial if its automorphism group
is $\sym(V)$.

The \emph{degree} of a vertex $x\in V$ is the number of edges
containing $x$.  $\calH$ is \emph{$r$-regular} if every vertex has
degree $r$.

\begin{definition}[Induced subhypergraph]
For a subset $W\subseteq V$ we define the \emph{induced
subhypergraph} $\calH[W]$ as follows: the vertex set of 
$\calH[W]$ is $W$ and $E\in\calE$ is an edge of $\calH[W]$
if and only if $E\subseteq W$.  
\end{definition}
In particular, every induced subhypergraph of a $d$-uniform hypergraph
is $d$-uniform.

\begin{definition}[Trace of hypergraph]
Let $\calH=(V,\calE)$ be a hypergraph.  The \emph{trace} $\calH_S$
on the set $S\subseteq V$ is the set $\{E\cap S\mid E\in\calE\}$.
\end{definition}

We can treat $d$-uniform hypergraphs as $d$-ary relational structures
$(V,R)$ with a symmetric relation $R\subseteq\Omega^{\,d}$, \ie,
$(\forall \pi\in\sss_d)(R^{\,\pi}=R)$, with the additional
condition that if $(x_1,\dots,x_d)\in R$ then all the $x_i$ are
distinct.  So some results on $d$-ary relational structures apply to 
$d$-uniform hypergraphs.  We shall in particular apply the
Design Lemma (Theorem~\ref{thm:design}) to uniform hypergraphs.

\subsubsection{Random hypergraphs}  \label{sec:random-hyp}
Let $|V|=n$,\ $d\le n$, and $m\le \binom{n}{d}$.  Following
Erd\H{o}s and R\'enyi, by a \emph{random $d$-uniform hypergraph 
with $n$ vertices and $m$ edges} we mean a uniform random member 
of the set $\hyp_d(n,m)=\binom{\binom{V}{d}}{m}$.  We are concerned
with very sparse hypergraphs.  For $X\in \hyp_d(n,m)$,
let $U(X)$ denote the union of all edges of $X$.  We call $U(X)$
the \emph{support} of $X$.  Obviously $|U(X)|\le md$.
The following observation will be used in Section~\ref{sec:BPjohnson}
to justify a subroutine.

\begin{proposition}   \label{prop:random-hyp}
For a random $d$-uniform hypergraph $X$ with $n$ vertices and $m$ edges,
the probability that its support has size $|U(X)|<md$ is less than
$(md)^2/n$.
\end{proposition} 
It follows that if $md=o(\sqrt{n})$ then the edges of a typical hypergraph
with these parameters are pairwise disjoint.
\begin{proof}
Let the edges of $X$ be $E_1,\dots,E_m$ (numbered uniformly at random).  
Let us say that 
$v\in V$ is a \emph{unique vertex of $E_i$} if 
$v\in E_i\setminus W_i$ where $W_i=\bigcup_{j:j\neq i} E_j$.
Let $\xi_i$ denote the number of unique vertices of $E_i$.
So $|U(X)|\ge \sum\xi_i$.  Let $\E$ stand for expected value.

\mn
Claim. $\E(\xi_i)\ge d - (m-1)d^2/n$.

\begin{proof}
Let $\eta_i = d-\xi_i = |E_i\cap W_i|$.  So our claim says 
that $\E(\eta_i)\le (m-1)d^2/n$.  Let us fix the set 
$\calE_i=\{E_j\mid 1\le j\le m, j\neq i\}$ edges other than $E_i$
and let $\eta$ denote $\eta_i$ conditioned on the given choice
of $\calE_i$.

We claim that for every $\calE_i$ we have $\E(\eta)\le (m-1)d^2/n$.

The edge $E_i$ is chosen uniformly at random from
$\binom{V}{d}\setminus \calE_i$.  Let $F$ be chosen uniformly at
random from $\binom{V}{d}$ and let $\zeta=|F\cap W_i|$.  Now if
$F\in\calE_i$ then $\zeta=d$, so 
$\E(\zeta)=\epsilon d + (1-\epsilon)\E(\eta) > E(\eta)$, where
$\epsilon = (m-1)/\binom{n}{d}$.  We claim that $E(\zeta)\le (m-1)d^2/n$.
Indeed, we have $\E(|F\cap W_i|)=|F||W_i|/n \le (m-1)d^2/n$.
\end{proof}

It follows from the Claim that 
$\E(|U(X)|)\ge\sum_i \E(\xi_i)\ge md-m(m-1)d^2/n$.
Let $\theta = md-|U(X)|$.  So $\E(\theta)\le m(m-1)d^2/n$.  But $\theta\ge 0$,
so by Markov's inequality, $\P(\theta\ge 1) \le \E(\theta) \le
m(m-1)d^2/n < (md)^2/n$.
\end{proof}

\subsubsection{Twins, symmetry defect}
In Section~\ref{sec:twins} we defined concepts of weak/strong twins,
weakly/strongly symmetrical sets, (relative) symmetry defect
for relational structures via their automorphism groups.
We make the analogous definitions for hypergraphs, so
for instance the \emph{symmetry defect} of the hypergraph
$\calH$ is the symmetry defect of $\aut(\calH)$.

\begin{proposition}[Weak is strong]   \label{prop:weak-is-strong}
Let $\calH=(V,\calE)$ be a hypergraph and $x,y\in V$, $x\neq y$.
If $x$ and $y$ are weak twins then they are strong twins.
\end{proposition}
\begin{proof}
Let $\sigma=(x,y,z)$ be a 3-cycle in $\sym(V)$ and let
$\tau=(x,y)$ (transposition).  We need to show that if
$\sigma\in\aut(\calH)$ then $\tau\in\aut(\calH)$.  Suppose otherwise;
let $E\in\calE$ be a witness to this, \ie, $E^{\,\tau}\notin\calE$
but $E^{\,\sigma}\in\calE$ and $E^{\,\sigma^2}\in\calE$.  
It follows that one of $x$ and $y$ is in $E$ and the other is not,
say $x\in E$ and $y\notin E$. Now if $z\not\in E$ then
$E^{\,\tau}=E^{\,\sigma}\in\calE$, a contradiction.
If $z\in E$ then $E^{\,\tau}=E^{\,\sigma^2}\in\calE$, 
again a contradiction.
\end{proof}
So for hypergraphs, we can omit the adjectives ``strong/weak''
when talking about twins and about symmetrical sets.

\subsubsection{Skeletons}

The ``Skeleton defect lemma'' (Lemma~\ref{lem:skeleton-defect}) 
below will play an important role in the analysis of the 
{\sf Split-or-Johnson} routine (see Section~\ref{sec:blockdesign},
Case~3b).

\begin{definition}  \label{def:hyp-skeleton}
The $t$-skeleton of the hypergraph $\calH=(V,\calE)$
is the $t$-uniform hypergraph $\calH^{(t)}=(V,\calE^{(t)})$
where $F\in\binom{V}{t}$ belongs to $\calE^{(t)}$ exactly  if
there exists $E\in\calE$ such that $F\subseteq E$.
\end{definition}
Note that for $d$-uniform hypergraphs, viewed as $d$-ary relational
structures, this definition does not agree with Def.~\ref{def:skeleton},
although it is similar in spirit.

\begin{proposition}   \label{prop:skeleton}
Let $\calH$ be a nontrivial $d$-uniform hypergraph with $n$ vertices
and $m$ edges, where $d\le n/2$.
Then there exists $t\le \min\{d, 1+\log_2 m\}$ such that
the $t$-skeleton $\calH^{(t)}$ is nontrivial.
\end{proposition}
\begin{proof}
Choose $t=d$ if $d\le 1+\log_2 m$.  Otherwise let $t=1+\lfloor\log_2 m\rfloor$.
Let $x_1,\dots,x_t$ be independently uniformly selected vertices of $\calH$.
The probability that all of them belong to an edge $E\in\calE$ is
$(|E|/n)^t\le 1/2^t$.   The probability that there exists an edge
to which all the $x_i$ belong is less than $m/2^t$ which is less than 1 if
$t > \log_2 m$.  So $\calH^{(t)}$ is not complete.  It is also not
empty since $t\le d$. 
\end{proof}

\begin{proposition}    \label{prop:proportion}
Let $\calH=(V,\calE)$ be a nonempty, regular, $d$-uniform hypergraph.
Let $S\subseteq V$.  Let $\alpha=|S|/|V|$.  Then there is an edge
$E\in\calE$ such that $|E\cap S|\ge\alpha d$.
\end{proposition}
\begin{proof}
Let $|V|=n$ and $|\calE|=m$.
Each vertex belongs to $md/n$ edges, so for each vertex
$x$, the probability that $x\in E$ for a randomly 
selected edge is $d/n$.  Therefore the expected number of
vertices in $|S\cap E|$ for a random edge $E$ is
$|S|d/n=\alpha d$.
\end{proof}

\begin{lemma}[Skeleton defect lemma]  \label{lem:skeleton-defect}
Let $\calH=(V,\calE)$ be a nontrivial, regular, $d$-uniform hypergraph 
with $n$ vertices and $m$ edges where $d\le n/2$.
Let $(7/4)\log_2 m\le t\le (3/4)d$.  Then the symmetry defect
of the $t$-skeleton $\calH^{(t)}$ is greater than $1/4$.
\end{lemma}
\begin{proof}
Let $S\subseteq V$ be a symmetrical subset.  Assume for a contradiction
that $|S|\ge 3n/4$.  Then, by Prop.~\ref{prop:proportion}, there is
an edge $E\in\calE$ such that $|S\cap E|\ge (3/4)d\ge t$.  Let 
$T\subseteq S\cap E$,\ $|T|=t$.   So $T\in\calE^{(t)}$.  Since 
$S$ is a symmetrical set, it follows that 
$\binom{S}{t}\subseteq \calE^{(t)}$.   Since every edge of $\calH$
contains at most $\binom{d}{t}$ of these $t$-sets, it follows that
\begin{equation}
 m\ge \frac{\binom{|S|}{t}}{\binom{d}{t}} >
      \left(\frac{3n/4}{d}\right)^t \ge   
      \left(\frac{3}{2}\right)^t > m,
\end{equation}
a contradiction.
\end{proof}

\subsection{Individualization and canonical refinement}
Let $\xxx$ be a structure such as a graph, digraph, 
$k$-ary relational structure, hypergraph, with 
colored elements (vertices, edges, $k$-tuples, hyperedges).
The colors form an ordered list.  A \emph{refinement}
of the coloring $c$ is a new coloring $c'$ of the same elements
such that if $c'(x)=c'(y)$ for elements $x,y$ then $c(x)=c(y)$;
this results in the refined structure $\xxx'$.
We say that the refinement is \emph{canonical} with respect to
a set  $\{\xxx_i\mid i\in I\}$ of objects of the same type if
it is executed simultaneously on each $\xxx_i$ 
and for all $i,j\in I$ we have
\begin{equation}
   \iso(\xxx_i',\xxx_j') = \iso(\xxx_i,\xxx_j) .
\end{equation}
(This is consistent with the functorial notion of
canonicity explained in Sec.~\ref{sec:functor}.)
Naive vertex refinement (refine vertex colors by 
number of neighbors of each color) has been the
basic isomorphism rejection heuristic for ages.
More sophisticated canonical refinement methods
are explained in the next section.

Another classical heuristic is \emph{individualization:}
the assignment of a unique color to an element.
Let $\xxx_x$ denote $\xxx$ with the element $x$ individualized.
If the number of elements of the given type is $m$
then individualization incurs a multiplicative cost of $m$:
when testing isomorphism of structures $\xxx$ and $\yyy$,
if we individualize $x\in\xxx$, we need compare $\xxx_x$
with all $\yyy_y$ for $y\in\yyy$: for any $x\in\xxx$ we have
\begin{equation}
     \iso(\xxx,\yyy)=\bigcup_{y\in\yyy} \iso(\xxx_x,\yyy_y).
\end{equation}
(Compare this with the more general categorical concept in
Sec.~\ref{sec:functor}.)

The individualization/refinement method (I/R) (individualization
followed by refinement) is a powerful heuristic and has also
been used to proven advantage (see e.\,g., 
\cite{lasvegas,uniprimitive,canonical,codenotti,bw,cst,bcstw}),
even though strong limitations of its isomorphism rejection capacity
have also been proven~\cite{cfi}.  I/R combines well with the group 
theory method and the combination is not subject to the CFI limitations
(\cite{lasvegas,canonical,codenotti, bcstw}).  The power of
this combination is further explored in this paper.

\subsection{Weisfeiler-Leman canonical refinement}
\label{sec:WL}

\subsubsection{Classical WL refinement}
The classical Weisfeiler--Leman\footnote{Weisfeiler's 
book~\cite{weisfeiler-book} transliterates Leman's
name from the original Russian as ``Lehman.'' 
However, Andre\u\i\ Leman (1940--2012) himself omitted the ``h.''
(Sources: private communications by Mikhail Klin, Aug. 2006,
and by Ilya Ponomarenko, Jan. 2016.
Both Klin and Ponomarenko forwarded to me 
email messages they had received in the late 1990s from Leman.
The ``From'' line of each message reads
``{\tt From: Andrew Leman <andyleman@etc.>},''
and Leman also verbally expressed this preference.)}
(WL) refinement 
\cite{weisfeiler-leman,weisfeiler-book} takes as input a
binary configuration and refines it to a 
coherent configuration (see Sec.~\ref{sec:UPCCdef}) as follows.
The process proceeds in rounds.  Let $\xxx$ be the input to a 
round of refinement.  For $(x,y)\in\Omega\times\Omega$,
we encode in the
new color $c'(x,y)$ the following information: the old color
$c(x,y)$, and for all $j,k\le r$, the number
$|\{z\in\Omega\mid c(x,z)=j$ and $c(z,y)=k\}|$.
These data form a list, naturally ordered.  To each list
we assign a new color; these colors are sorted lexicographically.
This gives a refined coloring that defines a new configuration
$\xxx'$.  We stop when we reach a
stable configuration ($\xxx=\xxx'$, \ie, no refinement occurs,
\ie, no $R_i$ is split).  

\begin{observation}  \label{obs:stable-coherent}
The stable configurations under WL refinement are precisely the
coherent configurations.
\end{observation}
The process is clearly
\emph{canonical} in the following sense.  Let $\xxx$ and $\yyy$
be configurations.  We simultaneously execute each round of
refinement (merging the lists of refined colors).  Let 
$\xxx^*$ and $\yyy^*$ be the coherent configurations obtained.
Then
\begin{equation}  \label{eq:WL-canonical}
    \iso(\xxx,\yyy)=\iso(\xxx^*,\yyy^*).
\end{equation}
In particular, if one of the colors of $\xxx^*$ does not occur in
$\yyy^*$ then $\xxx$ and $\yyy$ are not isomorphic, so WL gives
an isomorphism rejection tool.

\subsubsection{$k$-dimensional WL refinement}
The $k$-ary version of this process, referred to as ``$k$-dimensional
WL refinement,'' was introduced by
Mathon and this author~\cite{toronto}
in 1979 and independently by Immerman and Lander~\cite{immerman}
in the context of counting logic, cf.~\cite{cfi}.  
The refinement step is defined as follows. 
Let $\xxx=(\Omega;R_1,\dots,R_r)$ be a $k$-ary configuration
(Sec.~\ref{sec:kary}).  
For $\vx=(x_1,\dots,x_k)\in\Omega^k$
we encode in the
new color $c'(\vx)$ the following information: the old color
$c(\vx)$, and for all $i_1,\dots,i_k\le r$, the number
$|\{y\in\Omega\mid (\forall j\le r)(c(\vx_j^{\,y})=i_j)\}|$.
As before, these data form a list, naturally ordered.  To each list
we assign a new color; these colors are sorted lexicographically.
This gives a refined coloring that defines a new configuration
$\xxx'$.  We stop when we reach a
stable configuration ($\xxx=\xxx'$).  
Observation~\ref{obs:stable-coherent} remains valid,
as is the canonicity of the stable configuration stated
in Eq.~\eqref{eq:WL-canonical}.

As far as I know, this paper is the first to derive analyzable
gain from employing the $k$-dimensional WL method for unbounded
values of $k$ (or any value $k>4$).  (In fact, I am only aware
of one paper that goes beyond $k=2$~\cite{bcstw}.)
We use $k$-dimensional WL
in the proof of the Design Lemma (Thm.~\ref{thm:design}).
In our applications of the Design Lemma, the value of 
$k$ is polylogarithmic (see Secs. 
\ref{sec:localguide},~\ref{sec:aggregate}).

\mn
\subsubsection{Complexity of WL refinement.}
The stable refinement ($k$-ary coherent configuration)
can trivially be computed in time $O(k^2n^{2k+1})$ and nontrivially
in time $O(k^2n^{k+1}\log n)$ \cite[Sec. 4.9]{immerman}.

\section{Algorithmic setup}

\subsection{Luks's framework}   \label{sec:luks}
In this section we review Luks's framework using notation and
terminology that better suits our purposes.

\mn
Let $G\le\sym(\Omega)$ be a permutation group
acting on the domain $\Omega$.  $G$ will be represented
concisely by a list of generators; if $|\Omega|=n$ then
every minimal set of generators has $\le 2n$ elements~\cite{chain}.

Let $\Sigma$ be a finite alphabet.  We consider the set 
of strings $\xx$ over the alphabet $\Sigma$ indexed by $\Omega$, \ie,
mappings $\xx : \Omega\to\Sigma$.  For 
$\tau\in\sym(\Omega)$ and $\xx :\Omega\to\Sigma$
we define the string
$\xx^{\tau}$ by setting $\xx^{\tau}(u)=\xx(u^{\tau^{-1}})$
for all $u\in\Omega$.  In other words, for all $u\in\Omega$
and $\tau\in\sym(\Omega)$,
\begin{equation}
       \xx^{\tau}(u^{\tau}) = \xx(u).
\end{equation}
(The purpose of the inversion 
is to ensure that $\xx^{\sigma\tau}=(\xx^{\sigma})^{\tau}$
for $\sigma,\tau\in\sym(\Omega)$.)

For $K\subseteq\sym(\Omega)$ we say that $\tau$ is a
\emph{$K$-isomorphism} of strings $\xx$ and $\yy$ if $\tau\in K$
and $\xx^{\tau}=\yy$.  Let $\iso_K(\xx,\yy)$ denote the
set of $K$-isomorphisms of $\xx$ and $\yy$:
\begin{equation}
   \iso_K(\xx,\yy)=\{\tau\in K\mid \xx^{\tau}=\yy\}=
           \{\tau\in K\mid (\forall u\in\Omega)(\xx(u)=\yy(u^{\tau})\}
\end{equation}
and let $\aut_K(\xx)=\iso_K(\xx,\xx)$ denote the set of
$K$-automorphisms of $\xx$.   
\begin{remark}
The only context in which we use this concept is when $K$
is a coset.  However, the general principles are more
transparent in this more general context.
\end{remark}

In the Introduction we stated the String Isomorphism decision problem:
``Is $\iso_G(\xx,\yy)$ not empty?''  In the rest of the paper we shall
use the term ``String Isomorphism problem'' for the computation version
(compute the set $\iso(\xx,\yy)$).  The decision and computation versions
are polynomial-time equivalent (under Cook reductions).
\begin{definition}[String Isomorphism Problem]  \label{def:string}
\begin{tabbing} mmmmmm \= m \= m \= m \= m \kill \\
Input: \> a set $\Omega$, a finite alphabet $\Sigma$, 
       two strings $\xx,\yy:\Omega\to\Sigma$,  \\
       \> a permutation group
          $G\le\sym(\Omega)$ (given by a list of generators)  \\
Output: \> the set $\iso_G(\xx,\yy)$.\quad If this set is nonempty, 
           it is represented by a list of \\
        \>  generators of the group $\aut_G(\xx)$
            and a coset representative $\sigma\in\iso_G(\xx,\yy)$.
\end{tabbing}
\end{definition}

For $K\subseteq\sym(\Omega)$
and $\sigma\in \sym(\Omega)$ we note the \emph{shift identity}
\begin{equation}   \label{eq:shift}
   \iso_{K\sigma}(\xx,\yy)=\iso_K(\xx,\yy^{\sigma^{-1}})\sigma .
\end{equation}

For the purposes of recursion we need to introduce one more variable,
a subset $\Delta\subseteq \Omega$ to which 
we shall refer as the \emph{window}.  
\begin{definition}[Window isomorphism]
Let $\Delta\subseteq \Omega$ and $K\subset\sym(\Omega)$.
Let 
\begin{equation}
\iso_K^{\Delta}(\xx,\yy)=\{\tau\in K\mid 
  (\forall u\in\Delta)(\xx(u)=\yy(u^{\tau}))\} .
\end{equation}
\end{definition}
For $K\subseteq\sym(\Omega)$
and $\sigma\in \sym(\Omega)$ we
again have the \emph{shift identity:}
\begin{equation}   \label{eq:shift-window}
   \iso_{K\sigma}^{\Delta}(\xx,\yy)=
    \iso_K^{\Delta}(\xx,\yy^{\sigma^{-1}})\sigma .
\end{equation}

\begin{remark}[Alignment]
Applying Eq.~\eqref{eq:shift} to a \emph{subgroup}
$K=G\le \sym(\Omega)$, we see that
the isomorphism problem for the pair $(\xx,\yy)$ of strings
with respect to a coset $G\sigma$ is the same as the isomorphism
problem for $(\xx,\yy^{\sigma^{-1}})$ with respect to the group $G$.  
In view of Eq.~\eqref{eq:shift-window}, the same holds
for window-isomorphism.
The shift $\yy\leftarrow\yy^{\sigma^{-1}}$ is an
important {\bf alignment step} that will accompany every reduction
of the ambient group $G$.  
\end{remark}
\begin{remark}
When applying the concept of window-isomorphism, we shall always
assume that the window is \emph{invariant} under the group $G\le\sym(\Omega)$,
and $K$ is a coset, $K=G\sigma$ for some $\sigma\in\sym(\Omega)$.
Under these circumstances we make the following observations.
\begin{enumerate}[(i)]
\item  $\aut_G^{\Delta}(\xx)$ is a subgroup of $G$ 
\item  $\iso_{G\sigma}^{\Delta}(\xx,\yy)$ is either
       empty or a right coset of $\aut_G^{\Delta}(\xx)$, namely,
\begin{equation}
 \iso_{G\sigma}^{\Delta}(\xx,\yy)=\aut_G^{\Delta}(\xx)\sigma'
      \text{\quad for any\quad} \sigma'\in \iso_{G\sigma}^{\Delta}(\xx,\yy)\,.
\end{equation}
\end{enumerate}
However, again, the general principles are more
transparent in the more general context where 
$K$ is an arbitrary subset of $\sym(\Omega)$ and
$\Delta$ is an arbitrary subset of $\Omega$.
\end{remark}
The following straighforward identity plays a central role in Luks's
method.  Let $K,L\subseteq\sym(\Omega)$ and $\Delta\subseteq\Omega$.
Then
\begin{equation}   \label{eq:union}
   \iso_{K\cup L}^{\Delta}(\xx,\yy)=
   \iso_{K}^{\Delta}(\xx,\yy)\cup\iso_{L}^{\Delta}(\xx,\yy)
\end{equation}

Next we describe Luks's group-theoretic divide-and-conquer
strategies.

\begin{proposition}[Weak Luks reduction] \label{prop:weakluks}
Let $H\le G$.  Then finding $\iso_G^{\Delta}(\xx,\yy)$ reduces to 
$|G:H|$ instances of finding $\iso_H^{\Delta}(\xx,\yy^{\sigma})$ for
various $\sigma\in G$.
\end{proposition}
\begin{proof}
We can write $G=\bigcup_{\sigma}H\sigma$ where $\sigma$ ranges over a set 
of right coset representatives of $H$ in $G$.  
Applying Eq.~\eqref{eq:union} to this decomposition, we obtain
\begin{equation}   \label{eq:weakluks}
   \iso_G^{\Delta}(\xx,\yy)=
   \bigcup_{\sigma}\iso_{H\sigma}^{\Delta}(\xx,\yy)=
   \bigcup_{\sigma}\iso_H^{\Delta}(\xx,\yy^{\sigma^{-1}})\sigma
\end{equation}
where we also employed the shift identity, Eq.~\eqref{eq:shift-window}.
\end{proof}

The following identity describes Luks's basic recurrence for
sequential processing of windows.
\begin{proposition}[Chain Rule]  \label{prop:chainrule}
Let $\Delta_1$ and $\Delta_2$ be $G$-invariant 
subsets of $\Omega$ and let 
 $\iso_G^{\Delta_1}(\xx,\yy)=G_1\sigma$, where $\sigma\in G$
and $G_1\le G$.  Then
\begin{equation}
  \iso_G^{\Delta_1\cup\Delta_2}(\xx,\yy)= 
  \iso_{G_1\sigma}^{\Delta_2}(\xx,\yy)=
  \iso_{G_1}^{\Delta_2}(\xx,\yy^{\sigma^{-1}})\sigma .
\end{equation}
\end{proposition}
\begin{proof}
The first equation is immediate from the definitions.
The second equation uses the shift identity,
Eq.~\eqref{eq:shift-window}.
\end{proof}

We can now describe what we call ``strong Luks reduction.''
Recall the restriction notation $G^{\Delta}$ (Notation~\ref{not:restriction}).
\begin{theorem}[Strong Luks reduction] \label{thm:strongluks}
Let $G\le\sym(\Omega)$ and let $\Delta\subseteq\Omega$ be a $G$-invariant
subset.  Let $\{B_1,\dots,B_m\}$ be a $G$-invariant partition
of $\Delta$.  Let $\psi: G\to\Gbar\le\sss_m$ be the 
induced action of $G$ on the set of blocks and let $N=\ker(\psi)$.
Then finding $\iso_G^{\Delta}(\xx,\yy)$ reduces to $m|\Gbar|=m|G/N|$ 
instances of finding $\iso_{M_i}^{B_i}(\xx,\yy^{\sigma_i})$ for the
blocks $B_i$ and certain subgroups $M_i\le N$ and $\sigma_i\in G$.
\end{theorem}
(The cost of the reduction is polynomial per instance.)
\begin{proof}
First apply weak Luks reduction with $H=N=\ker\psi$.  Then
consider each $B_i$ to be the window in succession, reducing the
group at each round, following the Chain Rule.  In the end,
combine all the results into a single coset.
\end{proof}

Following Luks, the way this reduction is typically used is by
taking a minimal system of imprimitivity (a system with at least
two blocks that cannot be made coarser, \ie, the blocks are
maximal) so $\Gbar$ is a primitive group.  Therefore the order of
primitive groups involved in $G$ (action of subgroups on a system
of blocks of imprimitivity of the subgroup) is a critical parameter
of the performance of Luks reduction.

A final observation: when trying to determine $\iso_G^{\Delta}(\xx,\yy)$,
it suffices to consider the case $\Delta=\Omega$ (Obs.~\ref{obs:nowindow}
below).

\begin{definition}[Straight-line program] \label{def:straightline}
Given a group $G$ by a list $S$ of generators, a
\emph{straight-line program} of length $\ell$ in $G$ is a sequence
of length $\ell$ of elements of $G$ such that each element in the
sequence is either one of the generators or is a product of two
elements earlier in the sequence or is the inverse of an 
element earlier in the sequence.  We say that the straight-line program
\emph{computes} a set $T$ of elements if the elements of $T$ appear
  in the sequence and are marked as belonging to $T$.  A subgroup
  is computed if a set of generators of the subgroup is computed.
  A coset is computed if the corresponding subgroup and a coset
  representative are computed.
\end{definition}
\begin{observation}[Reducing to the window] \label{obs:nowindow}
Let $G\le \sym(\Omega)$ and let $\Delta$ be a $G$-invariant subset of
$\Omega$.  
Let $\xx^{\Delta}$ and $\yy^{\Delta}$ be the
restriction of $\xx$ and $\yy$ to $\Delta$, respectively.  
Given a straight-line program of length $\ell$ 
that computes
$\iso_{G^{\Delta}}(\xx^{\Delta},\yy^{\Delta})$, we can, in time
$O(n\ell)+\poly(n)$, compute $\iso_G^{\Delta}(\xx,\yy)$ (where $n=|\Omega|$).
\end{observation}
\begin{proof}
While we concentrate on the action of the elements of $G$ on the window,
we maintain their ``tails'' -- their action on the rest of the
permutation domain.  The set $\iso_{G^{\Delta}}(\xx^{\Delta},\yy^{\Delta})$
is empty if and only if $\iso_G^{\Delta}(\xx,\yy)$ is empty.
If $\iso_{G^{\Delta}}(\xx^{\Delta},\yy^{\Delta})$ is not empty,
in the end we obtain a subset $S\subseteq G$ and an element $\sigma\in G$
such that the restiction of the elements of $S$ to $\Delta$ generates
$\aut_{G^{\Delta}}^{\Delta}(\xx)$ and the restiction of $\sigma$ to
$\Delta$ belongs to $\iso_{G^{\Delta}}(\xx^{\Delta},\yy^{\Delta})$.
Adding to $S$ a set of generators of the kernel of the $G$-action 
on $\Delta$ we obtain a set of generators of $\aut_G^{\Delta}(\xx)$;
and $\iso_G^{\Delta}(\xx,\yy)=\aut_G^{\Delta}(\xx)\sigma$.
\end{proof}
Once again we stress that everything in this section was a review 
of Luks's work.

\subsection{Johnson groups are the only barrier}

The barriers to efficient application of Luks's reductions are
large primitive groups involved in $G$.

The following result reduces the Luks barriers to the 
class of Johnson groups at a multiplicative cost of $\le n$.

\begin{theorem}   \label{thm:cameron-to-johnson}
Let $G\le \sss_n$ be a primitive group of order $|G|\ge 2^{1+\log_2 n}$
where $n$ is greater than some absolute constant.  Then $G$ has
a normal subgroup $N$ of index $\le n$ such that $N$ has a 
system of imprimitivity on which $N$ acts as a Johnson group 
$\aaa_k^{(t)}$ with $k\ge \log_2 n$.
Moreover, $N$ and the system of imprimitivity in question
can be found in polynomial time.
\end{theorem}

The mathematical part of this result is an immediate consequence
of Cameron's classification of large primitive groups which we 
state below.

The \emph{socle} $\soc(G)$ of the group $G$ is defined as the
product of its minimal normal subgroups.  $\soc(G)$ can be written 
as $\soc(G) = R_1\times\dots\times R_s$ where the $R_i$ are isomorphic 
simple groups.  


\begin{definition}
$G\le \sss_n$ is a \emph{Cameron group} with parameters
$s,t\ge 1$ and $k\ge \max(2t,5)$ if for some $s, t\ge 1$ and 
$k > 2t$ we have $n=\binom{k}{t}^s$, the socle of $G$
is isomorphic to $\aaa_k^s$, and
$(\aaa_k^{(t)})^s\le G\le \sss_k^{(t)}\wr \sss_s$
(wreath product, product action), moreover the induced 
action $G\to \sss_s$
on the direct factors of the socle is transitive.
\end{definition}

Note that for $k\ge 5$ the Johnson groups $\sss_k^{(t)}$ and
$\aaa_k^{(t)}$ are exactly the Cameron groups with $s=1$.

\begin{theorem}[Cameron~\cite{cameron}, Mar\'oti~\cite{maroti}]
For $n\ge 25$, if $G$ is primitive and 
$|G|\ge n^{1+\log_2 n}$ then $G$ is a Cameron group.
\end{theorem}

We can further reduce Cameron groups to Johnson groups.

\begin{proposition}
If $G\le\sss_n$ is a Cameron group with parameters $k,t,s$ then
$ts\le \log_2 n$.  Moreover, $s\le \log n/\log k\le \log n/\log 5$.
\end{proposition}
\begin{proof}
We have  $n=\binom{k}{t}^s\ge (k/t)^{ts}\ge 2^{ts}.$  Moreover,
$n=\binom{k}{t}^s\ge k^{s}.$  
\end{proof}

\begin{proposition}
If $G\le\sss_n$ is a Cameron group with parameters $k,t,s$ 
and $|G|\ge n^{1+\log_2 n}$ then $k\ge \log_2 n$ and
$s! < n$, assuming $n$ is greater than an absolute constant.
\end{proposition}
\begin{proof}
As before, we have  $n\ge k^{s}.$  On the other hand
$n^{1+\log_2 n} \le |G| \le (k!)^s s! < k^{ks}s! \le n^k s!
  < n^k (\log_2 n)^{\log_2 n}=
   n^{k + \log_2\log_2 n}$.  Therefore $k>\log_2n -\log_2\log_2 n
  > \log_2 n/\log 5 \ge s$.  Hence, $s! < s^s\le k^s\le n$.
  Moreover, $n^{1+\log_2 n} < n^k s! < n^{k+1}$, hence 
  $k\ge \log_2 n$.  
\end{proof}

This completes the proof of the mathematical part of 
Theorem~\ref{thm:cameron-to-johnson}.  The algorithmic part
is well known: Cameron groups can be recognized and their
structure mapped out in polynomial time (and even in NC~\cite{nc}).

\subsection{Reduction to Johnson groups}
\label{sec:reducetojohnson}
We summarize the reduction to Johnson groups.

\mn
{\sf Procedure Reduce-to-Johnson}

\mn
Input: group $G\le\sym(\Omega)$, strings $\xx,\yy : \Omega\to\Sigma$

\mn
Output: $\iso_G(\xx,\yy)$ or updated $\Omega,G,\xx,\yy$, $G$ transitive, 
  with set
  $\calB$ of blocks on which $G$ acts as Johnson group $\ggg\le\sym(\calB)$

\begin{enumerate}
\item \bif\ $G\le \aut(\xx)$ \bthen\ 
   \begin{itemize}
   \item[] \bif\ $\xx=\yy$ \bthen\ \breturn\ $\iso_G(\xx,\yy)=G$, \bexit\ 
   \item[] \belse\ \breturn\ $\iso_G(\xx,\yy)=\emptyset$, \bexit\ 
   \end{itemize}
\item \bif\ $|G|<C_0$ for some absolute constant $C_0$ \bthen\
            compute $\iso_G(\xx,\yy)$ by brute force, \bexit\ 
\item \bif\ $G$ intransitive \bthen\ apply Chain Rule 
\item (: $G$ transitive :) Find minimal block system $\calB$. Let $m=|\calB|$.
       Let $\ggg\le\sym(\calB)$ be the induced $G$-action on $\calB$
       and $N$ the kernel of the $G\to \ggg$ epimorphism
       (: $\ggg$ is a primitive group :) 
\item \bif\ $|\ggg|< m^{1+\log_2 m}$ then reduce $G$ to $N$ via
       strong Luks reduction
\item \belse\ (: $\ggg$ a Cameron group of order $\ge m^{1+\log_2 m}$ :)
       reduce $\ggg$ to Johnson group via weak Luks reduction 
       (: Theorem~\ref{thm:cameron-to-johnson}, multiplicative cost $\le m$ :)
\item  (: $\ggg$ a Johnson group :) \\ 
       \breturn\ $\Omega, G, \calB, \ggg$ (Johnson group), $\xx, \yy$ 
\end{enumerate}

Our contribution is a {\sf ProcessJohnsonAction} routine
that takes the output of the last line as input.
The paper is devoted to this algorithm; it is
summarized in the Master Algorithm, starting with line~2
of that algorithm (Sec.~\ref{sec:master}).

\subsection{Cost estimate}
\label{sec:cost-estimate}
We describe the recurrent estimate of the cost.

By the cost of the algorithm we mean the number of group operations
performed on the domain $\Omega$.

For a real number $x\ge 1$, let $T(x)$ denote the worst-case cost
of solving String Isomorphism for strings of length $\le x$.  Let
$T_{\trans}(x)$ denote the same quantity restricted to transitive
groups and $T_{\jh}(x)$ the same quantity further restricted to the
case when $G$ acts on a minimal system of imprimitivity as a
Johnson group of order $\ge m^{1+\log_2 m}$ where $m$ is the number
of blocks $(2\le m\le x)$.  We obtain the following recurrences.  
Here $p(x)$ denotes a polynomial, representing the overhead incurred 
in the reductions.  $C_1$ is an absolute constant.  For $x<2$ we set 
$T(x)=T_{\trans}(x)=1$.  For $x\ge C_0$ (an absolute constant),
Luks reductions yield the following recurrences:

\begin{enumerate}[(i)]
\item  \label{item:recurrence-i}
  $T(x)\le \max\left\{\sum T_{\trans}(n_i)+p(x)\right\}$, where 
  the maximum is taken over all partitions of $\lfloor x\rfloor$ as
  $\lfloor x\rfloor =\sum_i n_i$ into positive integers $n_i$, 
  including the trivial partition $n_1=\lfloor x\rfloor$
  (Chain Rule)
\item \label{item:recurrence-ii} 
    $T_{\trans}(x) \le \max\{m^{2+\log_2 m}(T(x/m)+p(x)),
        m(T_{\jh}(x)+p(x))\}$,
  where the maximum is taken over all $m$ where $2\le m\le x$
  (strong Luks reduction; $m=n\le x$ covers the case when 
   $G$ is primitive)
\end{enumerate}

Assume we are looking for an upper bound $T_1(x)$ on $T(x)$
that satisfies $T_1(x)\ge x^{c\log_2 x}$ for some constant $c>1$
and is a ``nice'' function 
in the sense that $\log\log T_1(x)/\log \log x$ is monotone nondecreasing
for sufficiently large $x$. In this case we can replace 
item~\eqref{item:recurrence-i} by
\begin{itemize}
 \item[(i')]     $T(x)\le 1.1T_{\trans}(x)$ .
\end{itemize}
(The factor 1.1 absorbs the additive polynomial term.)
Moreover, we can ignore the first part of the right-hand 
side of Eq.~\eqref{item:recurrence-ii} since $T_1(x)$ automatically
satisfies 
$T_1(x) \ge m^{2+\log_2 n}(T_1(x/m)+p(x))$ (for all $m$, $2\le m\le x$,
assuming $x$ is sufficently large), so we only need to assume

\begin{itemize}
\item[(ii')]     $T_{\trans}(x) \le 1.1xT_{\jh}(x)$.
\end{itemize}
(Again, the factor 1.1 absorbs the additive
polynomial term.)  Combining inequalities (i') and (ii')
we obtain

\begin{itemize}
\item[(iii)]     $T(x) \le 2xT_{\jh}(x)$.
\end{itemize}

Our contribution is an inequality of the form

\begin{equation}   \label{eq:quasipoly1}
  T_{\jh}(x) \le q(x)T(4x/5),
\end{equation}
where $q(x)$ is a quasipolynomial function. Combining with
item (iii) we obtain 
\begin{equation}   \label{eq:quasipoly2}
  T(x) \le 2xq(x)T(4x/5) < q(x)^2T(4x/5)
\end{equation}
which resolves to $T(x)=q(x)^{O(\log x)}$, yielding
the desired quasipolynomial bound on $T(x)$.
\begin{definition}
We refer to $(G,\calB)$ as the \emph{Johnson case} if 
$G$ is a transitive group with a system $\calB$ of 
imprimitivity such that $G$ acts on $\calB$ as
a Johnson group $\sss_k^{(t)}$ or $\aaa_k^{(t)}$.
We refer to $k$ as the \emph{Johnson parameter}.
\end{definition}

To prove Eq.~\eqref{eq:quasipoly1}, we define a finer
complexity estimate that involves the Johnson parameter.

For real numbers $x\ge y\ge 5$, let $T_{\jh}(x,y)$ denote the
maximum cost of solving all Johnson cases with $n\le x$ and
Johnson parameter $\ell(x)\le k\le y$ for some specific 
polylogarithmic function $\ell(x)$.  For $y<\max\{5, \ell(x)\}$
we set $T(x,y)=0$.  We obtain recurrences of the form

\begin{itemize}
\item[(iv)] $T_{\jh}(x) = T_{\jh}(x,x)$
\item[(v)] $T_{\jh}(x,y)\le q_1(x)\left(T(4x/5)+T_{\jh}(x,0.9y)\right)$ 
\end{itemize}
where $q_1(x)$ is a quasipolynomial function.  An upper bound of the form \\
$T_{\jh}(x,y)\le T(4x/5)q_1(x)^{O(\log y)}$ follows,
hence Eq.~\eqref{eq:quasipoly1} with $q(x)=q_1(x)^{O(\log x)}$
and therefore
\begin{equation}   \label{eq:cost4}
                 T(x) = q_1(x)^{O(\log^2 x)}.
\end{equation}

Explanation of item (v): we shall either reduce the domain 
(window) size $n$ by a positive fraction, or reduce the 
Johnson parameter $k$ by a positive fraction while
not increasing $n$, at quasipolynomial multiplicative cost.
These reductions are covered under our concept of
``symmetry breaking.''

\section{Functors, canonical constructions}  \label{sec:functor}
It is critical that all our constructions be \emph{canonical}.  We
shall employ a considerable variety of constructions, so to define
canonicity for all of them at once, we find the language of
categories convenient.  (No ``category theory'' will be 
required, only the concept of categories and functors.)

The only type of category we consider will be \emph{Brandt
  groupoids}, \ie, categories in which every morphism is
invertible.  Our categories will be \emph{concrete}, \ie, the
objects $X$ have an \emph{underlying set} $\Box(X)$ and the
morphisms are mappings between the objects (bijections in our
case).  (Strictly speaking, $\Box$ is a functor from the given
category to {\sf Sets}.)  We assume $\Box$ is \emph{faithful}, \ie,
if objects $X$ and $Y$ have the same underlying set
$\Box(X)=\Box(Y)$ and the identity map on this set is a morphism
between $X$ and $Y$ then $X=Y$.  We refer to the elements of
$\Box(X)$ as the \emph{points} or the \emph{vertices} or the
\emph{elements} of $X$.  When using the term ``category,'' we shall
tacitly assume it is a concrete, faithful Brandt groupoid.  In
fact, we can limit ourselves to categories where all objects have
the same underlying set, so all morphisms are permutations.

We write $\iso(X,Y)$ for the set of $X\to Y$ morphisms 
and $\aut(X) = \iso(X,X)$.  For a category $\cal$ we write
$X\in\calC$ if $X$ is an object in $\calC$.

We shall consider categories of various types of relational
structures, including uniform hypergraphs, bipartite graphs with a
declared partition into first and second parts, partitions (\ie,
equivalence relations), any of these structures with colored
vertices and/or edges, and special subcategories of these such as
uniprimitive coherent configurations.  Three categories to be
referred to have self-explanatory names: $\sets$, $\coloredsets$,
$\partitionedsets$.  A group $G\le\sym(\Omega)$ defines the
category of $G$-isomorphisms of strings on the domain $\Omega$; the
natural notation for this category, the central object of study in
this paper, would seem to be ``{\sf $G$-Strings}.''

Given two categories $\calC$ and $\calD$, a mapping $F_o : \calC\to\calD$
is \emph{canonical} if it is the mapping of objects from a functor
$F : \calC\to\calD$.  For an object $X\in\calC$ we shall usually only 
describe the construction of the object $F(X)$; the assignment
of a morphism $F(f) : F(X)\to F(Y)$ to a morphism $f : X\to Y$
will usually be evident.  In such a case we refer to $F_o$
as a canonical assignment (or, most often, a canonical 
construction).  Canonical color refinement procedures
are examples of canonical constructions.

A \emph{canonical embedding} of objects from category $\calD$
into objects from category $\calC$ is a functor
$F :\calC\to\calD$ such that for every object $X\in\calC$
we have $\Box(F(X))\subseteq\Box(X)$ and for each morphism
$f : X\to Y$ the mapping $F(f) : F(X)\to F(Y)$ is the
restriction of $f$ to $\Box(F(X))$.

Thus, a \emph{canonical subset} of objects in $\calC$ 
is a canonical embedding of objects from the category
{\sf Sets} into the objects of $\calC$.  Note that the vertex
set of a canonically embedded object is a canonical subset.
If $F$ is a canonical embedding then the restriction of
$\aut(X)$ to $\Box(F(X))$ is a subgroup of $\aut(F(X))$.
In particular, a canonical subset of $\Box(X)$ is invariant
under $\aut(X)$.

We say that $F$ is a canonical embedding of objects from $\calD$
\emph{onto} objects from $\calC$ if $\Box(F(X))=\Box(X)$
for all $X\in\calC$.

A \emph{canonical vertex-coloring} of objects in $\calC$
is a canonical embedding of objects from {\sf ColoredSets}
\emph{onto} the objects of $\calC$ (all vertices receive
a color).  Similarly, a \emph{canonical partition} of 
objects in $\calC$ is a canonical embedding of objects from 
{\sf PartitionedSets} onto the objects of $\calC$
(all vertices belong to some block of the partition).

Finally, we would like to formalize the notion of \emph{canonicity
relative to an arbitrary choice}, such as individualization.  In
this case we consider a canonical set of objects; the objects
individually are not canonical.  Here is a possible definition.

\begin{definition}[Category of tuples]
Let $\calD$ be a category.  Let $\calE$ be a class of 
non-empty sets of objects from $\calD$ with the following properties:
\begin{enumerate}[(i)]
\item $X,X'\in\xxx\in\calE$ then $\Box(X)=\Box(X')$
\item   \label{item:tuples2}
      if $X, X'\in\xxx\in\calE$ and $Y\in\yyy\in\calE$ 
      and $f\in\iso(X,Y)$ then there exists $Y'\in\yyy$
      such that $f\in\iso(X',Y')$.
\end{enumerate}
Under these conditions we turn $\calE$ into a category as
follows: 
\begin{enumerate}[(a)]
\item for $\xxx\in\calE$ we set $\Box(\xxx)=\Box(X)$ for any $X\in\xxx$
\item for $\xxx,\yyy\in\calE$, we set 
      $\iso(\xxx,\yyy)=\bigcup\{\iso(X,Y)\mid X\in\xxx, Y\in\yyy\}$.
\end{enumerate}
\end{definition} 
\begin{proposition}
 $\calE$ is a category.
\end{proposition}
\begin{proof}
We need to show that the morphisms in $\calE$ are closed under composition.
Let $f\in\iso(\xxx,\yyy)$ and $g\in\iso(\yyy,\zzz)$.  We need to show that
$fg\in\iso(\xxx,\zzz)$.  By definition, there exist objects $X\in\xxx$
and $Y\in\yyy$ such that $f\in\iso(X,Y)$.  Now $g\in\iso(Y',Z')$
for some objects $Y'\in\yyy$ and $Z'\in\zzz$.  By 
assumption~\eqref{item:tuples2} there exists $Z\in\zzz$ such that
$g\in\iso(Y,Z)$.  Therefore $fg\in\iso(X,Z)\subseteq\iso(\xxx,\zzz)$.
\end{proof}

\begin{definition}[Reduction at multiplicative cost]
By a \emph{reduction of the isomorphism problem} for objects $X,Y\in\calC$
to objects in $\calD$ \emph{``at multiplicative cost $s$''} we mean
a functor $F : \calC\to\calE$ for some category $\calE$ of tuples of $\calD$
such that $|F(Y)|=s$.
\end{definition}

\begin{proposition}
If $F$ is a reduction of $\iso(X,Y)$ to $\calD$ as above then
for any $X'\in F(X)$ we have
\begin{equation}
    \iso(X,Y) =\bigcup \left\{F^{-1}(\iso(X',Y'))\mid Y'\in F(Y)\right\}.
\end{equation}
Moreover, the terms in this union are disjoint, and all the nonempty
terms have the same cardinality.
\end{proposition}
Note that $X'$ is fixed in this union and is chosen arbitrarily from 
$F(X)$.  
\begin{proof}
Clear.  
\end{proof}
So if $F$ and $F^{-1}$ are efficiently computable per item then the
cost of computing $\iso(X,Y)$ is essentially the cost of computing
$s$ instances of computing $\iso(X',Y')$ in $\calD$, where $X'$ is
up to us to choose from $F(X)$.

\begin{definition}
Let $F$ be a reduction of the isomorphism problem in $\calC$
to $\calD$ at a multiplicative cost.  Consider the category
$\calD^F$ whose objects are the pairs $(X,X')$ where 
$X\in\calC$ and $X'\in F(X)$.  We set $\Box(X,X')=\Box(X)$
and $\iso((X,X'),(Y,Y'))=F^{-1}\iso(X',Y')$.
\end{definition}
\begin{proposition}
$\calC^F$ is a category.
\end{proposition}
\begin{definition}
Let $H:\calC^F\to\calH$ be a functor and let $(X,X')\in\calC^F$.
We say that $F(X,X')$ is \emph{canonically assigned to $X$
relative to $X'$.}
\end{definition}
An example of this procedure is individualization. Let $\calC$
have two objects, each of them a hypergraph.  Suppose we
individualize an ordered set of $t$ vertices of the hypergraph $X$;
we do the same with $Y$.  We consider the category $\calD$ of all
individualized versions of $X$ and $Y$.  The category $\calE$ will
have two objects, the set of individualized versions of $X$ and the
set of individualized versions of $Y$.  Suppose after some choice
$\vu=(u_1,\dots,u_t)$ of the ordered set of individualized vertices 
we find a canonically (in $\calC^F$) embedded large UPCC $U$ in $X$.
We then say that $U$ is canonical \emph{relative to $\vu$}.
For those $\vu$ for which the procedure does not work,
we embed the empty UPCC.  The multiplicative cost will be
$s=n(n-1)\dots(n-t+1)\le n^t$ where $n$ is the number of vertices of $X$.

But this type of argument will also occur when it cannot be phrased
in terms of individualizing vertices of an object; for instance,
we shall canonically construct other objects and individualize vertices
of those with similar effect.


\section{Breaking symmetry: colored partitions}

\subsection{Colored $\alpha$-partitions}

\begin{definition}   \label{def:partition}
A \emph{colored partition} of a set $\Omega$ is a coloring of the
elements of $\Omega$ along with a partition of each color class.
We say that this is a \emph{colored equipartition} if all blocks
within the same color class have equal size.
Given a colored partition $\Pi$, let $C_1,\dots,C_r$ be the color
classes and $\{B_{ij} \mid 1\le j\le k_i\}$ be the blocks of $C_i$.
We say that $\Pi$ is \emph{admissible} if for each color class
$C_i$ of size $|C_i|\ge 2$, all the blocks of $C_i$ have size
$|B_{ij}|\ge 2$.  ($B_{ij}=C_i$ is permitted.)  Let
$\rho(\Pi)=\max_{i,j}|B_{ij}|$.  For $0<\alpha\le 1$, a
\emph{colored $\alpha$-partition} is an admissible colored
partition $\Pi$ such that $\rho(\Pi)\le\alpha n$ where
$n=|\Omega|$.  
\end{definition}

The category {\sf ColoredPartitions} has as its objects sets with a
colored partition.  The morphisms are the bijection that
preserve color and preserve
the given equivalence relation (partition) in each color class.

\begin{definition}  \label{def:canonical-partition}
A \emph{canonical colored partition} of objects of a category $\calC$
is a canonical embedding of objects from the category 
{\sf ColoredPartitions} onto the objects of $\calC$.
\end{definition}
In other words this means assigning a colored partition of the
vertex set of each object in $\calC$ such that isomorphisms in
$\calC$ preserve colors and preserve the equivalence relation
on each color class.

\begin{proposition}  \label{prop:equipartition1}
Given a colored partition, one can canonically refine it to a
colored equipartition.  Here refinement means refining the colors;
the blocks will not change, so if the partition was admissibe, it 
remains admissible.
\end{proposition}
\begin{proof}
Encode the size of each block in the color of its elements.
\end{proof}


Finding canonical colored $4/5$-partitions will be one of our
key indicators of progress.

\begin{observation}
Let $\alpha\ge 1/2$.  A colored equipartition is an $\alpha$-partition if
either each color class has size $\le \alpha n$, or the unique color-class
of size $>n/2$ (the ``dominant color class'') is nontrivially
partitioned (at leat least two blocks, the blocks have size $\ge 2$).
\end{observation}

\subsection{Effect of coloring on $t$-tuples}
\label{sec:effect}
Let $\Gamma$ be a set and $\Phi=\binom{\Gamma}{t}$ the set of
$t$-tuples of $\Gamma$. Let $|\Gamma|=m$; so $|\Phi|=\binom{m}{t}$.
We shall need to examine the effect of a coloring of $\Gamma$ on $\Phi$.
This will be used repeatedly in Section~\ref{sec:structure-discovered}.

\begin{lemma}    \label{lem:effect-on-tuples}
Let $\Gamma$ be the disjoint union of color classes $\Delta_1,\dots,\Delta_k$.
This induces a canonical coloring of $\Phi=\binom{\Gamma}{t}$ as follows:
the color of $T\in\binom{\Gamma}{t}$ is the vector
$\left(\left. |T\cap\Delta_i| \,\right|\, 1\le i\le k\right)$.  
Then 
\begin{enumerate}[(a)]
\item   \label{item:effect-a}
the size of each color
class in $\Phi$ is $\le (2/3)|\Phi|$ with the possible exception of
one of the $k$ sets $\binom{\Delta_i}{t}$. 
\item  \label{item:effect-b}
$|\binom{\Delta_i}{t}|/|\Phi|\le (|\Delta_i|/m)^t$.
\end{enumerate}
\end{lemma}
\begin{proof}
Item~\eqref{item:effect-b} is trivial.  We prove item~\eqref{item:effect-a}
by induction on $k$.  The statement is vacuously true for $k=1$.
The case $k=2$ is the content of Prop.~\ref{prop:binom-ineq} below
with $m_i=|\Delta_i|$ and $t_i=|T\cap\Delta_i|$.
Let $k\ge 3$ and let $\Gamma'=\Delta_{k-1}\cup\Delta_k$.
Apply the inductive hypothesis to the coloring
$(\Delta_1,\dots,\Delta_{k-2},\Gamma')$ of $\Gamma$.
We are done except that we need to consider the color classes
included in $\binom{\Gamma'}{t}$.  But applying the case
$k=2$ we see that all of those color classes have size
$\le (2/3)\binom{|\Gamma'|}{t}< (2/3)|\Phi|$
with the possible exception of 
the two sets $\binom{\Delta_i}{t}$ for $i=k-1, k$.
\end{proof}

\begin{corollary}   \label{cor:effect-on-tuples2}
We use the notation of Lemma~\ref{lem:effect-on-tuples}.
Let $\alpha<1$ and $t\ge 2$.  Then any $\alpha$-coloring of $\Gamma$
(every color class has size $\le \alpha|\Gamma|$)
induces a $\max(2/3,\alpha)$-coloring of 
$\binom{\Gamma}{t}$.
\end{corollary}
\begin{proof}
Combine the two conclusions in Lemma~\ref{lem:effect-on-tuples}.
\end{proof}

\subsubsection{A binomial inequality}

\begin{proposition}   \label{prop:binom-ineq}
Let $m_1,m_2,t_1,t_2$ be integers; let $m=m_1+m_2$
and $t=t_1+t_2$.  Assume $t\le m/2$ and $t_i\ge 1$
for $i=1,2$.  Then
\begin{equation}
     \binom{m_1}{t_1}\binom{m_2}{t_2}\le \frac{2}{3}\binom{m}{t} .
\end{equation}
\end{proposition}
We first make the following observation.
\begin{claim}    \label{claim:binom}
Let $1\le k\le n-1$.  Then 
\begin{equation}
   \binom{n}{k}^2 \le 4\binom{n}{k-1}\binom{n}{k+1} .
\end{equation}
\end{claim}
\begin{proof}
Expanding and simplifying, the Claim reduces to the statement
\begin{equation}
    \frac{k+1}{k}\le 4\cdot\frac{n-k}{n-k+1}.
\end{equation}
This is true because $(k+1)/k\le 2$ and $(n-k)/(n-k+1)\ge 1/2$.
\end{proof}

\begin{proof}[Proof of Prop.~\ref{prop:binom-ineq}]
By Claim~\ref{claim:binom}, if $1\le t_i\le m_i-1$ then we have
\begin{equation}    \label{eq:binom1}
    \binom{m_i}{t_i}^2 \le 4 \binom{m_i}{t_i-1}\binom{m_i}{t_i+1} .
\end{equation} 
Let $a_s = \binom{m_1}{s}\binom{m_2}{t-s}.$  Then, if
$1\le s \le m_i-1$ and $1\le t-s\le m_2-1$, 
multiplying Eq.~\eqref{eq:binom1} for $i=1,2$ and
substituting $t_1=s$ and $t_2=t-s$, we obtain
\begin{equation}
 a_s^2 \le 16 a_{s-1} a_{s+1} \le 4(a_{s-1}+a_{s+1})^2 
\end{equation}
and therefore $a_s \le 2(a_{s-1}+a_{s+1})$.
Observe that $\sum_{s=0}^t a_s=\binom{m}{t}.$
It follows that under the conditions 
$1\le s \le m_i-1$ and $1\le t-s\le m_2-1$
we have
$(3/2)a_{s}\le a_{s-1}+a_{s}+a_{s+1} \le \binom{m}{t}$,
hence $a_s \le (2/3) \binom{m}{t}$, as desired.

It remains to consider the cases when $t_i=m_i$ for $i=1$ or $2$.
Let us say $i=1$, so $t_1=m_1$.  So we have
\begin{equation}
\binom{m_1}{t_1}\binom{m_2}{t_2}= \binom{m_2}{t_2}
  \le \binom{m-1}{t_2} \le \binom{m-1}{t-1}=\frac{t}{m}\binom{m}{t}
  \le \frac{1}{2}\binom{m}{t} .
\end{equation}
\end{proof}

This inequality will be used many times in the analysis of our algorithms;
we shall refer to it each time we find a canonical coloring of our
set $\Gamma$.

We highlight a corollary that will be used in Case 2 in
Sec.~\ref{sec:subroutines}.

\begin{corollary}   \label{cor:binom-ineq}
Let $r\ge 1$, $t\ge 1$ and $m\ge 2t$.
Then $\binom{m}{t}^r \le \left(\frac{2}{3}\right)^{r-1}\binom{mr}{tr}$.
\end{corollary}
\begin{proof}
By induction on $r$, using Prop.~\ref{prop:binom-ineq}.
\end{proof}

\section{Breaking symmetry: the Design Lemma}
\label{sec:tools1-designlemma}
In this section we describe the first of two combinatorial
symmetry-breaking tools, the reduction of a canonical $k$-ary
relational structure to binary ($k=2$).

Given a relational structure $\xxx=(\Omega,\calR)$ with
non-negligible symmetry defect (see Def.~\ref{def:defect}), we wish
to efficiently find a subgroup $G\le\sym(\Omega)$ such that $G$ is
substantially smaller than $\sym(\Omega)$ such that $\aut(\xxx)\le G$. 
We are not able to achieve this, but we do achieve it after
individualizing a small number of vertices.  We divide the task
into two parts: first we reduce the general case of $k$-ary relational
structures to UPCCs (uniprimitive coherent configurations -- recall that
these are binary relational structures ($k=2$)) 
(the ``Design Lemma,'' Section~\ref{sec:designlemma}), and, second,
we solve the problem for UPCCs (Section~\ref{sec:johnson}).

\subsection{The Design Lemma: reducing $k$-ary relations to binary}
\label{sec:designlemma}
In this section we prove one of the main technical results of the paper.

\begin{theorem}[Design lemma]  \label{thm:design}
Let $1/2 \le \alpha <1$ be a threshold parameter.  Let
$\xxx=(\Omega,\calR)$ be a $k$-ary relational structure with
$n=|\Omega|$ vertices, $2\le k\le n/2$, and relative strong
symmetry defect $\ge 1-\alpha$.  Then in time $n^{O(k)}$ we can
find a sequence $S$ of at most $k-1$ vertices such that by
individualizing each element of $S$
we can find either
\begin{enumerate}[(a)]
\item a canonical (relative to $S$)
colored $\alpha$-partition of the vertex set, or 
\item a canonically (relative to $S$) embedded uniprimitive
  coherent configuration $\xxx^*$ on some set $W\subseteq\Omega$
  of vertices of size $|W|\ge \alpha n$.
\end{enumerate}
\end{theorem}
\begin{observation}
Let $\DL(\alpha)$ be the statement of the Design Lemma for a particular 
$\alpha\ge 1/2$.  If $1 >\alpha' \ge \alpha\ge 1/2$ then 
$\DL(\alpha')$ follows from $\DL(\alpha)$.
\end{observation}
\begin{proof}
Assume $\DL(\alpha)$ holds.  Let $U\subseteq\Omega$ be a largest
strong symmetric subset of $\Omega$; let $\beta=|U|/n$.
Assume $\beta\le \alpha'$ so the assumption of $\DL(\alpha')$ holds.

\mn
Case 1.\ $\beta\le \alpha$.

\mn In this case we can apply $\DL(\alpha)$.  If $\DL(\alpha)$
returns case (a) (a colored $\alpha$-partition, canonical with respect
to a set $S\subseteq\Omega$ with $|S|\le k-1$), we are done (case (a)
holds for $\DL(\alpha')$) because an
$\alpha$-partition is also an $\alpha'$-partition.  If $\DL(\alpha)$
returns case (b) (a certain set $W$ with $|W|\ge\alpha n$, canonical
with respect to $S$) then we are done (case (b)) if $|W|\ge\alpha'n$.  
If $\alpha n\le |W|<\alpha' n$ then the coloring $(W,\Omega\setminus W)$ 
is an $\alpha$-coloring (since $\alpha\ge 1/2$), 
and therefore an $\alpha'$-coloring, so we return case (a) for 
$\DL(\alpha')$.

\mn
Case 2.\ $\alpha < \beta\le \alpha'$.

\mn
In this case the coloring $(U,\Omega\setminus U)$ 
is an $\alpha'$-coloring, so we return case (a) for $\DL(\alpha')$.
\end{proof}

It follows that it would suffice to prove the Design Lemma for
$\alpha = 1/2$.

\begin{remark}   \label{rem:sunwilmes}
If we can compute $\aut(\xxx^*)$ then we achieve a major reduction
in $\aut(\xxx)$ because 
$|\aut(\xxx^*)|\le\exp(\wto(\sqrt{n}))$~\cite{uniprimitive}.

There are two ways to compute $\aut(\xxx^*)$: either directly
or recursively.  

Direct computation of $\aut(\xxx^*)$ can be done in
$\exp(\wto(n^{1/3}))$ (Sun--Wilmes~\cite{sun-wilmes}).
Using this result would yield an overall 
$\exp(\wto(n^{1/3}))$ GI test, sufficient to break
the decades-old $\exp(\wto(\sqrt{n}))$ barrier.
\end{remark}

\begin{notation}
Let $\xxx_{(S)}$ denote the $k$-ary coherent configuration obtained
from the $k$-ary relational structure $\xxx$ by individualizing each
element of $S$ and applying $k$-dimensional WL refinement.
\end{notation}

\mn
{\sf Procedure Split-or-UPCC}
\begin{tabbing} m \= m \= m \= m \= m \= m \=m  \kill \\
01 \>\ \bfor\ $S\subset\Omega$,\ $|S|\le k-1$
  (in non-decreasing order of $|S|$)  \\  
02 \>\> \bif\ each vertex-color class in $\xxx_{(S)}$ has size 
        $\le \alpha n$ \\
03 \>\>\>  \bthen\ \breturn\ the colored partition, \bexit\ 
    (: goal (a) achieved :) \\ 
04 \>\> \belse\ (: we have a vertex-color class $C(S)$ of size $>\alpha n$ :)\\
05 \>\>\> let $\xxx^*(S)$ denote the substructure 
     of the 2-skeleton $\xxx_{(S)}^{(2)}$ induced  on $C(S)$ \\
   \>\>\>\>\>\>   (see Defs.~\ref{def:skeleton},~\ref{def:induced})\\
06 \>\>\>\> (: $\xxx^*(S)$ is a 
        homogeneous classical coherent configuration :) \\
07 \>\>\> \bif\ $\xxx^*(S)$ is imprimitive \\
08 \>\>\>\> split $C(S)$ into the connected components of a disconnected 
       off-diagonal constituent \\
09 \>\>\>\> \breturn\ colored partition and blocks of $C(S)$, 
        \bexit\ (: goal (a) achieved :) \\
10 \>\>\> \belse\ (: now $\xxx^*(S)$ is primitive :) \\
11 \>\>\>\> \bif\ $\xxx^*(S)$ is uniprimitive,
       \breturn\ $\xxx^*(S)$, \bexit\ (: goal (b) achieved :) \\
\end{tabbing}

\begin{theorem}
Under the conditions of Theorem~\ref{thm:design},
{\sf Procedure Split-or-UPCC} terminates, achieving goals (a) or (b).
\end{theorem}

Note that if the configuration $\xxx^*(S)$ obtained on line 05
is a clique configuration then the procedure discards the
current $S$ and moves to the next $S$.
 
What the Theorem asserts is that this will not always be the case;
for some $S$ with $|S|\le k-1$, if we pass line 04 then
the configurations $\xxx^*(S)$ is not a clique configuration.

The proof relies on Fisher's inequality on 
block designs (see ``Case 1'' below).
\begin{proof}
Unless we succeed already for $S=\emptyset$ on line 02,
we have a (unique) color-class $C:=C(\emptyset)$ of size $> \alpha n$
in $\xxx_{(\emptyset)}$.  
Let $\Cbar=\Omega\setminus C$.

We can in fact assume $|C|> \alpha n +(k-1)$ since otherwise 
we succed on line 02 by individualizing any subset $S\subseteq C$
of size $k-1$.

The classical coherent configuration $\xxx^*(\emptyset)$ (induced
by the 2-skeleton $\xxx^{(2)}_{(\emptyset)}$ on the vertex set $C$)
is homogeneous (all its vertices have the same color).  If it is
imprimitive then we are done on line 09.  If it is uniprimitive, we
are done on line 11.

Henceforth we assume $\xxx^*(\emptyset)$ is primitive but not
uniprimitive, \ie, it is the clique configuration. In other words,
all ordered pairs of distinct elements in $C$ have the same color in
$\xxx^*(\emptyset)$.

\begin{claim}
No transposition of the form $\tau=(x,y)$ ($x, y\in C$)
belongs to $\aut(\xxx)$.
\end{claim}
\begin{proof}   
By coherence, the color of the pair $(x,y)$ is ``aware'' of whether or
not $\tau\in\aut(\xxx)$ (Cor.~\ref{cor:st-aware}).
(Alternatively, we could explicitly include this information in the
color of the pair before refinement -- but this is not necessary.)
But this means if one such transposition belongs to $\aut(\xxx)$ then
all do, so $\sym(C)\le \aut(\xxx)$, contradicting the assumption that
the relative strong symmetry defect of $\xxx$ is $\ge 1-\alpha$.
\end{proof}

Given $S\subseteq \Omega$, assume we are past line 04.  
So we have a vertex-color class $C(S)$
in $\xxx_{(S)}$ of size greater than $\alpha n$.  Since the vertex
coloring of $\xxx_{(S)}$ is a refinement of the vertex-coloring of 
$\xxx_{(\emptyset)}$ and $\alpha\ge 1/2$, we must have $C(S)\subseteq C$.
Since $S\cap C(S)=\emptyset$, we in fact have $C(S)\subseteq C\setminus S$.

Let now $S$ denote a smallest subset of $\Omega$ such that 
$C(S)\neq C\setminus S$.
Note that $S\neq\emptyset$; 
we need to prove that such a subset exists at all.

\begin{claim}  $S$ exists and $|S|\le k-1$.
\end{claim}
\begin{proof}
Let $x,y\in C\setminus S$, $x\neq y$.  Since the transposition $\tau=(x,y)$
does not belong to $\aut(\xxx)$, there exist $i$ and $\vz\in R_i$ such that
$\vz^{\,\tau}\notin R_i$.  Let $\vz=(z_1,\dots,z_k)$ and let 
$Z=\{z_1,\dots,z_k\}$.  Individualizing each vertex in 
$Z\setminus\{x,y\}$ splits $x$ from $y$.  
\end{proof}

If $|C(S)|\le\alpha n$, we succeed on line 02 since individualizing
$S$ splits $C$ into relative canonical subsets of size $\le\alpha n$ 
and $\Cbar$ is canonical and small ($|\Cbar|\le (1-\alpha)n$).
Assume now that $|C(S)|> \alpha n$.  

\begin{claim}
 Let $x\in S$ and $Q=S\setminus\{x\}$.  Then
 $\xxx^*_{(Q)}$ is not a clique. 
\end{claim}
\begin{proof}
By the minimality of $S$, the set $C\setminus Q$ is
a color class in $\xxx_{(Q)}$ and therefore in $\xxx^*_{(Q)}$. 
So the vertex set of $\xxx^*(Q)$ is $C\setminus Q = C\setminus S$.
Assume for a contradiction that $\xxx^*(Q)$ is a clique.

We break the situation
into two cases according to whether or not $x\in C$.

\sn
{\bf Case 1.}\  $x\notin C$.   

\sn

\sn
Let $B$ denote the vertex-color class of $x$ in $\xxx_{(Q)}$.  
For $y\in B$, let $M(y)=C(Q\cup \{y\})$
and $\Mbar(y)=C\setminus M(y)\setminus Q$.
Consider the hypergraph
$\calH=(C\setminus Q; \{\Mbar(y) : y\in B\})$.
This hypergraph is uniform (because all elements of $B$ ``look the
same'' for the point of view of $\xxx_{(Q)}$) 
and regular (because all elements of $C\setminus Q$ ``look the 
same'' for the point of view of $\xxx_{(Q)}$)\footnote{An alternative 
to proving these statements 
would be to explicitly include the size of $\Mbar(x)$ in the
color of $y\in B$ and the $\calH$-degree of $z\in C\setminus Q$
in the color of $z$.  But $k$-dim WL is automatically ``aware''
of these quantities.}  Moreover, $\calH$ is a
block design (BIBD), \ie, every pair $\{u,v\}\subset C\setminus Q$
belongs to the same number of sets $\Mbar(y)$, because of
our current assumption that $\xxx_{(Q)}^*$ is a clique
(all pairs of vertices in $C(Q)=C\setminus Q$ have the same
color in $\xxx_{(Q)}^*$) (cf. the previous footnote).
 
But Fisher's inequality asserts that a BIBD has at least as many
blocks as it has vertices, hence $|B|\ge |C\setminus Q|$,
a contradiction because $|B|< (1-\alpha)n < n/2$ and
$|C\setminus Q| > \alpha n > n/2$
This shows that Case~1 cannot occur.

\sn
{\bf Case 2.}\ $S\subset C$.

\sn 
For $y\in C\setminus Q$ let $M(y)=C(Q\cup \{y\})$
and \\ $\Mbar(y)=C\setminus (M(y)\cup Q\cup \{y\})$.
The value of $|M(y)|$ does not depend on $y$
because of the homogeneity of $C\setminus Q$ in $\xxx_{(Q)}$.
So $|M(y)|=|M(x)|=|C(S)|> \alpha n$ for every $y\in C\setminus Q$
and therefore $d:=|\Mbar(y)|<(1-\alpha)n$.
(This quantity also does not depend on $y$.)

Consider the following digraph $X(Q)$ on the vertex set $C\setminus Q$:
introduce the edge $y\to z$ if $z\in \Mbar(y)$.  This is
a $d$-biregular digraph.  
(Both the in-degrees and the out-degrees are equal because of 
the homogeneity of $C\setminus Q$ in $\xxx_{(Q)}$.)  

Compare $X(Q)$ with the coherent configuration $\xxx^*(Q)$.
They have the same set of vertices, $C\setminus Q$.   
Note that $|C\setminus Q| > |C\setminus S| > \alpha n$.
Moreover, $k$-dimensional WL is ``aware'' of $X(Q)$ (because $|Q|\le k-2$),
so $\xxx^*(Q)$ is a refinement of $X(Q)$ and therefore
$\xxx^*(Q)$ has rank $\ge 3$.  (Again, if not convinced,
see the previous footnote: we could explicitly include the
``edge/non-edge of $X(Q)$'' information in the color of
the pairs in $C\setminus Q$.)  So if $\xxx^*(Q)$ 
cannot be a clique; Case~2 cannot occur.
\end{proof}
This completes the proof of the Design Lemma.
\end{proof}

\subsection{Local asymmetry to global irregularity: local guides}
\label{sec:localguide}

In this section we describe the way the input to the Design Lemma will
arise at least twice in this paper:
in Case B1 in Sec.~\ref{sec:localguides-johnson} 
(the case of Johnson graphs) and
and Case 2c of the proof of Thm 10.14 
(aggregation of non-fullness certificates).

Let $|\Omega_1|=|\Omega_2|=n$ and let $\ellk \le n/4$.  Typically,
$\ellk$ will be polylogarithmic; $\ellk$ is our ``locality parameter.''

Recall that by ``categories'' we mean concrete (every object has an
underlying set, morphisms are mappings) Brandt groupoids
(every morphism is invertible -- an isomorphism) (see Sec.~\ref{sec:functor}).

Let $\calL$ be a category with $2\binom{n}{\ellk}$ objects, namely,
an objects $X_i(L)$ for every $L\in\binom{\Omega_i}{\ellk}$ $(i=1,2)$.
The underlying set of these objects is $L$, \ie, $\Box(X_i(L))=L$.

Let $\calC$ be a category with two objects, $\xxx_1$ and $\xxx_2$,
with underlying sets $\Omega_1$ and $\Omega_2$, respectively 
($\Box(\xxx_i)=\Omega_i$).

\begin{definition}
We say that $\calL$ is a \emph{$\ellk$-local guide} for $\calC$
if for every morphism $f : \xxx_i\to \xxx_j$ ($i,j\in\{1,2\}$),
the restriction of $f$ to any $L\in\binom{\Omega_i}{\ellk}$ is
a morphism $X_i(L)\to X_j(L^f)$.
\end{definition}

\begin{definition}
Let us say that the set $L\in\binom{\Omega_i}{\ellk}$ is \emph{full}
for index $i$ if 
$\aut(X_i(L))=\sym(L)$. 
\end{definition}

\begin{proposition}[Local guide]  \label{prop:localguide} 
Let $\alpha$ be a threshold parameter, $3/4\le \alpha <1$.
Let $\calC$ be a category with two objects, $\xxx_1$ and $\xxx_2$,
with underlying sets $\Box(\xxx_i)=\Omega_i$ where 
$|\Omega_1|=|\Omega_2|=n$.  Let $3\le\ellk \le n/4$.
Let the category $\calL$ be a $\ellk$-local guide to the category $\calC$.
Assume further that none of the sets
$L\in\binom{\Omega_1}{\ellk}$ is full for $i=1$.
Then, in time $n^{O(k)}$ 
we can either refute isomorphism is $\xxx_1$ and $\xxx_2$
find a pair of canonical $k$-ary relational structures
$\xxx_1'$ and $\xxx_2'$ with $\Box(\xxx_i')=\Box(\xxx_i)$
such that the strong symmetry defect of each $\xxx_i'$ is
$\ge n-k+1$.
\end{proposition}
\begin{remark}
Note that the time requirement, $n^{O(k)}$, is polynomially
bounded in the length of the input which is at least 
$\binom{n}{k}$.  The procedure does not incur a
multiplicative cost.
\end{remark}

Now we can feed the output to the Design Lemma to obtain
the desired canonically embedded structures.

\begin{remark}
``Local asymmetry'' in the title refers to the sets $L$ not being full;
global irregularity refers to finding a canonical structure
($k$-ary relational structure, and then a
colored $\alpha$-partition or large UPCC) that drastically restricts
potential isomorphisms.
\end{remark}

\begin{proof}
     We shall introduce invariants of $\xxx_i$.  If we find any
     invariant that differs for $\xxx_1$ and $\xxx_2$, exit
     with rejection.  We shall not mention this explicitly below
     but simply assume that the two invariants are the same for  
     $\xxx_1$ and $\xxx_2$.

     In particular, we assume that none of the $L\in\binom{\Omega_2}{\ellk}$ 
     is full for $i=2$.

     For a set $A$, let $A^{\langle \ellk\rangle}$ denote the set
     of odered $\ellk$-tuples of distinct elements of $A$;
     so $A^{\langle \ellk\rangle}\subseteq A^{\ellk}$.
     Let $\calP_i=\Omega_i^{\langle \ellk\rangle}$ and
     $\calP=(\calP_1\times\{1\})\,\dot\cup\,(\calP_2\times\{2\})$.

     For $\vu=(u_1,\dots,u_{\ellk})\in\calP_i$ 
     let $L(\vu)=\{u_1,\dots,u_k\}\in\binom{\Omega_i}{\ellk}$.
     For $i,j\in\{1,2\}$ let $\vu=(u_1,\dots,u_k)\in\calP_i$ and
     $\vv=(v_1,\dots,v_k)\in\calP_j$.
     We say that the pairs $(\vu,i)$ and $(\vv,j)$ are
     equivalent, written as $(\vu,i)\sim (\vv,j)$,
     if there exists $\alpha\in\iso(X_i(L(\vu)),X_j(L(\vv)))$ 
     such that $\vu^{\,\alpha}=\vv$.
     The $\sim$ relation is an equivalence relation on
     $\cal P$, because $\calL$
     is a category (closed under composition of morphisms).

     Let $\Xi$ denote the set of $\sim$ equivalence classes.
     For $Q\in \Xi$, let $Q_i=\{\vu\mid (\vu,i)\in Q\}$.
     Since $\calL$ is a $\ellk$-local guide for $\calC$, it follows that
     the assignment $\xxx_i\mapsto Q_i$ is canonical (for $\calC$)
     for each $Q$.  Let $\zzz_i=(\Omega_i; Q_i\mid Q\in\Xi)$.
     These are $\ellk$-ary relational structures, canonically
     assigned to $\xxx_i$.

\mn
Claim.  The strong symmetry defect of $\zzz_i$ is at least $n-k+1$.
\begin{proof}
     Indeed, no subset $L\in \binom{\Omega_i}{\ellk}$ can be strongly
     symmetrical in $\zzz_i$ since $\aut(X_i(L))\neq \sym(L)$.
\end{proof}
This completes the proof of  
\end{proof}

\section{Split-or-Johnson}   \label{sec:johnson}

In this section we provide our second main combinatorial 
symmetry-breaking tool.  The output of the Design Lemma
was either a canonical colored $\alpha$-partition 
for, say, $\alpha=3/4$, or a canonically embedded
large UPCC.  In this section our algorithm takes a 
UPCC as input and attempts to find a canonical
colored $\alpha$-partition for, say, $\alpha=3/4$.

This is not always possible.
Johnson schemes are barriers to good partitions;
the Johnson scheme $\jjj(m,t)$ requires a
multiplicative cost of $\exp(\Omega(m/t))$
for a canonical $\alpha$-partition with any constant
$\alpha <1$ to arise.  This follows from
Prop.~\ref{prop:resilient} below.

Since $n=\binom{m}{t}$,
this cost is prohibitive: for bounded $t$ it results
in an exponential, $\exp(\Omega(n^{1/t}))$, algorithm.

We shall demonstrate that in a well--defined sense,
\emph{Johnson schemes are the only barriers.}  
Our algorithm takes a UPCC and 
returns, at a quasipolynomial multiplicative cost,
a canonical colored $3/4$-partition or a canonically
embedded Johnson scheme that takes up at least a $3/4$ 
fraction of the vertex set.

This cost is equivalent to the cost of individualizing a
polylogarithmic number of vertices, although this is not how it
happens.  Canonical auxiliary structres are constructed, and
vertices of those are individualized -- these could be called
``ideal vertices'' from the point of view of the input UPCC.

The bulk of the work is the same task -- find a good partition or
return a large Johnson scheme -- where the input is
an uneven bipartite graph with large symmetry defect.
We want to partition the large part, or find an embedded
Johnson scheme in it; so that part stays essentially constant,
while we iteratively reduce the small part.

\subsection{Resilience of Johnson schemes}
\label{sec:resilient}

Johnson schemes are highly resilient against partitioning.
Here is a formal statement of this observation.
\begin{proposition}  \label{prop:resilient}
Let $0<\epsilon \le 1/3$.  
The multiplicative cost of a (relative) canonical
$1-\epsilon$-partition of the Johnson scheme $\jjj(m,t)$
is $\ge (t/\epsilon)^{\epsilon m/t}$.
\end{proposition}
\begin{proof}
This is an immediate consequence of Corollary~\ref{cor:intact} below.
\end{proof}

\noindent
The following lemma says that if we try to break up a Johnson scheme
at moderate multiplicative cost, we fail badly; a 
large Johnson subscheme remains intact.
\begin{lemma}[Intact Johnson subscheme]   \label{lem:intact-subscheme}
Let $G\le\aut(\jjj(m,t))$.  Assume $m!/|G| < \binom{m}{r+1}$ for some 
$r< m/2-1$.  Then $G\ge \alt_{m-r}^{(t)}$ where $\alt_{m-r}^{(t)}$ 
acts on a $\jjj(m-r,t)$ subscheme of $\jjj(m,t)$ corresponding to
a subset of size $m-r$ of $[m]$.
\end{lemma}

\begin{proof}
Let us view $G$ as a subgroup of $\sss_m$, so 
$G^{(t)}\le \sss_m^{(t)}$ is the subgroup of $\aut(\jjj(m,t))$
in question.  Then, by the Jordan--Liebeck Theorem
(Thm.~\ref{thm:liebeck}) we have that $G\ge (\aaa_m)_{(T)}$
for some $T\subset [m]$,\ $|T|\le r$.  Let $\Gamma=[m]\setminus T$.
This means that $G^{(t)}\ge \alt^{(t)}(\Gamma)$
where $|\Gamma| \ge m-r$.
\end{proof}

\begin{corollary}     \label{cor:intact}
Let $G\le\aut(\jjj(m,t))$ and $0<\epsilon\le 1/3$.  If 
$m!/|G|< (t/\epsilon)^{\epsilon m/t}$ then
$G$ acts as a primitive group on a subset of
relative size $\ge (1-\epsilon)$.
\end{corollary}
\begin{proof}
The condition implies that $|\sss_m :G| < 1.9^m$.  
Let $r$ be the smallest value such that $|\sss_m:G|<\binom{m}{r+1}$.  
By Lemma~\ref{lem:intact-subscheme}, we have a Johnson group
$\alt_{m-r}^{(t)}\le G^{(t)}$ act on a subset of size $\binom{m-r}{t}$.
This group is primitive on this subset.
Now $\binom{m-r}{t}/\binom{m}{t}\ge (1-r/m)^t>1-(rt/m)$.
So we are done if $rt/m > \epsilon$.  Let us assume $rt/m\le \epsilon$.
Then
\begin{equation}
|\sss_m:G|\ge \binom{m}{r}\ge (m/r)^r=\left((m/r)^{r/m}\right)^m
 \ge (t/\epsilon)^{\epsilon m/t} ,
\end{equation}
contrary assumption.
\end{proof}

\begin{remark}
This result means that for fixed $t$ (e.\,g., $t=2$,
the most severe bottleneck case for decades),
the multiplicative cost of obtaining a constant-factor 
reduction in the domain size $n=\binom{m}{t}$
is exponential in $m$; and $m > n^{1/t}$.
\end{remark}

\subsection{Bipartite graphs: terminology, preliminary observations}

We use the term ``bipartite graph'' in the sense of having a declared
ordered bipartition of the vertex set.  
Let $X=(V_1,V_2;E)$ be a bipartite graph; here $E\subseteq V_1\times V_2$.
The vertex set is $V_1\cup V_2$; the $V_i$ are its ``parts.''
Isomorphisms of $X$ and $X'=(V_1',V_2';E')$ are bijections 
$V_1\cup V_2\to V_1'\cup V_2'$ that map $V_i$ to $V_i'$
and induce a bijection $E\to E'$.

We shall consider vertex-colored bipartite graphs $X=(V_1,V_2;E,f)$
where $f:V_1\cup V_2\to \{$colors$\}$; under this scenario, vertices
in the two parts do not share colors.  Isomorphisms preserve color
by definition.

The \emph{neighborhood} $N_X(v)$ (or $N(v)$ if $X$ is clear from the
context) of vertex $v\in V_i$ is the set of vertices adjacent
to $v$; so $N(v)\subseteq V_{3-i}$.  

Recall the definition of strong and weak twins (Def.~\ref{def:twins}).

\begin{definition}
Let $X=(V_1,V_2;E,f)$ be a colored bipartite graph.  We say that
vertices $x,y\in V_i$, $x\neq y$ are \emph{twins} if they have the
same color (in particular, they belong to the same part) and they have
the same neighborhood: $N(x)=N(y)$.  For a subset $T\subseteq V_i$
we use the phrase \emph{``all vertices in $T$ are twins''} to mean that 
all pairs of distinct vertices in $T$ are twins.
\end{definition}

It is clear that the ``twin-or-equal'' realtion is an equivalence
relation.

\begin{observation}
In a bipartite graph the following are equivalent for vertices $x,y$
($x\neq y$): 
\begin{enumerate}[(a)]
\item $x$ and $y$ are twins;
\item $x$ and $y$ are strong twins;
\item $x$ and $y$ are weak twins.
\end{enumerate}
\end{observation}
\begin{proof} Obvious.
\end{proof}

As a consequence, we don't need to make a distinction
between weak and strong ``symmetrical sets:''
a subset of $V_i$ is \emph{symmetrical} if
all vertices in $V_i$ are of the same color and
all of them have the same neighborhood.
The maximal symmetrical sets are the twin
equivalence classes.

\begin{definition}
Let $X=(V_1,V_2;E,f)$ be a vertex-colored bipartite graph.
We call a subset $T\subseteq V_i$
a \emph{symmetrical subset} of $V_i$ if 
$|T|\ge 2$ and all vertices in $T$ are twins.
\end{definition}

\begin{definition}
For $0\le\alpha\le 1$ we say that $X$ is \emph{$\alpha$-symmetrical}
in part $i$ if there is a symmetrical subset $T\subseteq V_i$ of size
$|T|\ge \alpha|V_i|$.
\end{definition}

\begin{definition}[Semiregular]
We say that the bipartite graph $X=(V_1,V_2;E)$ is \emph{semiregular}
if for $i=1,2$, each vertex in $V_i$ has the same degree.
\end{definition}

Recall that Prop.~\ref{prop:CCbipartite} asserts that each bipartite
edge-color class in a coherent configuration is semiregular.
This is especially useful for us in combination with the following
fact that is used to justify a subroutine in the main algorithm
in this section (Sec.~\ref{sec:johnson}), see Lemma~\ref{lem:contracting}.

\begin{proposition}[Semiregular defect]  \label{prop:semiregular-defect}
Let $X=(V_1,V_2;E)$ be a nontrivial (not empty and not complete)
semiregular bipartite graph. 
Then the symmetry defect of $V_1$ in $X$ is $\ge 1/2$.
\end{proposition}
\begin{proof}
By taking the complement if necessary, we may assume the density
of $Y$ is $|E|/(|V_1||V_2|)\le 1/2$; so every vertex in $V_i$ has
degree $\le |V_{3-i}|/2$.  Assume $S\subseteq V_i$ is symmetrical.
Let $x\in V_i$ be adjacent to $y\in V_{3-i}$.  But then 
$y$ is adjacent to all vertices if $S$, so $|S|\le\deg(y)\le |V_i|/2$.
\end{proof}

\subsection{Split-or-Johnson: the Extended Design Lemma}
In each result in this section, canonicity involves a combination 
of the following categories (cf. Section~\ref{sec:functor}): 
binary relational structures (Theorem~\ref{thm:UPCC}),
vertex-colored bipartite graphs (Theorem~\ref{thm:bipartite}),
$k$-ary relational structures (Theorem~\ref{thm:extended-design}),
and the category of colored partitions in each result.

Recall the definitions of \emph{canonical colored partition}
and an \emph{$\alpha$-partition} (Defs.~\ref{def:canonical-partition}
and~\ref{def:partition}).

Recall that ``UPCC'' means \emph{uniprimitive coherent configuration}
(Def.~\ref{def:UPCC}).

We can now state the two main results of Section~\ref{sec:johnson}.

\begin{theorem}[UPCC Split-or-Johnson]   \label{thm:UPCC}
Let $\xxx=(V;R_1,\dots,R_r)$ be a UPCC
with $n$ vertices and let $2/3 \le \beta<1$ be a threshold
parameter.  Then at quasipolynomial multiplicative cost we can
find either
\begin{enumerate}[(a)]
\item \label{item:CCpartition} a canonical colored $\beta$-partition of $V$, or
\item \label{item:CCjohnson}
a canonically embedded nontrivial Johnson scheme on a 
      subset of $V$ of size $\ge \beta n$.
\end{enumerate}
(The time bounds do not depend on $\beta$.)
\end{theorem}

\begin{theorem}[Bipartite Split-or-Johnson]   \label{thm:bipartite}
Let $X=(V_1,V_2;E,f)$ be a vertex-colored bipartite graph 
with $|V_1|\ge 2$
and let $2/3\le\alpha<1$ be a threshold parameter.  Assume 
$|V_2| < \alpha|V_1|$.  Assume moreover that
the symmetry defect of $X$ on $V_1$ is at least $1-\alpha$.
Then at quasipolynomial multiplicative cost we can find either 
\begin{enumerate}[(a)]
 \item \label{item:BPpartition}
       a canonical colored $\alpha$-partition of $V_1$, or
 \item \label{item:BPjohnson}
       a canonically embedded nontrivial Johnson scheme on a 
       subset of $V_1$ of size $\ge \alpha |V_1|$.
\end{enumerate}
(The time bounds do not depend on $\alpha$.)
\end{theorem}

These results will be proved recursively by mutual reduction to
each other.

Combining the Design Lemma and Theorem~\ref{thm:UPCC}
we obtain our overall combinatorial partitioning tool,
the main result of the combination of Sections~\ref{sec:tools1-designlemma}
and~\ref{sec:johnson}.

\begin{theorem}[Extended Design Lemma]   \label{thm:extended-design}
Let $3/4 \le \alpha <1$ be a threshold parameter.  Let
$\xxx=(\Omega,\calR)$ be a $k$-ary relational structure with $n$
vertices, $2\le k\le n/4$, and relative strong symmetry defect
$>1-\alpha$.  Then at a multiplicative cost of 
$q(n)n^{O(k)}$, where $q(n)$ is a quasipolynomial function, 
we can find either
\begin{enumerate}[(a)]
\item a canonical colored $\alpha$-partition of the vertex set, or 
\item a canonically embedded nontrivial Johnson scheme
  on a subset $W\subseteq\Omega$ of size $|W|\ge \alpha n$.
\end{enumerate}
(The time bounds do not depend on $\alpha$.)
\end{theorem}


\subsection{Minor subroutines}  \label{sec:minor}
First we describe a reduction of Theorem~\ref{thm:UPCC} 
to Theorem~\ref{thm:bipartite}.  The procedure will also serve
as a subroutine to the algorithm for Theorem~\ref{thm:bipartite}.

\begin{lemma}[UPCC-to-bipartite]  \label{lem:UPCC-to-bipartite}
Let $\xxx=(V;\calR)$ be a UPCC
with $n$ vertices and let $2/3 \le \beta \le 1$ be a threshold parameter.
Then at a multiplicative cost of\, $\le n$ and polynomial additive cost
one can either
\begin{enumerate}[(i)]
  \item achieve objective~\eqref{item:CCpartition} 
        of Theorem~\ref{thm:UPCC}, or
  \item reduce the given instance of Theorem~\ref{thm:UPCC} 
        to Theorem~\ref{thm:bipartite}
        by computig a threshold parameter $\alpha\ge 2/3$
        and a (relative) canonically embedded semiregular
        bipartite graph $X=(V_1,V_2;E)$ with $V_1\cup V_2\subseteq V$,\
        and $|V_1|\ge \beta n$
        such that a solution to
        each part of Theorem~\ref{thm:bipartite} for $X$
        is also a solution to the corresponding part of
        Theorem~\ref{thm:UPCC} for $\xxx$.
\end{enumerate}
\end{lemma}
\begin{proof}
Let $\xxx=(V;R_1,\dots,R_r)$ where $R_1=\diag(V)$ is the diagonal.
Let $d_i$ be the out-degree of the vertices in $R_i$; so $d_1=1$.
Pick a vertex $x\in V$.
Let $C_i=\{y\in V\mid (x,y)\in R_i\}$; so $|C_i|=d_i$.
Individualize $x$; this splits
$V$ into the (relative) canonical subsets $C_i$.
(See the definition of relative canonicity in Sec.~\ref{sec:functor}.)
If $d_i\le \beta n$ for all $i$, we are done
(objective~\eqref{item:CCpartition} has been achieved).

Assume now that (say) $d_2 > \beta n$; so $(V,R_2)$ is
an undirected graph (since $d_2\ge n/2$) and its complement has 
diameter 2 (\cite[Prop. 4.10]{uniprimitive}).  Let $(x,z)\in R_2$ and let 
$y\in V$ be such that $(x,y)\in R_i$ and 
$(z,y)\in R_j$ where $i,j\ge 3$.  Consider the
bipartite graph $X=(C_2,C_i;E)$ where $E=(C_2\times C_i)\cap R_j$.

$X$ is a semiregular (Prop.~\ref{prop:CCbipartite})
bipartite graph with $|C_2|>\beta n\ge 2n/3$
and therefore $|C_i|<n/3 <|C_2|/2$.
We have $E\neq\emptyset$ since $(z,y)\in E$.   The degree
of $y\in C_i$ is $d_j<n/3<2n/3\le d_2$ and therefore $E$ is not complete,
\ie, $E\neq C_2\times C_i$.   
It follows that in each part, the relative symmetry defect of $X$ 
is $\ge 1/2$ (Prop.~\ref{prop:semiregular-defect}).

Let now $\alpha = \beta n/d_2$.  So $\alpha > \beta \ge 2/3$.

If the relative symmetry defect of $X$ in $C_2$ is between $1/2$ and 
and $\alpha$ then we have a 
canonical colored $\beta$-partition of $V_1$
(the nontrivial twin equivalence classes of $X$, one block for the 
vertices in $C_2$ without twins, and one block $V_1\setminus C_2$).

Else, apply Theorem~\ref{thm:bipartite} to $X$ to obtain either
obtain a canonical colored $\alpha$-partition of $C_2$ (and thereby a
canonical colored $\beta$-partition of $V_1$ as above)
or the embedded nontrivial Johnson scheme of the required size.
\end{proof}

Our next routine takes a colored bipartite graph $X=(V_1,V_2;E)$
and helps make $V_2$ homogeneous.  Recall that we say that $x,y\in V_1$ are
\emph{twins} if the transposition $\tau=(x,y)$ is an automorphism
of $X$, \ie, if $x$ and $y$ have the same neighborhood.

\mn
{\sf Procedure Reduce-Part2-by-Color}
\begin{itemize}
\item[Input:] A threshold parameter $\alpha$, $2/3\le\alpha<1$ \\
       a colored bipartite graph $X=(V_1,V_2;E,f)$ where $|V_1|\ge 3$
       and $|V_2|<\alpha|V_1|$ such that there are no twins in $V_1$; \\
       a partition $V_2=C_1\cup C_2$ where each $C_j$ is a union of
       color classes.
\item[Output:] $j\in\{1,2\}$ such that in the induced colored bipartite
        subgraph $X_j=X[V_1,C_j]$ the symmetry defect of $V_1$
        is $\ge 1-\alpha$
\end{itemize}
The procedure computes the symmetry defect of $V_1$ in each $X_j$.

\begin{lemma}
In at least one of $X_1$ and $X_2$, the symmetry defect of $V_1$ is
at least $1-\alpha$.
\end{lemma}
\begin{proof}
Let $n_1=|V_1|$.  Assume for a contradiction that for $j=1,2$ there
exists a subset $D_j\subseteq V_1$ of size $|D_j|> \alpha n_1$ that is
symmetrical in $X_j$.  This means all
vertices in $D_j$ are twins with respect to $C_j$; therefore all
vertices in $D_1\cap D_2$ are twins in $X$.  Since $X$ has no twins,
we infer $|D_1\cap D_2|\le 1$.  But $|D_1\cap D_2| > (2\alpha-1)n_1
\ge n_1/3$, so $n_1\le 2$, a contradiction.
\end{proof}


\subsection{Bipartite Split-or-Johnson}  \label{sec:BPsplit-or-johnson}

In this section and the two subsequent sections we prove
Theorem~\ref{thm:bipartite}.

\begin{proof}
We use the notation of Theorem~\ref{thm:bipartite}.
Let $n_i=|V_i|$.  We view $X$ as a vertex-colored graph where the 
vertex-colors discriminate between $V_1$ and $V_2$.  
We may assume at all times that 
$|E|\le |V_1||V_2|/2$ (otherwise take the bipartite complement).
$E$ is not empty because of the positive symmetry defect assumption.

In this proof we say that $x,y\in V_1$ are \emph{twins} if the
transposition $\tau=(x,y)$ belongs to $\aut(X)$, \ie, $x$ and $y$
have the same neighborhood in $X$.   This is the same as being
\emph{strong twins} according to definition~\ref{def:twins}.
We note that in a bipartite graph, weak twins are automatically
strong twins (Prop.~\ref{prop:weak-is-strong}), so there is no need 
for this distinction.

Here is the algorithm.

\mn
{\sf Procedure Bipartite Split-or-Johnson}

\mn
\begin{tabbing} mmmmmm \= m \= m \= m \= m \kill \\
Input: \> a threshold parameter $3/4\le \alpha < 1$ \\
       \> a vertex-colored bipartite graph $X=(V_1,V_2;E,f)$ such that \\
       \>\> $|V_2| < \alpha|V_1|$ and the symmetry defect of $X$ on
              $V_1$ is at least $1-\alpha$ \\
Output: \> Output: item~\eqref{item:BPpartition} or~\eqref{item:BPjohnson}
        of Theorem~\ref{thm:bipartite}.
\end{tabbing}


\begin{enumerate}

\item If $n_1\le C_0$ for some absolute constant $C_0$,
individualize $(1-\alpha) n_1$ vertices of $V_1$, exit \\
(: objective~\eqref{item:BPpartition} achieved :)
\item If $n_2\le q(n_1)$ for some specific
quasipolynomial function $q$ then individualize
all vertices of $V_2$, apply naive vertex refinement, return
colored partition of $V_1$, exit \\
Claim.  This is a colored $\alpha$-partition.
\begin{proof}
All vertices of the same color in $V_1$ are twins.
\end{proof}
\item Apply WL refinement to $X$.  Let
$\xxx=(V;R_1,\dots,R_r)$ denote the resulting CC and let $\xxx_i$ be
the subconfiguration induced by $V_i$.  
Let $\xxx_3=(V_1,V_2; R_i \mid R_i\subseteq V_1\times V_2)$. \\
(: $\xxx_1$ and $\xxx_2$ are coherent; $\xxx_3$ is a refinement
of $E$ :)  \\
(: $\aut(X)=\aut(\xxx_3)$ because $\aut(X)\le\aut(\xxx_3)$
by the canonicity of $\xxx_3$; and $\aut(\xxx_3)\le\aut(X)$
because $\xxx_3$ is a refinement of $E$. :)
\item
If all color classes in $V_1$ have size $\le \alpha n_1$ then
return the colored partition of $V_1$, exit \\
(: objective~\eqref{item:BPpartition} achieved :)
\item 
Let $W_1\subseteq V_1$ be a color class such that $|W_1|>\alpha n_1$.
Update $\alpha\leftarrow \alpha |V_1|/|W_1|$,\
$V_1\leftarrow W_1$, $X$ and $\xxx_i$: the induced substructures
on $W_1\cup V_2$ \quad (: $\xxx_1$ is homogeneous :)
\item
If there are twins in $V_1$, let $\{C_1,\dots,C_k\}$ be the 
twin-equivalence classes.  \\
(: This is an equipartition because of the homogeneity of $\xxx_1$;
and each class has $\ge 2$ elements by definition. :)\\
Return this partition, exit. (: canonical colored $1/2$-partition found :) \\
(: To verify this, we need $k\ge 2$.  Indeed if $k=1$,
all elements of $V_1$ (previously $W_1$) are twins now,
but then they were twins before the update, contradicting 
the assumption on symmetry defect. :)

\item \label{item:nonhomogeneous} 
      (: There are no twins in $V_1$ :) \\
      If $\xxx_2$ is not homogeneous, partition $V_2$ as
      $V_2=C_1\dot\cup C_2$ where the $C_i$ are nonempty unions 
      of color classes.  Apply
      procedure {\sf Reduce-Part2-by-Color} (Sec.~\ref{sec:minor}).
      The procedure selects $j\in\{1,2\}$.  Update $X$ and
      $\xxx_i$ to their induced subconfigurations on 
      $V_1\cup C_i$ \\ 
      (: The symmetry defect of $\xxx$ on $V_1$ is $\ge 1/3 > 1-\alpha$ :) \\
\item  \label{item:cases}
      (: Both $\xxx_1$ and $\xxx_2$ are homogeneous and there are
         no twins in $V_1$ :) \\
      We need to consider the following cases:
      \begin{enumerate}[(i)]
      \item $\xxx_2$ is imprimitive: Section~\ref{sec:imprimitive} 
      \item $\xxx_2$ is primitive but not uniprimitive, \ie,
          $\xxx$ is the clique configuration (has rank 2): 
            ``block design case,'' Section~\ref{sec:blockdesign}
      \item $\xxx_2$ is uniprimitive but not known to be
             a Johnson scheme: Section~\ref{sec:UPCC9}
      \item \label{item:subJohnson}
            $\xxx_2$ is a Johnson scheme: Section~\ref{sec:BPjohnson}
      \end{enumerate}
\end{enumerate}
\end{proof}

\begin{remark}
An explanation to subitem (iii) of item~\eqref{item:cases}.  When a
UPCC is received, we could determine in polynomial time whether or not
it is a Johnson scheme, and if so, find an isomorphism to a Johnson
scheme.  But we don't investigate; the information that a given UPCC
is \emph{not} a Johnson scheme seems useless at this point (although
conjecturally it could be helpful\footnote{See the last question in
  Sec.~\ref{sec:conclusions}}.).  We only land in case (iv) when the
algorithm receives a UPCC explicitly labeled as a Johnson scheme.
This happens in Case 1 of the ``Block design case,''
Sec.~\ref{sec:blockdesign}.  In this case the embedded Johnson scheme
will be received along with an explicit isomorphism with some Johnson
scheme $\jjj(m,t)$.
\end{remark}

\subsection{Measures of progress}

Throughout the process, $n_1=|V_1|$ will not increase.
We say that a parameter $m$ is \emph{significantly reduced}
if $m_{\new} \le 0.9 m_{\old}$.
We deem to have made major progress if any of the following 
occurs:
\begin{itemize}
   \item $n_2$ is significantly reduced
   \item $\xxx_2$ moves from clique to UPCC while $n_2$ does not increase
   \item $\xxx_2$ moves from UPCC to Johnson scheme 
                while $n_2$ does not increase
\end{itemize}

\subsection{Imprimitive case}  \label{sec:imprimitive}

\begin{itemize}
\item[Case:] $\xxx_2$ is a homogeneous, imprimitive coherent configuration;
$\xxx_1$ is homogeneous, and there are no twins in $V_1$.
\end{itemize}
\begin{lemma}  \label{lem:imprimitive}
Under the assumptions of item~\ref{item:cases} in 
{\sf Procedure Bipartite Split-or-Johnson} we can
either return a canonical colored $1/2$-partition of $V_1$
at a multiplicative cost of $< n_2$, or
return, at only additive polynomial cost (no
multiplicative cost) a canonical bipartite graph 
$Y=(V_1,W_2;F)$ such that $|W_2|\le |V_2|/2$
such that the symmetry defect of $V_1$ in $Y$
is $\ge 1/2$.
\end{lemma}

Let $B_1,\dots,B_m$ be the connected components of a
disconnected non-diagonal color, say $R_2$.
The idea is either to replace $V_2$ by one of the blocks
(reducing $n_2$ to $n_2/m\le n_2/2$) or
to contract each block (reducing $n_2$ to $m\le n_2/2$),
significant progress in each case.  We shall see that one
of these is always possible without reducing the symmetry
defect on $V_1$ below 1/2.

\sn
Let $J = \{c(x,y)\mid x\in V_1, y\in V_2\}$.  Let $d_j$ be
the degree of $y\in V_2$ in color $R_j^-$.   ($d_j$ does not
depend on $y$ because of the homogeneity of $\xxx_2$.)
Note that $|J| \ge 2$ because the coloring of $V_1\times V_2$
is a refinement of $E$; so $d_j < n_1$ for all $j\in J$.

{\sf Procedure ImprimitiveCase}
\begin{enumerate}
\item If  $(\forall j\in J)(d_j\le n_1/2)$ then
      individualize some $x\in V_2$.   This splits
      $V_1$ into color classes of size $d_j$.  Return this
      partition of $V_1$, exit. \\
     (: canonical colored $1/2$-partition of $V_1$ found :)
\item \label{item:lowdegrees}
      else (: for some $j\in J$ we have $d_j >  n_1/2$ :) \\
      For $i=1,\dots, m$ let $Z_i=X(V_1,B_i;R_j)$.
  \begin{itemize}
  \item[(i)] if $(\exists i)($the symmetry defect of $V_1$ in $Z_i$ is 
           $\ge 1/2)$ then $Y\leftarrow Z_i$ \\
         (: This involves choosing $i$ at a multiplicative cost of $m$.
            The gain is a reduction $n_2\leftarrow n_2/m$ :)
  \item[(ii)]  \label{item:contract}
        (: the symmetry defect of $V_1$ in each $Z_i$ is 
             less than $1/2$ :) \\
        Let $h\in J$, $h\neq j$.  Let $Y=(V_1,[m];\Rbar_h)$ where
        $(x,i)\in\Rbar_h$ if $(\exists y\in B_i)((x,y)\in R_h)$ \\
        (: contracting each block, $n_2\leftarrow m$ :) 
  \end{itemize}
       return $Y$
\end{enumerate}

\begin{lemma}   \label{lem:contracting}
In subcase~(ii) of item~\ref{item:lowdegrees}
(contracting the blocks), $V_1$ has symmetry defect $\ge 1/2$ in 
the contracted bipartite graph $Y$.
\end{lemma}
\begin{proof}
$Y$ is semiregular by Cor.~\ref{cor:contract}. 
Moreover $\Rbar_h$ is not empty because $R_h$ is not empty.

\sn
Claim. $Y$ is not complete.
\begin{proof}
For each $i\le m$
there is a (unique) $Z_i$-twin equivalence class
$C_i\subseteq V_1$ such that $|C_i|>n_1/2$.

\sn
Subclaim.  $C_i\times B_i\subseteq R_j$.
\begin{proof}
The vertices of $C_i$ are twins in $Z_i$.
In other words, for each $x\in B_i$ the set
$C_i\times \{x\}$ is monochromatic
(has a single color), \ie, $C_i\times \{x\}\subseteq R_{\ell}$
for some $\ell\in J$.  It follows that $d_{\ell} > n_1/2$.
Therefore $\ell=j$, proving the Subclaim.
\end{proof}

Now $Y$ is not complete because it has no edge from $i$ to $C_i$.
\end{proof}

Since $Y$ is semiregular, nonempty and not complete, we infer by 
Prop.~\ref{prop:semiregular-defect} that $Y$ has symmetry defect 
$\ge 1/2$, as claimed.
\end{proof}
This also completes the proof of Lemma~\ref{lem:imprimitive}.

\subsection{Block design case} \label{sec:blockdesign}

\mn
Assumptions: no twins in $V_1$, 
$\xxx_2$ is the clique configuration (rank-2).

Let $\calH=(V_2,\calE)$ be the hypergraph of neighborhoods of
vertices in $V_1$.  This hypergraph has no multiple edges because
there are no twins in $V_1$.

\mn
{\bf Case 1.}\ $\calH$ is the complete $d_1$-uniform hypergraph.

\mn
In this case $V_1$ can be indentified with $V_1=\binom{V_2}{d_1}$,
the vertex set of a canonically embedded Johnson scheme,
achieving goal~\eqref{item:BPjohnson} of Theorem~\ref{thm:bipartite}.
Note that the vertices of this Johnson scheme (elements of $V_1$)
come labeled by the $d_1$-subsets of $V_2$.
With this we not only exit this routine but exit the main algorithm.

\mn
{\bf Case 2.}\ There is an $\calH$-twin equivalence class $C\subseteq V_2$
of size $|C|\ge n_2/2$.  (Note that the vertices of $C$ are not
necessarily twins in $X$.)

Apply procedure {\sf Reduce-Part2-by-Color} to the coloring
$(C,V_2\setminus C)$.  If $V_2\setminus C$ is selected, 
we have made significant progress (reduced $|V_2|$ by half).
If $C$ is selected, the vertices of $C$ continue to be 
twins in the reduced $\calH$ which brings us to Case 1,
terminating the main algorithm.

\mn
{\bf Case 3.}\ The relative symmetry defect of $\calH$ is $\ge 1/2$.
(Note that for hypergraphs we don't need to make the distinction
between strong and weak symmetry defect, Prop.~\ref{prop:weak-is-strong}.)

\mn
{\bf Case 3a.}\ $d_1\le (7/3) \log_2 n_1$.

Apply the Design lemma to $\calH$, viewed as a $d_1$-ary relational
structure.  (Multiplicative cost $n_2^{d_1} < n_2^{(7/3)\log_2 n_1}$).  

\mn
{\bf Case 3a1.}\ The Design lemma returns a 
  canonical colored $3/4$-partition of $V_2$.

\mn
Recompute $\xxx$ and thereby $\xxx_2$.  Colored $(3/4)$-partition on
$V_2$ persists (colors can only get refined).  Apply 
{\sf Reduce-Part2-by-Color} to the coloring.  If color class
selected is greater than $3n_2/4$, this color class is
equipartitioned; apply {\sf Procedure ImprimitiveCase}
(Sec.~\ref{sec:imprimitive}).   In either case, 
significant progress: $n_2$ reduced to $\le 3n_2/4$.

\mn
{\bf Case 3a2.}\ The Design lemma returns a UPCC\ $\yyy$ canonically
embedded on a subset $W\subseteq V_2$ with $|W|\ge (3/4)n_2$.

\mn
Apply {\sf Reduce-Part2-by-Color} to the partition $(W, V_2\setminus W)$.
If the procedure selects $V_2\setminus W$, significant progress
($n_2$ reduced to $\le n_2/4$).  If it selects $W$, go to
Sec.~\ref{sec:UPCC9}.

Let $U\subseteq V_2$
be the part selected, and $(\xxx_2)_{\new}$ the homogeneous coherent 
configuration obtained on $U$.  If $(\xxx_2)_{\new}$ is a UPCC,
exit, significant progress.

If $(\xxx_2)_{\new}$ is not a UPCC, \ie, it has rank 2, then $U$
was a clique in $\yyy$ and therefore $|U|\le |W|/2\le n_2/2$
by Prop.~\ref{prop:digraph-indep},
a significant reduction of $|V_2|$.



\mn
{\bf Case 3b.} $d_1 > 7/3\log_2 n_1$.

Let $t=\lceil(7/4)\log_2 n_1\rceil$.
So $t\le (3/4)d_1$.  
By Lemma~\ref{lem:skeleton-defect}, the symmetry defect of
the $t$-skeleton $\calH^{(t)}$ of $\calH$ is greater than $1/4$.
Let us apply the Design lemma to $\calH^{(t)}$.
The result is either a $3/4$-partition of $V_2$ (done, exit),

\mn
or a UPCC on a subset of $V_2$ of size $\ge (3/4)n_2$
(significant progress, exit).

\subsection{UPCC}    \label{sec:UPCC9}
Situation: $\xxx_2$ is uniprimitive, not known to be a Johnson scheme,
there are no twins in $V_1$

\mn
Apply the procedure of Lemma~\ref{lem:UPCC-to-bipartite}
      (``UPCC-to-Bipartite'') to $\xxx_2$ with $\alpha:=2/3$.
   \begin{enumerate}[(I)]
   \item If the algorithm of Lemma~\ref{lem:UPCC-to-bipartite} 
      returns a canonical colored $2/3$-partition of $V_2$,
      apply procedure {\sf Reduce-Part2-by-Color} (Sec.~\ref{sec:minor})
      to reduce $V_2$ to one of its color classes (:~significant progress~:)
   \item else (: the algorithm of Lemma~\ref{lem:UPCC-to-bipartite} 
      returns a canonically embedded nontrivial semiregular
      bipartite graph $X''=(V_1',V_2';E')$ with 
      $W:=V_1'\cup V_2'\subseteq V_2$
      and $|V_1'|\ge 2/3 |V_2|$. :)  (: Note that $|V_1'|<|V_1|$
      and $|V_2'|\le |V_2|/3$, a significantly smaller instance :)

      Recursively apply Theorem~\ref{thm:bipartite} to $X'$
      with threshold parameter $\alpha = 2/3$.  

      If a canonical colored $2/3$-partition of $V_1'$ is returned, 
      add $V_2'$ and $V_2\setminus W$, separately colored, to it,
      to obtain a canonical colored $2/3$-partition of $V_2$.
      Apply {\sf Reduce-Part2-by-Color} to make significant
      progress (reducing $n_2$ by a factor of $2/3$).
      
      Else (: the algorithm returns a nontrivial Johnson scheme
      on a subset $W$ of $V_2$ of size $\ge 2n_2/3$ :)\\
      apply procedure {\sf Reduce-Part2-by-Color}
      to the canonical 2-coloring
      $(W_2,V_2\setminus W_2)$ of $V_2$ \\
      if procedure {\sf Reduce-Part2-by-Color} selects $V_2\setminus W_2$, 
      update $V_2\leftarrow V_2\setminus W_2$, recurse \\
      (: significant progress: $n_2$ reduced by a factor of $3$ :) \\
      else (: procedure {\sf Reduce-Part2-by-Color} selects $W$ :)\\
      enter subcase (iv) 
      of item~\ref{item:cases}.
   \end{enumerate}

\subsection{Local to global symmetry}
In this section we make preparations to handle the case when
$\xxx_2$ is a Johnson scheme.

\begin{notation}
Let $\calH=(V,\calE)$ be a hypergraph and $x\in V$.
We denote the induced subhypergraph $\calH[V\setminus\{x\}]$
by $\calH - x$.
\end{notation}

\begin{lemma}[Exchange/augment]   \label{lem:stepsym}
Let $\calH=(V,\calE)$ be a $t$-uniform hypergraph. 
Let $V=A\dot\cup B$ where $|B|\ge 1$ and $|A| > (t+1)|B|$.
Assume $\sym(A)\le\aut(\calH)$.
Assume moreover that $(\forall x\in A)(\exists u\in B)$ such that
$\sym(A\cup\{u\}\setminus\{x\})\le\aut(\calH - x)$.  Then 
$(\exists v\in B)(\sym(A\cup\{v\})\le \aut(\calH))$.
\end{lemma}

\begin{proof}
For $x\in A$ let $f(x)$ denote an element $u\in B$ such that 
$\sym(A\cup\{u\}\setminus\{x\})\le\aut(\calH - x)$.  
By the pigeon-hole principle there exists $u\in B$ such that 
$|f^{-1}(u)| \ge t+2$.   We claim that $v:=u$ satisfies the conclusion 
of the lemma.

Let $\tau = (x,u)$ denote the transposition that swaps $x$ and $u$.
We need to show that $\tau\in \aut(\calH).$
I.\,e., we need to show that for all $E\in \calE$ we have $E^{\,\tau}\in\calE$.

Let $y\in f^{-1}(u)\setminus (E\cup\{x\})$.
So $\sym(A\cup\{u\}\setminus\{y\})\le \aut(\calH - y)$.
In particular, $\tau\in \aut(\calH - y)$ and therefore
$E^{\,\tau}\in \calE$.
\end{proof}
\begin{remark} It is easy to see that $|A|>|B|t$ would suffice 
in place of $|A|>|B|(t+1)$.
\end{remark}

The next lemma asserts that if every small induced subhypergraph
of a uniform hypergraph has high symmetry then the hypergraph
has extremely high symmetry.
\begin{lemma}[Local to global symmetry]    \label{lem:local-to-global}
Let $s\ge 0$, $t\ge 2$, and \\ $\ell=\max\{t+3, (t+2)(t+3)s\}$.
Let $\calH=(V,\calE)$ be a $t$-uniform hypergraph with $m\ge \ell$ vertices.
Assume that for every subset $L\subseteq V$ of size
$|L|=\ell$ 
the symmetry defect of the induced subhypergraph $\calH[L]$
is at most $s$.  
Then the symmetry defect of $\calH$ is at most $s$.
\end{lemma}
Note that this lemma talks about the \emph{absolute} (as opposed 
to relative) symmetry defect.

\begin{proof}
Let $k$ be the largest value such that the following holds:

\begin{itemize}
\item[$(C_k)$]  For every subset $K\subseteq V$ of size $|K|=k$
the symmetry defect of the induced subhypergraph $\calH[K]$
is at most $s$.  
\end{itemize}
The assumption is $(C_{\ell})$.  We need to show $(C_m)$.  
We prove $(C_k)$ for $k=\ell,\ell+1,\dots,m$ by induction.

Assume $(C_{k-1})$ for some $k>\ell$.  Let $|K|=k$.
For $x\in K$ set $K_x=K\setminus \{x\}$.  Let
$S_x\subset K_x$ denote a subset of size $|S_x|=s$ such that 
$\sym(K_x\setminus S_x)\le \aut(\calH[K_x])$.

Let $x_1,\dots,x_{t+3}$ be $t+3$ distinct elements of $K$.  
(These exist because $\ell\ge t+3$.)
Let $U=\bigcup_{i=1}^{t+3} S_{x_i}$; so $|U|\le (t+3)s < |K|$.

\mn
Claim 1.\ $\sym(K\setminus U)\le \aut(\calH[K])$.

\begin{proof}
Let $\tau=(u,v)$ be a transposition in $K\setminus U$ ($u,v\in K\setminus U$).
We need to show that $\tau\in\aut(\calH[K])$, \ie,
for all $E\in\calE[K]$ we need to show $E^{\,\tau}\in\calE[K])$.

Pick an $i$ such that $x_i\notin E\cup\{u,v\}$.  So $E\subset K_{x_i}$
and $\tau$ fixes $S_{x_i}$ pointwise (because $S_{x_i}\subseteq U$).
It follows that 
$\tau\in\aut(\calH[K_x])$ and therefore $E^{\,\tau}\in\calE[K]$,
completing the proof of Claim 1.
\end{proof} 

Let now $B$ be a minimal subset of $U$ such that 
$\sym(K\setminus B)\le \aut(\calH[K])$. 

\mn
Claim 2. \ $|B|\le s$.

\mn
\begin{proof}
Assume for a contradiction that $|B|\ge s+1$.  Let $A=K\setminus B$.

\mn
Claim 3.  The assumptions of Lemma~\ref{lem:stepsym} hold with $\calH[K]$
in the place of $\calH$.

\begin{proof}
We have $|A| > (t+1)|B|$ because $|B|\le |U|\le s(t+3)$ and
$|K|>\ell\ge (t+2)(t+3)s\ge (t+2)|B|$.

We have $\sym(A)\le \aut(\calH[K])$ by the definition of $B$.  
Let now $x\in A$.  We need to find $u\in B$ such that 
$\sym(A\cup\{u\}\setminus\{x\})\le \aut(\calH[K_x])$.
We claim that any $u\in B\setminus S_x$ will do.

By assumption, $\sym(K_x\setminus S_x)\le \aut(\calH[K_x])$.
Moreover, $\sym(K_x\setminus B)\le \aut(\calH[K_x])$
because $\sym(K\setminus B)\le \aut(\calH[K])$.
Now $\sym(K_x\setminus S_x)$ and $\sym(K_x\setminus B)$
generate $\sym(K_x\setminus (B\cap S_x))$ because
$K_x\setminus (B\cap S_x)$ is the union of
$K_x\setminus S_x$ and $K_x\setminus B$
and the intersection of these two sets is
nonempty (because $|S_x\cup B| \le (t+4)s < \ell \le k-1 = |K_x|$).

Finally we observe that 
$\sym(K_x\setminus (B\cap S_x))\ge \sym(A\cup\{u\}\setminus\{x\})$
(because $u\notin S_x$).
This completes the proof of Claim 3.
\end{proof}
Now by Lemma~\ref{lem:stepsym} we infer that for some $v\in B$ we have
$\sym(A\cup\{v\})\le\aut(\calH[K])$, contradicting the minimality
of $B$, completing the proof of Claim 2.
\end{proof}
This completes the proof of $(C_k)$, the inductive step in
the proof of Lemma~\ref{lem:local-to-global}.
\end{proof}

\subsection{Bipartite graph with Johnson scheme on small part}
   \label{sec:BPjohnson}    \label{feketeleves}
Situation: $X=(V_1,V_2;E)$ is a semiregular bipartite graph
with $|V_2|<|V_1|$, there are no twins in $V_1$,
and $\xxx_2$ is a Johnson scheme $\jjj(\Gamma,t)$ with $t\ge 2$
(so $V_2=\binom{\Gamma}{t}$).

We write $|\Gamma|=m$,\ $n_i=|V_i|$, and $n=n_1+n_2$.
So $n_2=\binom{m}{t}$ and therefore $t< 2\log n_2/\log m$.

Let $d_i$ be the degree of the vertices in $V_i$ in $X$
(so $d_i|V_i|=|E|$).

The goal is a canonical colored $\alpha$-partition of $V_1$.

We may assume the density of $X$ is $|E|/|V_1||V_2|\le 1/2$
(otherwise take the bipartite complement).

For $v\in V_1$ let $\calE(v)$ denote the set of neighbors of
$v$ in $X$ viewed as a set of $t$-subsets of $\Gamma$.
We call the $t$-uniform hypergraph $\calH(v)=(\Gamma,\calE(v))$ the 
\emph{neighborhood hypergraph} of $v$.

Color $v\in V_1$ by 
$|\calE(v)|$.  If no color-class is greater than $\alpha n_1$,
return this coloring, exit.  Otherwise let $C$ be the largest
color class (so $|C|>\alpha n_1$).  Let $\alpha'=\alpha |V_1|/|C|$.
Let $X'$ be the subgraph of $X$ induced by $(C,V_2)$. 
Update $X\leftarrow X'$ (in particular $V_1\leftarrow C$)
and $\alpha\leftarrow\alpha'$.  Do not update $\xxx_2$.
Note that there will still be no twins in $V_1$,
and $|V_2|$ continues to be less than $\alpha|V_1|$
(because the value $\alpha |V_1|$ has not changed).

Set $\ell = \max\{(\log_2 n_1)^2,  (\log_2 n_2)^3/\log_2\log_2 n_2\}$.  

Our algorithm will be recursive on $m=|\Gamma|$, so
a significant reduction of $m$ will count as
significant progress.

If $m\le 2\ell$,
individualize each element of $\Gamma$.  This immediately
individualizes each vertex in $V_2$ (each $t$-subset
of $\Gamma$ has a distinct set of colors), and thereby each
vertex in $V_1$ (since there are no twins in $V_1$).
Return the resulting canonical coloring of $V_1$, exit.

Assume now $m > 2\ell$.  

\subsubsection{Subroutines}  \label{sec:subroutines}
The algorithm will try to find canonical structures on $\Gamma$
(coloring, equipartition, $\ell$-ary relation, UPCC, Johnson scheme). 
First we describe subroutines how to proceed if such a 
structure is found.

\mn
{\bf Case 1:}\ A canonical coloring of $\Gamma$ is found.

\mn
Let $\Gamma_i$ denote the $i$-th color class, so
$\Gamma=\Gamma_1\dot\cup\dots\dot\cup\Gamma_r$ ($2\le r\le m$).
This results in a canonical coloring of $V_2$ as follows.  
Let $t=t_1+\dots +t_r$ ($t_i\ge 0$) be an ordered $r$-partition of 
the integer $t$.  Associate with this partition the set
$V_2(t_1,\dots,t_r)=\{T\in\binom{\Gamma}{t}\mid 
(\forall i)(|T\cap\Gamma_i|=t_i)\}$.  These sets will be the color 
classes of $V_2$.  Apply {\sf Reduce-Part2-by-Color} to this
coloring.  Let $V_2(t_1,\dots,t_r)$ be the color class selected
by {\sf Reduce-Part2-by-Color}.  If more than one of the $t_i$ 
satisfy $0<t_i<|\Gamma_i|$ then we have a nontrivial canonical partition
of $V_2(t_1,\dots,t_r)$ into $\binom{|\Gamma_i|}{t_i}$ blocks for
one of these values $i$ (each block being defined by the set $T\cap\Gamma_i$).
Apply {\sf Procedure ImprimitiveCase} to this partition.

In the remaining case, all but one of the $t_j$ are either $0$ or $|\Gamma_j|$.
Let $i$ denote the one exception, so the vertices of $V_2$ can now be
labeled by $\binom{\Gamma_i}{t_i}$.  We are back to the Johnson case,
having reduced $\Gamma$ to one of the $\Gamma_i$.  This is 
significant progress if the partition of $\Gamma$ we started from was 
good, say all $|\Gamma_i|\le 3m/4$.

\mn
{\bf Case 2.}\  A nontrivial canonical equipartition of $\Gamma$ is found.

Let $\Gamma_i$ denote the $i$-th block of the partition, so
$\Gamma=\Gamma_1\dot\cup\dots\dot\cup\Gamma_r$ ($2\le r\le m/2$).
This results in a canonical coloring of $V_2$ as follows:
Let $t=t_1+\dots +t_r$ ($t_i\ge 0$) be an unordered $r$-partition of 
$t$ ($t_1\ge t_2\ge\dots\ge t_r$).  Such an unordered partition
will correspond to those $T\in\binom{\Gamma}{t}$ for which
the multiset $\{|T\cap\Gamma_i|\mid 1\le i\le r\}$ is the same as
the multiset $\{t_1,\dots,t_r\}$.  Each of these color classes is
further canonically partitioned into blocks corresponding to
the ordered partitions (now the $t_i$ are not necessarily
nonincreasing); the blocks will be of the form
$V_2(t_1,\dots,t_r)$ as defined above.  We first apply
{\sf Reduce-Part2-by-Color} to the coloring of $V_2$, and then
apply {\sf Procedure ImprimitiveCase} to this partition.  

The case remaining is when the color class selected is not partitioned
into blocks.  This occurs if all the $t_i$ are equal: $t_i=t/r$.
The size of this color class is $\binom{m/r}{t/r}^r$ which is at most
$2/3$ of $\binom{m}{t}$, the original size of $V_2$, according to the 
Cor.~\ref{cor:binom-ineq}
(setting $t\leftarrow t/r$ and $m\leftarrow m/r$).  
Return the reduced $V_2$, exit (significant progress).

This completes the subroutine for Case 2.  

\mn
{\bf Case 3.}\ A canonical $\ell$-ary relational structure
on $\Gamma$ with $\ell\ge 3$ and
relative strong symmetry defect $\ge 1/4$ found.

\mn
In this case, apply the Design Lemma (multiplicative cost
$m^{O(\ell)}$).  If a canonical colored $3/4$-partition is returned,
apply Cases 1 and~2 for significant progress in either
reducing $|V_2|$ or reducing to the Johnson case with
significantly reduced $|\Gamma|$.

If the Design Lemma returns a UPCC canonically 
embedded in $\Gamma$ on a set $W\subseteq \Gamma$
with $|W|\ge 3m/4$ vertices, apply Case 1 to the
coloring $(W, \Gamma\setminus W)$.  We make significant
progress unless Case 1 returns the update
$\Gamma\leftarrow W$.   In this case, go to Case 4.

\mn
{\bf Case 4.}\ $\Gamma$ is the set of vertices of a canonically
embedded UPCC.

\mn

In this case we recursively apply Theorem~\ref{thm:UPCC} (``UPCC
Split-or-Johnson'') through Lemma~\ref{lem:UPCC-to-bipartite} 
(``UPCC to bipartite'') with parameter $\beta=3/4$.
Note that $|\Gamma|< 1+\sqrt{2 n_2}$, a dramatic reduction
of the size of the problem.

If Theorem~\ref{thm:UPCC} returns a canonical colored $3/4$-partition 
of $\Gamma$, apply Cases 1 and~2 to make significant progress.

If Theorem~\ref{thm:UPCC} returns a Johnson scheme
on a subset $W\subseteq \Gamma$ of size $|W|\ge 3m/4$
then apply Case 1 to the coloring $(W,\Gamma\setminus W)$.
We make significant progress unless Case 1 returns
the update $\Gamma\leftarrow W$.  So in this case we
have a Johnson scheme on $\Gamma$ and move to Case 5.

\mn
{\bf Case 5.}\ $\Gamma=\binom{\Gamma'}{t'}$ is the vertex set
of a canonically embedded Johnson scheme $\jjj(\Gamma',t')$
$(t'\ge 2)$.

\mn
In this case $V_2$ is identified with the set
\begin{equation}
 V_2 \leftarrow\binom{\binom{\Gamma'}{t'}}{t}.
\end{equation}
This is a highly structured set with ample imprimitivity of which
we take advantage.  Each vertex $x\in V_2$ can now be viewed as a
$t'$-uniform hypergraph $\calH(x)$ with $t$ edges on the vertex set
$\Gamma'$.  

If $m'\le \ell= (\log_2 n)^3/\log_2\log_2 n$, we individualize each
element of $\Gamma'$. This in turn individualizes each element of
$\Gamma$, then each element of $V_2$, and finally each element of
$V_1$, at a multiplicative cost of 
$|\Gamma'|! < \ell^{\ell} < \exp_2(3(\log_2 n_2)^3)$.

Now assume $m'>\ell$.  

\begin{lemma}
$V_2$ is canonically $1/6$-color-partitioned
by classical WL applied to the tripartite graph on
$(V_2,\Gamma,\Gamma')$ where each $v\in V_2$ is adjacent
to the corresponding $t$-tuple in $\Gamma$ and 
each $w\in\Gamma$ is adjacent 
to the corresponding $t'$-tuple in $\Gamma'$.
\end{lemma}

\begin{proof}
Let $U(x)\subseteq \Gamma'$ denote the union 
of the edges of $\calH(x)$.  So $|U(x)|\le tt'$.  

\mn 
Claim.  If $m' >\ell$ then $tt' < 2\log_2 n_2 / \log_2 \log_2 n_2$.

\begin{proof}
First note that if $1\le b\le a/2$ then $\binom{a}{b}\ge (a/b)^b \ge 2^b$, 
so $b\le \log_2 \binom{a}{b}$.  
Applying this twice we see that $t\le \log_2 \binom{m}{t} 
= \log_2 n_2$ and $t'\le \log_2 \binom{m'}{t'} = \log_2 m <\log_2 n_2$.  
So $t' < \sqrt{m'}$ and 
$t < \sqrt{m'}<\sqrt{\binom{m'}{t'}}$.  Now if $b\le \sqrt{a}$ then
$\binom{a}{b}> (a/b)^b \ge a^{b/2}$, so, bootstrapping, we obtain
 \begin{equation}
n_2 = \binom{\binom{m'}{t'}}{t}> {\binom{m'}{t'}}^{t/2}> (m')^{tt'/4}
\end{equation}
and therefore $tt' < 2\log n_2/\log_2\log_2 n_2$.
\end{proof}
Let us color $V_2$ as $V_2=Q\dot\cup R$ where $x\in Q$ exactly
if $|U(x)|=tt'$ (all edges of $\calH(x)$ are disjoint).
(WL is ``aware'' of this coloring, so the WL coloring is a refinement
of this.)
By Prop.~\ref{prop:random-hyp} (``Random hypergraphs'')
we have $|R|/|V_2|\le (tt')^2/m'$.
The right-hand side goes to zero, so for sufficiently large $n$
we have $|R|/|V_2|\le 1/6$, say.

For $x,y\in Q$ let us now say that $x\sim y$ if $U(x)=U(y)$.
This equivalence relation gives a canonical equipartition on $Q$.  
(Again, WL is ``aware'' of this relation.)
The size of each equivalence class is $(tt')!/(t!(t'!)^t)\ge 3$
(because $t,t'\ge 2$; the smallest case occurs when $t=t'=2$).  
The number of equivalence classes is 
$\binom{m'}{tt'}\ge \binom{m'}{4}$ which goes to infinity,
so for sufficiently large $n$ it will be at least 6, say.

In summary, the $(Q,R)$ canonical coloring together with
the $\sim$ canonical equivalence relation on $Q$ provides
a canonical colored $1/6$-partition of $V_2$.  (In fact, it is 
a $o(1)$-partition.)  WL gives a refinement of these.
\end{proof}

Now apply {\sf Reduce-Part2-by-Color} to the color-partition
obtained.  If a color class of size $\le n_2/6$ 
is returned, we reduced $n_2$ by a factor of 6, significant 
progress.  If a larger color class is returned, it is
nontrivially equipartitioned.  Apply {\sf Procedure ImprimitiveCase},
significant progress.

This completes Case 5.

\subsubsection{Local guides for the Johnson case}
\label{sec:localguides-johnson}

After these subroutines, we now turn to the main algorithm
for the case when $V_2=\binom{\Gamma}{t}$
is the vertex set of the canonically
embedded Johnson scheme $\jjj(\Gamma,t)$.

For $L\subseteq\Gamma$ let $X[L]$ denote the subgraph 
of $X$ induced by $\left(V_1,\binom{L}{t}\right)$.  
Our ``test sets'' will be the $\ell$-subsets $L\subset\Gamma$
(\ie, $L\in\binom{\Gamma}{\ell}$).
Determine the isomorphisms among all the $X(L)$ 
for all the test sets $L$ (by brute force over
all bijections between each pair $L_1,L_2$) 
at an additive cost of $m^{O(\ell)}n^{O(1)}$.

\mn
{\bf Case A:} Not all the $X[L]$ are isomorphic.

This gives us a canonical coloring of $\binom{\Gamma}{\ell}$ at
additive quasipolynomial cost.  (The color of $L$ is the
isomorphism type of $X(L)$.  This does not require canonical forms,
only a comparison of all graphs of the form $X(L)$ arising from $X$
and its counterpart airising from the other input string (with
which we are testing isomorphism).  But in this
quasipolynomial-size category, even lex-first labeling would be
feasible in quasipolynomial time.)

We consider the symmetry defect of the resulting edge-colored
$\ell$-uniform hypergraph $\calK$ on vertex set $\Gamma$.
Recall that we do not need to
distinguish weak and strong defect in case of hypergraphs 
(see Prop.~\ref{prop:weak-is-strong}).

\mn
{\bf Case A1.}\   $\calK$ has symmetry defect $< m/4$.  

\mn
Let $S\subset \Gamma$ be the unique largest symmetrical set for
$\calK$; so $|S|>3m/4$.  Apply Case~1 to the coloring $(S,
\Gamma\setminus S)$.  We make significant progress except when Case 1
returns the update $\Gamma\leftarrow S$.  If this is the case, we
enter Case B below.

\mn
{\bf Case A2.}\   $\calK$ has symmetry defect $\ge m/4$.  

\mn
In this case we view $\calK$ as an $\ell$-ary relational structure
and move to Case 3.  

\mn
{\bf Case B.}\  All the $X[L]$ are isomorphic.

\mn
Let $G(L)\le \sym(L)$ denote the image of the
action of $\aut(X(L))$ on $L$.  We say that the test set
$L$ is \emph{full} if $G(L)=\sym(L)$.  

\mn
{\bf Case B1.}\  $G(L)\neq \sym(L)$  \quad (the test set $L$ is not full).

\mn
We should always have discussed the progress of the two 
input strings $\xx_1$ and $\xx_2$ (of which we wish to
decide $G$-isomorphism) in combination.  The general
argument is that either the progress of the parameters is
not identical, in which case we reject isomorphism and exit,
or it is identical, so it suffices to follow one of them
as long as canonicity with respect to the pair of objects
is observed.  


Due to the delicate nature of the argument involved in the 
prerequisite for the present case,
Prop.~\ref{prop:localguide} (``Local guide''),
we make an exception at this time.  This will also
illustrate the general principle tacitly present in all
combinatorial partitioning arguments.

Rename $X$ as $X_1$.  The graph $X_1$ derives from an input string
$\xx_1$.  Let $\xx_2$ be the other input string and
let $X_2$ be the graph $X$ derived by the same procedure from 
$\xx_2$.  This defines our notation for the associated structures;
in particular, $\Gamma$ becomes $\Gamma_1$ and the corresponding
set derived from the input string $\xx_2$ will be denoted $\Gamma_2$.

We shall apply Prop.~\ref{prop:localguide} 
with the following assignment of the variables: $\alpha\leftarrow 3/4$,
$\Omega_i\leftarrow \Gamma_i$, $n\leftarrow m$, $\ellk\leftarrow\ell$.  

The objects of the category $\calL$ correspond to the pairs $(L,i)$ where
$L\in\binom{\Omega_i}{\ell}$.  The morphisms $(L,i)\to (L',j)$
are the bijections $L\to L'$ induced by the $X_i(L)\to X_j(L')$
isomorphisms.  The two abstract objects of category $\calC$ are
denoted $\xxx_1$ and $\xxx_2$.  The underlying set of $\xxx_i$
is $\Box(\xxx_i)=\Gamma_i$.  The morphisms are the bijections
$\Gamma_1\to\Gamma_2$ induced by the isomorphisms $X_1\to X_2$.

Our assumption is that all objects of the form $(L,1)$ are
isomorphic.  If this is not true for the objects $(L,2)$, reject
isomorphism, exit.  Let now $L_i\in \binom{\Gamma_i}{\ell}$.
If the objects $(L_1,1)$ and $(L_2,2)$ are not isomorphic,
reject isomorphism, exit.

So now all objects in the category $\calL$ are isomorphic.

For $L\in\binom{\Gamma_i}{\ell}$, let $G_i(L)$ denote the
restriction of $\aut(X_i(L))$ to $L$.  We say that
$L\in\binom{\Gamma_i}{\ell}$ is \emph{full} for $i$
if $G_i(L)=\sym(L)$.  Our current assumption is that 
$G_i(L)$ is not full (for at least one pair $(L,i)$,
and therefore for all pairs $(L,i)$).

Thus the assumptions of Prop.~\ref{prop:localguide}
are satisfied.   The algorithm of Prop.~\ref{prop:localguide}
returns canonical $k$-ary relational structures with strong 
symmetry defect $\ge m-\ell+1 > m/2$.

We feed these to the Design Lemma.
If a colored $3/4$-partition of $\Gamma_i$ is received, return that
partition and move to Cases 1 and~2 to make significant progress.

If the Design Lemma returns a UPCC canonically embedded
on at least $3m/4$ elements of $\Gamma_i$, move to 
Case 4.  Note that $m\le 1+\sqrt{2n_2}$, a dramatic reduction in the
domain size.

Now we return to discussing only one of the input strings.

\mn
{\bf Case B2.}\  $G(L)= \sym(L)$   \quad (the test set $L$ is full).

\mn
Since all the $X(L)$ are isomorphic, this is now true for all 
$L\in\binom{\Gamma}{\ell}$.
For $v\in V_1$ let $\calH(v,L)=(L,\calE(v,L))$ be
the induced subhypergraph of $\calH(v)$ on $L$.  Then, for each $v\in
V$ and $\sigma\in G(L)$ there is at least one vertex $v(\sigma)$ such
that $\calH(v,L)^{\sigma}=\calH(v(\sigma),L)$.  So the number of
vertices is $n_1\ge |\sym(L):\aut(\calH(v,L))|$.  Let $s$ be the
largest value such that $n_1\ge \binom{\ell}{s}$.
Since $n_1 < \exp(\sqrt{\ell})$, it follows from 
Theorem~\ref{thm:liebeck} that $\aut(\calH(v)) \ge\alt(L\setminus S)$ 
for some set $S\subset L$, $|S|=s$.
Since $\binom{\ell}{s}> (\ell/s)^s$, we see that 
$s\le 2\log_2 n_1/\log_2 \log_2 n_1$.

Now it follows from Lemma~\ref{lem:local-to-global}
(``Local to global symmetry lemma'')  that for every $v$ there is a set 
$S(v)\subset \Gamma$ with $|S(v)|=s$ such that 
$\aut(\calH(v))\ge \alt(\Gamma\setminus S(v))$.
(By Prop.~\ref{prop:weak-is-strong}
it follows in fact that $\aut(\calH(v))\ge \sym(\Gamma\setminus S(v)$.)
Let $S(v)$ denote the smallest such subset for $\calH(v)$.
This subset is easy to find in polynomial time.
(To apply Lemma~\ref{lem:local-to-global} 
we need to verify that $\ell > (t+2)(t+3)s$ which is true because 
both $t$ and $s$ are $O(\log n/\log\log n)$.)

\mn
Now consider the bipartite graph $Y=(V_1,\Gamma;F)$ where 
$g\in \Gamma$ is adjacent to $v\in V_1$ if $g\in S(v)$.

The map $X\mapsto Y$ is a canonical transformation (object map of a
functor) to an instance with $V_2$ greatly reduced 
($|\Gamma|< 1+\sqrt{2|V_2|}$).

If none of the twin-equivalence classes of $Y$ in $V_1$ are 
greater than 
$2n_1/3$, return $Y$, exit.

Else, let $S\subset \Gamma$ be the set that is equal to $S(v)$ for at
least a 
$2/3$ fraction of the elements of $V_1$.  Individualize each
point in $S$.  In the next round of refinement of the graph $Y$ this
will individualize each vertex $v$ with $S(v)=S$.  We justify
this last statement in Cor.~\ref{cor:johnson-individualize}
below.

\begin{observation}   \label{obs:trace}
The $X$-neighborhood of each $v\in V_1$ is determined by the trace
$\calH(v)_{S(v)}$ of the neighborhood hypergraph $\calH(v)$ in $S(v)$.
\end{observation}
\begin{proof}
The $X$-neighborhood of each $v\in V_1$ is determined by its
neighborhood hypergraph $\calH(v)$.
A subset $E\subset \Gamma$ belongs to $\calH(v)$ if and only if
$E\cap S(v)\in \calH(v)_{S(v)}$.
\end{proof}
\begin{corollary}    \label{cor:johnson-individualize}
If, for some $v\in V_1$, we individualize each element of the set $S(v)$
then, after a refinement step in the graph $Y$, the vertex $v$ is 
individualized.
\end{corollary}
\begin{proof}
This follows from Obs.~\ref{obs:trace} in view of the assumption that
there are no $X$-twins in $V_1$.
\end{proof}

This completes Case B2 and thereby the subroutine for the
Johnson case.  This in turn completes the 
{\sf Bipartite Split-or-Johnson} and
{\sf UPCC Split-or-Johnson} algorithms.


\section{Alternating quotients of a permutation group} 
\label{sec:altquotient}
To understand the structure of the groups where Luks reduction
stops, we need some group theory.

In Luks's barrier situation we have a 
\emph{giant representation} $\vf : G\to\sym(\Gamma)$, meaning
the image of $G$ is a giant, \ie, $G^{\,\vf}\ge \alt(\Gamma)$.  
We shall assume 
that $k=|\Gamma|> \max\{8, 2+\log_2 n_0\}$ where 
$n_0$ is the length (size) of the largest orbit of $G$.
We say that $x\in\Omega$ is \emph{affected} by $\vf$
if $G_x^{\,\vf}\ngeq \alt(\Gamma)$.  The key result
is that the pointwise stabilizer of all unaffected points
is still mapped onto $\sym(\Gamma)$ or $\alt(\Gamma)$
by $\vf$ (Unaffected Stabilizer Theorem, Thm.~\ref{thm:unaffected}).
This result will be responsible for the key algorithm
of the paper ({\sf Procedure LocalCertificates}
in Sec.~\ref{sec:localcertificates}).

Finally we show
that if a permutation group $G\le\sss_n$ has an alternating
quotient of degree 
$k\ge \max\{9, 2\log_2 n\}$
this can only happen in the trivial
way, namely, that for some $t\ge 1$, \ $G$ has a system of
imprimitivity with $\binom{k}{t}$ blocks on which $G$ acts as
the Johnson group $\aaa_k^{(t)}$.  
Moreover, in every orbit there is a canonical choice of the blocks
corresponding to $\vf$
which is unique; we refer to these as the \emph{standard blocks}.
These are some of the items in our ``Main Structure Theorem''
(Theorem~\ref{thm:altquotient}).  

\subsection{Simple quotient of subdirect product}

First we state a lemma that is surely well known but to which I
could not find a convenient reference.

\begin{lemma}[Simple quotient of subdirect product]
  \label{lem:topfactor}
Let $G\le K_1\times \dots \times K_m$ be a subdirect product;
let $M_i$ be the kernel of the $G\to K_i$ epimorphism.
Assume there is an epimorphism $\vf\,:\, G\to S$ where $S$ is
a nonabelian simple group.
Then $(\exists i)(M_i\le \ker \vf)$.  In particular, one of the $K_i$
admits an epimorphism onto $S$.
\end{lemma}


\smallskip\noindent
\begin{proof}[Simplified proof by P. P. P\'alfy]
Let $N = \ker \vf$.  Assume for a contradiction that 
$N \ngeq M_i$ for all $i$.   Then $M_iN = G$ (because $N$ is a maximal
normal subgroup).   It follows that 
$[G,\dots,G]=[M_1N,\dots,M_mN]\le N[M_1,\dots,M_m]\le 
   N\left(\bigcap_{i=1}^m M_i\right) = N$,
so $[G/N,\dots,G/N]=1$, a contradiction because $G/N\cong S$ is
nonabelian simple.
\end{proof}
In our applications, $S$ will be $\aaa_k$ and the $K_i$ the restrictions
of the permutation group $G$ to its orbits.

\begin{remark} 
A group $H$ is \emph{perfect} if $H'=H$ (where $H'=[H,H]$ denotes
the commutator subgroup of $H$).  P\'alfy points out that the following
result appears as Lemma 2.8 in \cite{meier}.
\begin{lemma}[Meierfrankenfeld]
Let $G$ be a finite group and $N\normal G$ such that $G/N$ is perfect.
Then there exists a unique subnormal subgroup $R$ of $G$ which is 
minimal with respect to $G=RN$.
\end{lemma}
Lemma~\ref{lem:topfactor} follows from this result.  Indeed,
observe that if $M_iN=G$ for all $i$ then $R\le M_i$ for all $i$
and therefore $R\le \bigcap_i M_i=1$, contradicting the equation $RN=G$.

We believe, though, that Lemma~\ref{lem:topfactor} must have been folklore
decades before this 1995 reference.
\end{remark}

\subsection{Large alternating quotient of a primitive group}

The result of this section is Lemma~\ref{lem:prim-topquotient}.

First we need to state a corollary to the basic structure theorem of 
primitive groups, the O'Nan--Scott--Aschbacher 
Theorem (\cite{sc,aschbacher}, cf.~\cite[Thm. 4.1A]{dixon}).
The proof of this theorem is elementary.  In fact, according
to Peter Neumann (cited by Peter Cameron~\cite{cameronblog})
much of it
already appears in Jordan's 1870 monograph~\cite{jordan}.

\begin{definition}
The \emph{socle} $\soc(G)$ of a group $G$ is the product of its minimal
normal subgroups.  
\end{definition}
\begin{fact} \label{fact:socle}
\begin{enumerate}[(i)]
\item The socle of any group is a direct product of simple groups.
\item  The socle of a primitive permutation group is a direct product
 of isomorphic simple groups.
\end{enumerate}
\end{fact}
The following is an easy corollary to the O'Nan--Scott--Aschbacher Theorem.
\begin{corollary}[O'Nan--Scott--Aschbacher] \label{cor:onan}
Let $G\le\sss_n$ be a primitive group with a nonabelian socle
$\soc(G)\cong R^s$ where $R$ is a nonabelian simple group.  
Then $n\ge 5^s$.
\end{corollary}
Now we state our result.
\begin{lemma}   \label{lem:prim-topquotient}
Let $G\le \sym(\Omega)$ be primitive. 
Assume $\vf : G\to\aaa_k$ is an epimorphism where
$k > \max\{8, 2+\log_2 n\}$.  Then $\vf$ is an isomorphism;
hence $G\cong \aaa_k$.
\end{lemma}
\begin{proof}
Let $N=\soc(G)$. 
By Fact~\ref{fact:socle}
$N$ can be written as
$N = R_1\times\dots\times R_s$ 
where the $R_i$ are isomorphic simple groups.  

\mn Case 1.  $N$ is abelian (the ``affine case'') and therefore
regular, \ie, $n=|N|$.  In this case $N\cong \Z_p^s$ and $G/N\le \gl(s,p)$
for some prime $p$, so $n=p^s$.  Moreover $\aaa_k$ is involved in
$\gl(s,p)$.  It is shown in~\cite[Prop. 1.22]{saxl} that if $\aaa_k$ is
involved in $\gl(s,p)$ then, combining a result of 
Feit and Tits~\cite{feit-tits} with~\cite[Prop. 5.3.7]{kleidman}, 
for $k\ge 9$ it follows that $k\le s+2$.  But we have 
$s+2 \le 2+\log_p n<k$, a contradiction, so this case cannot occur.

\mn
Case 2. $N$ is nonabelian. 
By Cor.~\ref{cor:onan} we have $s\le\log_5 n$.  
In particular, $k>s$.

Following~\cite{beals}\footnote{Introduced in~\cite{beals} (1999),
this notation was subsequently adopted in computational group 
theory (see. e.\,g., \cite{holt}).}, 
let $\pker(G)$ (``permutation kernel'') denote the
kernel of the induced permutation action $G\to \sss_s$ (permuting the
copies of $R$ by conjugation by elements of $G$).  Then 
$\pker(G)\le \aut(R_1)\times\dots\times\aut(R_s)$.  It follows that
$\pker(G)/N\le \out(R_1)\times\dots\times\out(R_s)$ is solvable by
Schreier's Hypothesis (a known consequence of the CFSG).

Now $G/\pker G\le \sss_s$ and $s<k$ so $G/\pker G$ cannot involve
$\aaa_k$.  The solvable group $\pker(G)/N$ also does not involve
$\aaa_k$.  It follows that $G/N$ does not involve $\aaa_k$ and
therefore $\ker\vf \ngeq N$.  

Let $M$ be a minimal normal subgroup of $G$.

\mn
Case 2a.  
$M\neq N$.  Then there is a unique other minimal 
normal subgroup, $M^*$, the centralizer of $M$, which is
isomorphic to $M$.   It follows that $M$ is regular, so
$n=|M|$.  Moreover, $s$ is even and $|M|=|\aaa_k|^{s/2}$.
Hence, $n\ge |\aaa_k|> 2^k > n$, a contradiction. So this case
cannot occur.   

\mn
Case 2b. $M=N$ is the unique minimal normal subgroup of $G$.
Since $N \nleq \ker\vf$, it follows that $\ker \vf =1$.
\end{proof}

\begin{remark}   \label{rem:zeroweight}
The assumption $k>2+\log_2 n$ is tight infinitely often, as shown by
the affine case of even dimension in characteristic 2.  In this case
$G=\Z_2^{k-2}\rtimes \aaa_k$ acts primitively on $n=2^{k-2}$
elements as follows: $\aaa_k$ acts on $\Z_2^{k}$ by permuting
the coordinates; restrict this action to the zero-weight subspace
$\sum x_i =0$, and then to the quotient space by the one-dimensional
subspace $x_1=\dots=x_k$ 
(this is contained in the zero-weight
subspace when the dimension is even).  In this case, $k=2+\log_2 n$,
and $\aaa_k$ is a proper quotient of $G$.
\end{remark}

\subsection{Alternating quotients versus stabilizers}
\begin{lemma}   \label{lem:proper}
Let $G\le \sym(\Omega)$ be a transitive 
permutation group and $\vf : G\to\aaa_k$ an epimorphism where
$k > \max\{8, 2+\log_2 n\}$.  
Then $G_x^{\,\vf}\neq \aaa_k$ for any $x\in\Omega$.
\end{lemma}
\begin{proof}
We proceed by induction on the order of $G$.
Let $N=\ker\vf$.  Assume for a contradiction that $G_x^{\,\vf}=\aaa_k$,
\ie, $NG_x=G$.

Let $B$ be a maximal block of imprimitivity containing $x$ (so
$|B|<|\Omega|$).  (If $G$ is primitive then $B=\{x\}$.)  
So $G_B\ge G_x$ and therefore $NG_B=G$.

Let $\Omega'$ be the set of $G$-images of $B$.  This is a system of
imprimitivity on which $G$ acts as a primitive group; let $K$ be
the kernel of this action.  

Since $N$ is a maximal normal subgroup of $G$, we have $K\le N$ or
$KN=G$. 

If $K\le N$ then $\vf$ maps the primitive group $G/K$ onto $\aaa_k$
and therefore by Lemma~\ref{lem:prim-topquotient}, $G/K\cong \aaa_k$,
hence $K=N$ and therefore $KG_B=G$.  But obviously $G_B\ge K$, so
$G=KG_B=G_B$ and therefore $|\Omega'|=1$, \ie, $B=\Omega$,
a contradiction.

So we have $KN=G$, \ie,
$K^{\,\vf}=\aaa_k$.  Let $\Omega_1,\dots,\Omega_m$ denote the
orbits of $K$ 
$(m\ge 2)$. 
Let $K_i$ denote the restriction of $K$ to
$\Omega_i$ and $M_i\normal K$ the kernel of the $K\to K_i$
epimorphism.  By Lemma~\ref{lem:topfactor}, $(\exists i)(M_i\le N)$.
The set $(\Omega_1,\dots,\Omega_m)$ is a system of imprimitivity for $G$
on which $G$ acts transitively, so the $M_i$ are conjugate subgroups
in $G$ and therefore $M_i\le N$ for all $i$.  
Let $x\in\Omega_i$.  It follows from $M_i\le N$  that
the epimorphism $K\to \aaa_k$ (restriction of $\vf$ to $K$)
factors across $K_i$ as $K\to K_i\stackrel{\psi}{\to} \aaa_k$,
so $K_i^{\psi}=\aaa_k$.
By the inductive hypothesis, applied to
$K_i$, we infer that $(K_i)_x^{\,\psi}\neq \aaa_k$. 
On the other hand, 
$(K_i)_x^{\psi}=K_x^{\,\vf}\normal G_x^{\,\vf}=\aaa_k$. 
We conclude that
$|K_x^{\,\vf}|=1$ 
and therefore
$n > |x^K|=|K:K_x|\ge |K^{\,\vf}:K_x^{\,\vf}|
 =|K^{\,\vf}|=k!/2 > 2^{k-2} > n$, a contradiction.
%
\end{proof}

\begin{remark}
Again, the assumption $k>2+\log_2 n$ is tight; the Lemma fails infinitely
often if $k=2+\log_2 n$ is permitted.  This is shown by the same 
examples as in Remark~\ref{rem:zeroweight}.
\end{remark}

Next we extend Lemma~\ref{lem:proper} to not necessarily transitive groups.
\begin{lemma}   \label{lem:proper1}
Let $G\le \sym(\Omega)$ be a permutation group and 
$\vf : G\to\aaa_k$ an epimorphism.  Assume 
$k > \max\{8, 2+\log_2 n_0\}$ where $n_0=n_0(G)$ denotes
the length of the largest orbit of $G$.  
Then $G_x^{\,\vf}\neq \aaa_k$ for some $x\in\Omega$.
\end{lemma}
\begin{proof}
Let $\Omega_1,\dots,\Omega_m$ be the orbits of $G$ and let
$G_i$ be the restriction of $G$ to $\Omega_i$.  So $G$ is
a subdirect product of the $G_i$.  Let $M_i$ denote the
kernel of the $G\to G_i$ epimorphism.   By Lemma~\ref{lem:topfactor},
$(\exists i)(M_i\le \ker\vf)$, so $\vf$ factors across the
restriction $G\to G_i$ as $G\to G_i\stackrel{\psi}{\to}\aaa_k$.
So $G_i^{\psi}=\aaa_k$.  

Let $x\in\Omega_i$.  We apply Lemma~\ref{lem:proper} to 
$G_i$ and notice that $G_x^{\,\vf} = (G_i)_x^{\psi}\neq\aaa_k$.
\end{proof}

The following result, Theorem~\ref{thm:unaffected},
along with a companion observation, Cor.~\ref{cor:affected},
will be the principal tools
for our central algorithm, the {\sf LocalCertificates}
procedure.  Recall that $G_{(D)}$ denotes the
pointwise stabilizer of $D$ in $G$ ($D\subseteq\Omega$).

\begin{definition}[Affected]
We say that the homomorphism $\vf : G\to \sss_k$ is a 
\emph{giant representation}
if $G^{\,\vf}\ge\aaa_k$.  We say that $x\in\Omega$ is \emph{affected}
by $\vf$ if $G_x^{\,\vf}\ngeq \aaa_k$.
\end{definition}

\begin{theorem}[Unaffected Stabilizer Theorem] \label{thm:unaffected}
Let $G\le \sym(\Omega)$ be a permutation group and 
$\vf : G\to\sss_k$ a giant representation.  Assume 
$k > \max\{8, 2+\log_2 n_0\}$ where $n_0=n_0(G)$ denotes
the length of the largest orbit of $G$.  
Let $D$ be the set of elements of $\Omega$ not affected by $\vf$.
Then $G_{(D)}^{\,\vf}\ge\aaa_k$.
\end{theorem}
\begin{proof}
First assume $G^{\,\vf}=\aaa_k$.  The set
$D$ is $G$-invariant and $G_{(D)}$ is the kernel of the 
restriction map $G\to\sym(D)$.  Let $P\le\sym(D)$ be the
image of this map (restriction of $G$ to $D$), so
$P\cong G/G_{(D)}$.  Since $G_{(D)}\normal G$, we have
$G_{(D)}^{\,\vf}\normal G^{\,\vf}=\aaa_k$.  Assume for a contradiction
that $G_{(D)}^{\,\vf}\neq\aaa_k$; it follows that
$|G_{(D)}^{\,\vf}|=1$, \ie, $G_{(D)}\le\ker(\vf)$.
Hence $\vf$ factors across $P$ as 
$G\to P\stackrel{\psi}{\to} \aaa_k$. It follows that
$P^{\psi} = G^{\,\vf}=\aaa_k$ so $\psi$ is an epimorphism.  
By Lemma~\ref{lem:proper1} we have $P_x^{\,\psi}\neq\aaa_k$
for some $x\in D$.  But $P_x^{\,\psi}=G_x^{\,\vf}=\aaa_k$
(because $x\in D$ is not affected by $\vf$),
a contradiction.

Now if $G^{\,\vf}=\sss_k$ then
let $G_1=\vf^{-1}(\aaa_k)$.  Let $\vf_1$ be the restriction of $\vf$
to $G_1$, so $\vf_1 : G_1\to\aaa_k$ is an epimorphism.  Moreover,
$x\in\Omega$ is affected by $\vf$ if and only if $x$ is affected
by $\vf_1$ (because $\aaa_k$ has no subgroup of index 2).
An application of the previous case
to $(G_1,\vf_1)$ completes the proof.
\end{proof}

\begin{proposition}   \label{prop:affected-orbit}
Let $G\le \sym(\Omega)$ be a permutation group and 
$\vf : G\to H$ an epimorphism.  Let $\Delta\subseteq\Omega$
be an orbit of $G$ and $x\in\Delta$.  Let $L=G_x^{\,\vf}$.
Then each orbit of $\ker(\vf)$ in $\Delta$ 
has length $|\Delta|/k$ where $k=|H:L|$.  
\end{proposition}

\begin{proof}
First we note that $k$ only depends on $\Delta$, not on the
specific element $x\in\Delta$ because if $y\in \Delta$ then
$G_x$ and $G_y$ are conjugates in $G$.
Now let $N=\ker(\vf)$ and $|\Delta|=d$.  So $d=|G:G_x|$.   
The length of the $N$-orbit
$x^N$ is $|N:N_x|$.  We have
$|G:NG_x|=|G^{\,\vf}:G_x^{\,\vf}|=|H:L|=k$.  Therefore
$|NG_x : G_x|= d/k$.  But   
$|N:N_x|=|N:(N\cap G_x)|=|NG_x :G_x|=d/k$.
\end{proof}

\begin{corollary}[Affected Orbit Lemma] \label{cor:affected}
Let $G\le \sym(\Omega)$ be a permutation group and 
$\vf : G\to\sss_k$ a giant representation.  Assume $k\ge 5$.
Then, if $\Delta$ is an affected $G$-orbit, \ie,
$\Delta\cap D=\emptyset$, then $\ker(\vf)$ is not
transitive on $\Delta$; in fact, each orbit of
$\ker(\vf)$ in $\Delta$ has length $\le |\Delta|/k$.
\end{corollary}
\begin{proof}
For $k\ge 5$, the largest proper subgroup of $\aaa_k$
has index $k$, and the largest subgroup of $\sss_k$ not
containing $\aaa_k$ also has index $k$.  So the
statement follows from Prop.~\ref{prop:affected-orbit}.
\end{proof}

\begin{remark}
If $k\ge \max\{9,2\log_2 n_0\}$ then we can use
Theorem~\ref{thm:altquotient} to make a more detailed statement.  We
observe that $\ker(\vf)$ fixes each standard block (setwise)
(see item~\eqref{item:johnson} in Theorem~\ref{thm:altquotient})
so the length of each orbit of $\ker(\vf)$ contained in $\Delta$ is 
$\le |\Delta|/\binom{k}{t_{\Delta}}$.
\end{remark}

\subsection{Subgroups of small index in $\sss_n$}


\begin{observation}   \label{obs:unique}
Let $T,U\subset\Omega$, $|T|, |U| < n/2$, where
$n=|\Omega|\ge 5$.
Assume \\
$\alt(\Omega)_{(T)}\le \sym(\Omega)_{U}$.  Then $U\subseteq T$.
\end{observation}
\begin{proof}
By assumption, $|\Omega\setminus T|\ge 3$ and therefore
$\Omega\setminus T$ is an orbit of
$\alt(\Omega)_{(T)}$ so it must be part of an orbit of
$\sym(\Omega)_{U}$.  Since $|\Omega\setminus T|>n/2>|U|$, 
we must have 
$\Omega\setminus T \subseteq\Omega\setminus U$, as claimed.
\end{proof}

According to Dixon and Mortimer~\cite[p. 176]{dixon},
the following result goes back to Jordan 
(1870)~\cite[pp. 68--75]{jordan}; a modern 
treatment was given by Liebeck~\cite[Lemma 1.1]{liebeck}.  
We cite from the version given in \cite[Thm. 5.2A,B]{dixon}.
Uniqueness follows from Observation~\ref{obs:unique}.

\begin{theorem}[Jordan--Liebeck]    \label{thm:liebeck}
Let $\alt(\Omega)\le K\le\sym(\Omega)$.  Let $H\le K$ and 
$1\le r < n/2$ where $n=|\Omega|\ge 9$.  
Assume $|K:H|<\binom{n}{r}$.  Then there exists
a unique $T\subset\Omega$ with $|T|<n/2$ such that
$\alt(\Omega)_{(T)}\le H \le \sym(\Omega)_{T}$.
This unique $T$ satisfies $|T|<r$.
\end{theorem}

\begin{notation}   \label{not:sandwich}
Under the assumptions of Theorem~\ref{thm:liebeck}
we write $T(H)$ for the unique subset $T\subset\Omega$ 
guaranteed by the Theorem.  So we have
\begin{equation}  \label{eq:sandwich1}
   \alt(\Omega)_{(T(H))}\le H \le \sym(\Omega)_{T(H)}.
\end{equation}
\end{notation}
\begin{remark} \label{rem:empty}
$T(H)=\emptyset$ if and only if 
$\alt(\Omega)\le H\le \sym(\Omega)$.
\end{remark}

\subsection{Large alternating quotient acts as a Johnson group on blocks}

Recall that the homomorphism $\vf : G\to\sss_k$ is a \emph{giant
representation} if $G^{\,\vf}\ge \aaa_k$.

\begin{theorem}[Main structure theorem]   \label{thm:altquotient}
Let $G\le \sym(\Omega)$ be a permutation group and\\ 
$\vf : G\to\sss_k$ a giant representation.  Assume 
$k \ge \max\{9, 2\log_2 n_0\}$ 
where $n_0=n_0(G)$ denotes
the length of the largest orbit of $G$.  
\begin{enumerate}[(a)]
\item \label{item:sandwich}
 For every $x\in\Omega$
there exists a unique subset $T(x)\subset [k]$ such that 
$|T(x)|<k/4$ and
\begin{equation}   \label{eq:sandwich} 
        (\aaa_k)_{(T(x))} \le G_x^{\,\vf} \le (\sss_k)_{T(x)} .
\end{equation}
\item \label{item:affected}
      The element $x\in\Omega$ is affected by $\vf$ if and only if 
      $|T(x)|\ge 1$.
\item \label{item:orbitaffected}
     For each orbit $\Delta$ there is an integer $t_{\Delta}\ge 0$
     such that $|T(x)|=t_{\Delta}$ for every $x\in\Delta$.
     We say that $\Delta$ is \emph{affected by $\vf$} if $t_{\Delta}\ge 1$,
     \ie, the elements of $\Delta$ are affected.
\item \label{item-nontrivial}
     At least one orbit is affected.  In fact, if $D$ is the union
     of the unaffected blocks then\\ $G_{(D)}^{\,\vf}\ge \aaa_k$.
\item   \label{item:johnson}
      \emph{(Johnson group action on blocks)}
      For every orbit $\Delta$ 
      the equivalence relation $T(x)=T(y)$ 
      $(x,y\in\Delta)$ splits $\Delta$ 
      into $\binom{k}{t_{\Delta}}$
      blocks of imprimitivity, 
      labeled by the $t_{\Delta}$-subsets of $[k]$.
      We refer to these blocks as the \emph{standard blocks} for $\vf$.
      The action of $G$ on the set of standard blocks in $\Delta$ is 
      $\aaa_k^{(t_{\Delta})}$
      or $\sss_k^{(t_{\Delta})}$.  
      If $t_{\Delta}\ge 1$ then this is a Johnson group and the 
      kernel of this action is $\ker \vf$; if $t_{\Delta}=0$ then
      the action is trivial (its kernel is $G$, there is just one block,
      namely $\Delta$).
\item \label{item:block-stabilizer}
      If $B\subseteq\Delta$ is a standard block and $x\in B$ then
      $G_B^{\,\vf}=(G^{\,\vf})_{T(x)}$ (so it is
      either $(\sss_k)_{T(x)}$ or $(\aaa_k)_{T(x)}$).
\item \label{item:othersystems}
      If $\Psi=\{C_1,\dots,C_r\}$ 
      is another system of imprimitivity
      on the orbit $\Delta$ 
      such that the kernel of the action
      $G\to\sym(\Psi)$ is $\ker(\vf)$ then $r=\binom{k}{t'}$ for some
      $t'<t_{\Delta}$ and the $G$-action on $\Psi$ is $\sss_k^{(t')}$
      or $\aaa_k^{(t')}$.  In particular, the standard blocks form
      the unique largest system of imprimitivity on which
      the kernel of $G$-action is $\ker(\vf)$.    Moreover,
      if $x\in C_i$ then $|T(G_{C_i})|=t'$ and  
      $T(G_{C_i})\subset T(x)$.
\end{enumerate}
\end{theorem}
\begin{proof}
Item~\eqref{item:sandwich} follows from the Jordan--Liebeck 
theorem (Thm.~\ref{thm:liebeck}), setting $K=G^{\,\vf}$ and
$H=G_x^{\,\vf}$ (so $T(x)=T(G_x^{\,\vf})$) and noting that
\begin{equation} 
\binom{k}{\lfloor k/4\rfloor} > 2^{k/2} \ge n_0\ge 
|x^G|=|G:G_x|\ge |G^{\,\vf}:G_x^{\,\vf}| .
\end{equation} 

Item~\eqref{item:affected} is immediate from Eq.~\eqref{eq:sandwich} 
and the definition of being ``affected.''  

Item~\eqref{item:orbitaffected} follows from the observation that for 
$x\in\Omega$ and $\sigma\in G$ we have
\begin{equation}   \label{eq:conjugates}
G_{x^{\sigma}}=G_x^{\,\sigma} \text{\quad and therefore\quad}
T(x^{\sigma})=T(x)^{\sigma^{\,\vf}} .
\end{equation}

Item~\eqref{item-nontrivial} is of greatest importance; it is
the content of the ``Unaffected Stabilizer Theorem''
(Thm.~\ref{thm:unaffected}).

To see Item~\eqref{item:johnson}, let $\Delta$ be an orbit and let
$[x]$ denote the equivalence class (block) of $x\in\Delta$ under the
equivalence relation stated.  By Eq.~\eqref{eq:conjugates}, this
equivalence relation is $G$-invariant and $G$ acts transitively on the
blocks.  We also infer from Eq.~\eqref{eq:conjugates} that the blocks
in $\Delta$ are in 1-to-1 correspondence with the $t_{\Delta}$-subsets
of $[k]$ (noting that $\aaa_k$ acts transitively on
$\binom{[k]}{t_{\Delta}}$).  Moreover, through this
bijection, the $G$-action on the blocks in $\Delta$ is equivalent to
the action of $\aaa_k$ on $\binom{[k]}{t_{\Delta}}$.  This bijection
also proves item~\eqref{item:block-stabilizer}.

To see item~\eqref{item:othersystems}, first we note that
$r\ge 3$ (in fact, $r\ge k$) because the kernel of the
action on $\Psi$ has index $\ge k!/2$ and therefore
$r! \ge k!/2$.  Let $x\in C_i$ and $H=G_{C_i}$.  So
$G_x\le H$ and $H$ is a maximal subgroup of $G$ of index $\ge 3$.
Let $N=\ker(\vf)$; so $N\le H$ and $3\le |G:H|=|G^{\,\vf}:H^{\,\vf}|$.
Moreover, $H^{\,\vf}$ is a maximal subgroup of $\sss_k$ or $\aaa_k$ 
containing $G_x^{\,\vf}\ge (\aaa_k)_{(T(x))}$.  For $T\subset [k]$ with
$|T|<k/2$, the only maximal 
subgroups of $\sss_k$ containing $(\aaa_k)_{(T)}$ are of the form
$(\sss_k)_{U}$ for $U\subseteq T$.  Intersecting these with $\aaa_k$
we obtain the maximal subgroups of $\aaa_k$ containing $(\aaa_k)_{(T)}$.
This proves that $T(G_{C_i})\subset T(x)$.  Setting $t'=|T(G_{C_i})|$,
the corresponding Johnson group action on $\Psi$ follows the lines
of the proof of item~\eqref{item:johnson}.
\end{proof}
\begin{remark}[Tight bound for $k$]      \label{rem:sandwich}
The actual condition on $k$, sufficient for most conclusions
of the theorem, is that $k>\max\{8,2+\log_2 n_0\}$ and
$\frac{1}{2}\binom{k}{\lfloor k/2\rfloor} > n_0$. The latter translates to
$k > \log_2 n_0 +(1/2+o(1))\log_2\log_2 n_0$.  
The only difference would be that instead of $|T(x)|<k/4$
we would only get $|T(x)|<k/2$, sufficient for our goals.

Our assumption $k \ge \max\{9,2\log_2 n_0\}$ is generously sufficient
for both conditions above.  
Under this condition we shall not only have $|T(x)|<k/4$
but $|T(x)|< H^{-1}(1/2)(1+o(1))k < k/9$ (for large $k$).
Here $H(x)$ is the binary entropy function,
so $H^{-1}(1/2)\approx 0.11003 <1/9$. --- We note that 
any bound of the form $k > c\log n_0$ would work for
the purposes of this paper; the actual value of $c$ will
not affect our complexity estimate.
\end{remark}

\begin{remark}[Multiple systems of imprimitivity]  \label{rem:othersystems}
The presence of multiple systems of imprimitivity
with the same kernel as discussed in Item~\eqref{item:othersystems}
is a real possibility.  Consider for instance the 
action $\sss_k\to\sss_{k(k-1)}$ defined by the action of
$\sss_k$ on the $k(k-1)$ ordered pairs; let $G\le\sss_{k(k-1)}$
be the image of this action.  Then $G$ has two systems
of imprimitivity on which $\sss_k$ acts in its natural action
(there are $k$ blocks in each system), and there is 
a unique system of imprimitivity with $\binom{k}{2}$ blocks
on which the action is $\sss_k^{(2)}$.  The latter
are the \emph{standard blocks}; in this case each standard
block has 2 elements.  Each of the three actions is faithful,
so their kernel is the same, namely, the idenity.
\end{remark}

Finally, and algorithmic observation.
\begin{proposition}
Given a giant representation $\vf : G\to \sss_k$, we can find the
standard blocks in each $G$-orbit in polynomial time.
\end{proposition}
\begin{proof} Standard.  
\end{proof}

\section{Verification of top action}
\label{sec:topaction}

In this section we show that if $\vf : G\to\sym(\Gamma)$ is a
giant representation then we can recognize whether or not 
$\vf$ maps $\aut_G(\xx)$ onto a giant and if so can find
$\iso_G(\xx,\yy)$, all this at the cost of $O(m)$ calls
to String Isomorphism on windows of size $\le n/m$, where 
$m=|\Gamma|$.  Note that the solution to the recurrence
$f(n)=O(mf(n/m))$ is subquadratic.

\begin{proposition}[Lifting]  \label{prop:lifting}
Let $G\le\sym(\Omega)$ and $H\le \sym(\Gamma)$ be permutation 
groups, $\vf : G\to H$ a homomorphism, and $N=\ker(\vf)$.
Given these data, the strings $\xx,\yy : \Omega\to\Sigma$ and 
$\sigmabar\in H$, one can reduce, in polynomial time,
the computation of the set
$\vf^{-1}(\sigmabar)\cap\iso_G(\xx,\yy)$ 
(set of liftings of $\sigmabar$ to isomorphisms)
to a single call to $\iso_N(\xx',\yy)$ for some string $\xx'$.
\end{proposition}
\begin{proof}
If $\sigmabar\notin G^{\vf}$ then return ``empty.''
Otherwise let $\sigma\in\vf^{-1}(\sigmabar)$ and $\xx'=\xx^{\sigma}$.
Then $\vf^{-1}(\sigmabar)\cap\iso_G(\xx,\yy)=
   \sigma\iso_N(\xx^{\sigma},\yy)$.
\end{proof}
\begin{remark}
$H$ does not need to be a permutation group.  What we need is that
$H$ permit constructive membership testing, \ie, for any
list of elements $\tau_1,\dots,\tau_k,\rho\in H$ we should be able
to efficiently decide whether or not $\rho\in K$ where
$K$ is the subgroup generated by the $\tau_i$, and
if the answer is affirmative, to produce a straight-line program
that constructs $\rho$ from the $\tau_i$ 
(see Def.~\ref{def:straightline}). 
Constructive membership testing can be done, for instance, for matrix
groups over finite fields of odd characteristic in quantum polynomial
time~\cite{bbs}.
\end{remark}

\begin{definition}
A subcoset of a group $G$ is a coset of a subgroup.
Let $H\le G$ be groups and $\tau\in G$.  We say that the subset
$S\subseteq H\tau$ is a set of \emph{coset generators} of $H\tau$
if $H\tau$ is the smallest subcoset of $G$ containing $S$.  (Note that
ay intersection of subcosets of $G$ is either empty or a subcoset;
so every subset of $G$ generates a subcoset of $G$.)
\end{definition}
\begin{observation}
Let $S$ be a set of generators of the group $G$.  Then 
$S\cup\{1\}$ is a set of coset generators of $G$,
\ie, no proper subcoset of $G$ contains $S\cup\{1\}$.
\end{observation}

\begin{proposition}[TopAction1]   \label{prop:topaction}
Let $G\le\sym(\Omega)$ and $H\le \sym(\Gamma)$ be permutation 
groups, $\vf : G\to H$ a homomorphism, and $N=\ker(\vf)$.
Let $S$ be the given set of generators of $H$.
Given these data and the strings $\xx,\yy : \Omega\to\Sigma$,
we can achieve the following by recursively calling $|S|+1$
instances of String Isomorphism with respect to $N$,
at polynomial cost per instance:
\begin{enumerate}[(i)]
\item \label{topitem:i}
      decide whether or not $\vf$ maps $\iso_G(\xx,\yy)$ 
      onto $H$;
\item  \label{topitem:ii}
      if the answer is affirmative, find $\iso_G(\xx,\yy)$.
\end{enumerate}
\end{proposition}
\begin{proof}
Let $S'=S\cup\{1\}$; so $S'$ is a set of coset generators of $H$.  
Apply Prop.~\ref{prop:lifting} to each $\sigmabar\in S'$.
If there is a $\sigmabar\in S'$ for which the algorithm returns
the empty set ($\sigmabar$ does not lift to an isomorphism),
return the answer ``no'' to item~\eqref{topitem:i}. Else, return
the answer ``yes'' to item~\eqref{topitem:i} and observe that
$\iso_G(\xx,\yy)$ is the right subcoset of $G$ generated by
the subcosets $\vf^{-1}(\sigmabar)\cap\iso_G(\xx,\yy)$ 
$(\sigmabar\in S')$ found by Prop.~\ref{prop:lifting}.
\end{proof}

\begin{corollary}[TopAction2]   \label{cor:topaction}
Let $G\le\sym(\Omega)$ be a transitive permutation group
and $\vf : G\to\sym(\Gamma)$
a giant representation, where $|\Gamma|=m > \max\{8, 2+\log_2 n\}$.
Given these data and the strings $\xx,\yy : \Omega\to\Sigma$,
we can achieve the following by recursively calling $\le 4k$
instances of String Isomorphism with window size $\le n/k$
for some $m\le k\le n$, at polynomial cost per instance:
\begin{enumerate}[(i)]
\item \label{top2item:i}
      decide whether or not $\vf$ maps $\iso_G(\xx,\yy)$ 
      onto a giant coset, \ie, 
      $\iso_G(\xx,\yy)^{\vf}\ge \alt(\Gamma)\tau$
      for some $\tau\in\sym(\Gamma)$; 
\item  \label{top2item:ii}
      if the answer is affirmative, find $\iso_G(\xx,\yy)$.
\end{enumerate}
\end{corollary}
\begin{proof}
First assume $G^{\vf}=\alt(\Gamma)$.  Apply Prop.~\ref{prop:topaction}
to $H:=\alt(\Gamma)$ with $S$ a pair of generators of $H$.
This reduces our questions to three instances of $N$-isomorphism
where $N=\ker(\vf)$.  Now $N$ is intransitive with $k$ orbits
for some $k\le n$.  Each orbit has equal length (because
$N\normal G$) so Luks's Chain Rule performs the desired reduction,
calling $3k$ instances of window size $n/k$.
We need to justify the inequality $k\ge m$.
Lemma~\ref{lem:proper} (our first lemma toward the 
Unaffected Stabilizer Theorem, Theorem~\ref{thm:unaffected})
says that $\Omega$ is affected.  Then the Affected Orbit Lemma
(Cor.~\ref{cor:affected}) asserts that each orbit of $N$
has length $\le n/m$. 

Now if $G^{\vf}=\sym(\Gamma)$ then apply weak Luks reduction,
reducing $G$-isomorphism to two instances of $G_1$-isomorphism
where $G_1=\vf^{-1}(\alt(\Gamma))$.
\end{proof}
\begin{remark}
If $m \ge \max\{9, 2\log_2 n\}$ then $n/k=\binom{m}{t}$ for
some $1\le t < m/4$ by item~\eqref{item:johnson} of
the Main Structure Theorem (Theorem~\ref{thm:altquotient}).
\end{remark}
\begin{corollary}[TopAction3]   \label{cor:topaction-aut}
Let $G\le\sym(\Omega)$ be a transitive permutation group
and $\vf : G\to\sym(\Gamma)$
a giant representation, where $|\Gamma|=m > \max\{8, 2+\log_2 n\}$.
Given these data and the strings $\xx,\yy : \Omega\to\Sigma$,
we can achieve the following by recursively calling $\le 6k$
instances of String Isomorphism with window size $\le n/k$
for some $m\le k\le n$, at polynomial cost per instance:
\begin{enumerate}[(i)]
\item \label{top3item:i}
      decide whether or not $\vf$ maps $\aut_G(\xx)$ 
      onto a giant, \ie, 
      $\aut_G(\xx)^{\vf}\ge \alt(\Gamma)$;
\item  \label{top3item:ii}
      if the answer is affirmative, find $\iso_G(\xx,\yy)$.
\end{enumerate}
\end{corollary}
\begin{proof}
If $\iso_G(\xx,\yy)^{\vf}$ is a giant coset, we are done by 
Cor.~\ref{cor:topaction}.  We claim that if
If $\iso_G(\xx,\yy)^{\vf}$ is not a giant coset then
$\xx$ and $\yy$ are not $G$-isomorphic.  Indeed, if
$\xx\cong_G y$ then $\iso_G(\xx,\yy)=\aut_G(\xx)\sigma$ where
$\sigma$ is any element of $\iso_G(\xx,\yy)$.  It follows
that $\iso_G(\xx,\yy)^{\vf}=\aut_G(\xx)^{\vf}\sigmabar$
is a giant coset (where $\sigmabar=\sigma^{\vf}$).
\end{proof}
\begin{corollary}[TopAction4]   \label{cor:topaction-col}
Let $G\le\sym(\Omega)$ be a transitive permutation group
and $\vf : G\to\sym(\Gamma)$ a giant representation.
where $|\Gamma|=m \ge \max\{16, 4+2\log_2 n\}$. 
Let $\xx,\yy : \Omega\to\Sigma$ be strings.
Assume $\Gamma$ has a canonical coloring with respect to $\xx$
with a color class $C$ of size $|C| > m/2$
such that the restriction of $\aut_G(\xx)^{\vf}$ to $C$
is a giant (includes $\alt(C)$).  Then
we can find $\iso_G(\xx,\yy)$
by recursively calling $\le 6k$
instances of String Isomorphism with window size $\le n/k$
for some $|C|\le k\le n$, plus a number of instances
of total size $\le n$ and maximum size $\le 2n/3$.
\end{corollary}
\begin{proof}
Since $m \ge \max\{9, 2\log_2 n\}$, by the
Main Structure Theorem (Theorem~\ref{thm:altquotient})
$\Omega$ can be divided into standard blocks on which
$G$ acts as a Johnson group.   The standard blocks are 
labeled by $\binom{\Gamma}{t}$ for some $t\ge 1$;
and $\Omega(C)$ denotes the union of the standard 
blocks labeled by the elements of the set $\binom{C}{t}$.

Let $C_{\xx}=C$.  By canonicity, there is a corresponding
color class $C_{\yy}\subseteq \Gamma$ (which may be empty).
Apply items~\ref{item:aligned-reject} to~\ref{item:chain2}
of {\sf Procedure Align} (Sec.~\ref{sec:align})
with $\xxx(\xx):=C_{\xx}$ and $\xxx(\yy):=C_{\yy}$.
The result is that
\begin{itemize}
 \item if $|C_{\xx}|\neq |C_{\yy}|$ then isomorphism is rejected
 \item else $\yy$ is updated so now $C_{\xx}= C_{\yy}=C$ 
 \item from the coloring $(C,\Gamma\setminus C)$ of $\Gamma$ we infer
         a canonical coloring of $\Omega$; one of the color classes
         is $\Omega(C)$; and we begin the application of the Chain Rule
         with this color class.
\end{itemize}
Now we process $\Omega(C)$ via Cor.~\ref{cor:topaction-aut}.
This can be done because $|C| >m/2 \ge \max\{8, 2+\log_2 n\}$.
Then proceed to the remaining color classes in
accordance with the Chain Rule.

The bound $2n/3$ on the length of the remaining 
color classes comes from Lemma~\ref{lem:effect-on-tuples}.
\end{proof}
\begin{remark}
The cost of this procedure can generously be overestimated
by $6T(2n/3)$ where $T(n)$ is the maximum cost of instances
of size $\le n$.
\end{remark}

\section{The method of local certificates}    \label{sec:localcertificates}

\subsection{Local certificates for giant action: the core algorithm} 
\label{subsec:localcertificates}

In this section we consider the case of an imprimitive $G$ and
present the group-theoretic Divide-and-Conquer method.
This is the core algorithm of the entire paper.

The situation we consider is as follows.

The input is a transitive permutation group $G\le\sym(\Omega)$, 
a giant representation\\ $\vf : G\to\alt(\Gamma)$
(\ie, a homomorphism such that $G^{\,\vf}\ge\alt(\Gamma)$),
and two strings
$\xx,\yy :\Omega\to\Sigma$ ($\Sigma$ is a finite alphabet).

\mn
Notation: $n=|\Omega|$, $m=|\Gamma|$.  We shall assume 
$m \ge 10\log_2 n$.  
%

\begin{notation}   \label{not:ingamma}
Recall that for a subgroup $L\le G$ and a subset $A\subseteq\Gamma$ we
write $L_A$ to denote the setwise stabilizer of $A$ in $L$
with respect to the representation $\vf : L\to \sym(\Gamma)$.
We say that $A$ is $L$-invariant if $L_A=L$.  We write
$\psi_A : G_A\to\sym(A)$ for the map that restricts the 
$G^{\,\vf}$-action to $A$.
If $A$ is $L$-invariant then $L^A:=L^{\psi_A}$ is the
restriction of $L^{\,\vf}$ to $A$.  In particular,
$\psi_{\Gamma}=\vf$ and $L^{\Gamma}=L^{\,\vf}$.
\end{notation}
We note that the group $G_A$ can be computed trivially in polynomial
time as\\
$G_A = \vf^{-1}(\sym(\Gamma)_A)$.  

We note further that if $\vf : G\to \sym(\Gamma)$ is a giant
representation then $\psi_A : G_A\to\sym(A)$ is an epimorphism.
Indeed, the setwise stabilizer of $A$ in $\alt(\Gamma)$ 
acts on $A$ as $\sym(A)$, assuming $k\le m-2$ (\ie, $|A|\le |\Gamma|-2$).

\begin{definition}[Full set]
Let $A\subseteq\Gamma$.  We say that $A$ is \emph{full}
with respect to $\xx$ if \\
$\aut_G(\xx)_A^A \ge\alt(A)$, \ie, the $G$-automorphisms of $\xx$
induce a giant on $A$.\\
Notation: $\calF(\xx)=\{A\in\binom{\Gamma}{k}\mid A \text{ is full }\}$
and $\calFbar(\xx)=\binom{\Gamma}{k}\setminus\calF(\xx)$.
\end{definition}

We consider the problem of deciding whether or not a given
small ``test set'' $A\subset \Gamma$
is full and compute useful certificates of either outcome.
We show that this question can efficiently (in time 
$k!\poly(n)$) be reduced to the String Isomorphism problem 
on inputs of size $\le n/k$ where $k=|A|$ is the size of our 
test set; we shall choose $k=O(\log n)$.

\mn
{\bf Certificate of non-fullness.} 
We certify non-fullness of the test set $A\subset\Gamma$
by computing a permutation group $M(A)\le\sym(A)$ such that
(i) $M(A)\ngeq \aaa(A)$ and (ii) $M(A)\ge \aut_G(\xx)_A^A$
($M(A)$ is guaranteed to contain the projection of the 
$G$-automorphism group of $\xx$).

\sn
Such a group $M(A)$ can be thought of as
a constructive refutation of fullness.

\mn
{\bf Certificate of fullness.} 
We certify fullness of the test set $A\subset\Gamma$
by computing a permutation group $K(A)\le\sym(\Omega)$ such that
(i) $K(A)\le \aut_G(\xx)$ and (ii) $A$ is $K(A)$-invariant
and $K(A)^A\ge \alt(A)$.

\sn
Note that $K(A)$ represents an easily (poly-time) verifiable
proof of fullness of $A$.


Our ability to find $K(A)$, the certificate of fullness,
may be surprising because 
it means that from a local start (that may take only a small segment
of $\xx$ into account), we have to build up global automorphisms
(automorphisms of the full string $\xx$).  Our ability to do
so critically depends on the ``Unaffected Stabilizer Theorem''
(Thm.~\ref{thm:unaffected}).

\begin{theorem}[Local certificates]    \label{thm:local-certificates}
Let $A\subseteq\Gamma$ where $|A|=k$.  We refer to $A$ as 
our ``test set.''
Assume $\max\{8,2+\log_2 n\} < k \le m/10$.  
By making $\le k!n^2$ calls to String Isomorphism
problems on domains of size $\le n/k$ 
and performing $k!\poly(n)$ computation we can
decide whether or not $A$ is full and
\begin{enumerate}[(a)]
\item  if $A$ is full, find a certificate $K(A)\le \aut_G(\xx)$ 
       of fullness of $A$
\item    \label{item:witness} 
      if $A$ is not full, find a certificate $M(A)\le\sym(A)$
      of non-fullness.
\end{enumerate}
The families $\{(A, K(A)) : A\in\calF(\xx)\}$ and
$\{(A, M(A)) : A\in\calFbar(\xx) \}$ are canonical.
\end{theorem}


\begin{definition}[Affected]
Let $G\le\sym(\Omega)$ be a permutation group and
and $\vf:G\to \sym(\Gamma)$ a homomorphism.
Consistently with previous usage,
for a subgroup $H\le G$ we say that 
$x\in\Omega$ is \emph{affected by $(H,\vf)$} if 
$H_x^{\,\vf}\ngeq \alt(\Gamma)$.  
Let $\aff(H,\vf)$ denote the set of elements affected by $(H,\vf)$, \ie,
\begin{equation}
\aff(H,\vf)=\{x\in\Omega\mid H_x^{\,\vf}\ngeq \alt(\Gamma) \}.
\end{equation}
\end{definition}
Note that if $\vf$ restricted to $H$
is not a giant representation then all of
$\Omega$ is affected by $(H,\vf)$.

If $x\in\Omega$ is affected by $(H,\vf)$ then all elements of
the orbit $x^H$ are affected by $(H,\vf)$.  In other words,
$\aff(H,\vf)$ is an $H$-invariant set.  
So we can speak of \emph{affected orbits} of $H$ (of which all 
elements are affected).

We observe the dual monotonicity of the $\aff$ operator.
\begin{observation}
If $H_1\le H_2 \le G$ then $\aff(H_1,\vf)\supseteq \aff(H_2,\vf)$.
\end{observation}

The algorithm will consider the input in an increasing sequence
of ``windows'' $W\subseteq\Omega$; in each round, the part of
the input outside the window will be ignored.  The group
$H(W)$ will be the subgroup of $G_A$ that respects the
string $\xx^{W}$, the restriction of $\xx$ to $W$.

The initial window is the empty set (the input is wholly ignored),
so the initial group is $G_A$.  Then in each round we add to $W$
the set of elements of $\Omega$ affected by the current group
$H(W)$.  I like to visualize this process as ``growing the beard''
($W$ being the beard).  By the second round $W\neq\emptyset$
because $\aff(G_A,\psi_A)$ cannot be empty (by the Unaffected
Stabilizer Theorem).

As an increasing segment of $\xx$ is taken into account,
the group $H(W)$ (the automorphism group of this segment)
decreases, and thereby the set of elements affected
by $H(W)$ increases.  (Previous windows will always be
invariant under $H(W)$.)

We stop when one of two things happens:
either $\psi_A$ restricted to $H(W)$ is no longer a giant 
homomorphism, or the beard stops growing: no element outside
$W$ is affected by $H(W)$.  

In the former case we declare
that our test set $A$ is \emph{not full} 
(witnessed by a non-giant group $M(A):=H(W)^A\le\sym(A)$).
Note that the reason $M(A)$ is not a giant is still ``local,''
it only depends on the restriction of $\xx$ to the current window.

In the latter case we declare that $A$ is \emph{full}, and bring
as witness the group $K(A)=H(W)_{(\Wbar)}$, the pointwise
stabilizer of $\Wbar=\Omega\setminus W$ in $H(W)$.
We claim two things about $K(A)$.  First, $K(A)^{\,\vf}\ge \alt(\Gamma)$.
This follows from the Unaffected Stabilizer Theorem 
(Thm.~\ref{thm:unaffected}) since none of the elements of $\Wbar$
is affected.  (This is why the beard stopped growing.)
Second, we observe that $K(A)\le \aut_G(\xx)$.  
Indeed, $K(A)$ respects the letters of the string $\xx$
on $W$ (this is an invariant of the algorithm);
and it fixes all elements outside $W$, so the letters
of the string restricted to $\Wbar$ are automatically
respected\footnote{\label{fn:eureka}%
This observation was the culmination of a 
long struggle to construct global automorphisms from local 
information.  It amounted to the realization
of the decisive role the affected/unaffected dichotomy
was to play in the algorithm; indeed this was the moment
when the concept of this dichotomy crystallized.  It was the 
``eureka moment'' of this long quest.  It occurred around noon 
on September 14, 2015.}.

Here is the algorithm in pseudocode, with a more formal proof.

\begin{proof}[Proof of Theorem~\ref{thm:local-certificates}]
%
For $W\subseteq \Omega$ let $H(W)=\aut_{G_A}^W(\xx)$. 

All sets denoted $A, A'$, and $A_i$ below will be subsets 
of $\Gamma$ of size $k$ (the ``test sets'').  An invariant of 
the \bwhile\ loop 
will be that $A$ is invariant under the action of the group
$H(W)$, \ie, $H(W)\le G_A$.

\mn
{\sf Procedure LocalCertificates}

\mn
Input: 
Input: $G\le\sym(\Omega)$, epimorphism $\psi_A : G_A\to\sym(\Gamma)$,
       test set $A\in\binom{\Gamma}{k}$

\noindent
Output: decision: ``$A$ full/not full,'' group $K(A)$ (if full) or
$M(A)$ (if not full), set $W(A)\subseteq\Omega$

\noindent
Notation: $H(W) := \aut_{G_A}^W(\xx)$ \quad (to be updated as $W$ is updated)
\begin{tabbing} m \= m \= m \= m \= m \= mmmmmmmmmmmmmmmmmmmm \= m \kill \\
01 \>\> $W:=\emptyset$ \>\>\>\>  (: so $H(W)=G_A$ :) \\
02 \>\> \bwhile\ $H(W)^A \ge\alt(A)$ 
       and $\aff(H(W), \psi_A)\not\subseteq W$ \\
03 \>\>\> $W \leftarrow \aff(H(W),\psi_A)$ \>\>\> (: growing the beard :) \\
04 \>\>\> recompute $H(W)$ \\
05 \>\>\bendwhile \\
06 \>\> $W(A)\leftarrow W$ \\
07 \>\> \bif\ $H(W)^A \ge\alt(A)$ \>\>\>\> (: so
           $\aff(H(W), \psi_A)\subseteq W$:) \\
08 \>\>\> \bthen\ $K(A)\leftarrow H(W)_{(\Wbar)}$ 
                 where $\Wbar=\Omega\setminus W$  \\
09 \>\>\> \breturn\ $W(A)$, $K(A)$, ``$A$ full,'' \bexit 
         \>\>\> (: certificate of fullness found :)\\
10 \>\> \belse\ $M(A)\leftarrow H(W)^A$ \\
11 \>\>\> \breturn\ $W(A)$, $M(A)$, ``$A$ not full,'' \bexit 
         \>\>\> (: certificate of non-fullness found :)\\
\end{tabbing}

We need to show how to recompute $H(W)$ on line 4.
We write $W_{\old}$ for the value of $W$ before
the execution of line 03 and $W_{\new}$ after.

\mn
{\sf Procedure Recompute $H(W)$}
\begin{tabbing} m \= m \= m \= m \= m \= mmmmmmmmmm \= m \kill \\
04a \>\> $N\leftarrow H(W_{\old})_{(A)}^A$   
      \>\>\>\> (: kernel of $H(W_{\old})\to\sym(A)$ map :)\\
04b \>\> $L \leftarrow\emptyset$ 
           \>\>\>\> (: $L$ will collect elements of $H(W_{\new})$ :) \\
04c \>\> \bfor\ $\sigmabar\in H(W_{\old})^A$   
           \>\>\>\> (: $H(W_{\old})^A=\alt(A)$ or $\sym(A)$ :)\\   
04d \>\>\>  select $\sigma\in H(W_{\old})$ such that 
                   $\sigma^A=\sigmabar$ 
          \>\>\> \qquad\qquad\qquad\qquad
              (: lifting $\sigmabar$ to $\Omega$ :)\\
04e \>\>\>  $L(\sigmabar)\leftarrow \aut_{N\sigma}^{W_{\new}}(\xx)$  
      \>\>\>  (: performing strong Luks-reduction to $N$ :)\\
04f \>\>\>  $L \leftarrow L\cup L(\sigmabar)$  \\
04g \>\> \bendfor \\
04h \>\> \breturn\ $H(W_{\new})\leftarrow L$
\end{tabbing}

\sn
Justification.  First we observe that on each iteration of the 
\bwhile\ loop on lines 02--05,
$H(W_{\new})\le H(W_{\old})$ and $W_{\new}\supseteq W_{\old}$.
In fact, these inclusions are proper or else we exit on line 02.
In particular,
$A$ is invariant under $H(W)$ throughout the process because
it is invariant in line 01.  It also follows that on line 07
we actually have $\aff(H(W),\psi_A)=W$.  We also note that
the \bwhile\ loop will be executed at least once (by the comment
on line 01).

\begin{claim}   \label{claim:full}
On line 08, \quad $K(A)^A\ge \alt(A)$ and $K(A)\le\aut_G(\xx)$.
In particular, $A$ is full.
\end{claim}
\begin{proof}
$K(A)\ge\alt(A)$ is the crucial consequence of 
Theorem~\ref{thm:unaffected},
applied to the giant representation $\psibar_A : H(W_{\old})\to\sym(A)$.
($\psibar_A$ denotes the restriction of $\psi_A$ to
$H(W_{\old})$.)

To show that $K(A)\le\aut_G(\xx)$ let $\sigma\in K(A)$ and $u\in\Omega$.
We need to show that $\xx(u^{\sigma})=\xx(u)$.  If $u\in W$ then this
follows because $\sigma\in H(W)=\aut_G^W(\xx)$.  If $u\in \Wbar$ then 
$u^{\sigma}=u$.
\end{proof}
\begin{claim}
If $A$ is not full then we reach line 10 with $M(A)\not\ge\alt(A)$
and $\aut_G(\xx)_A^A\le M(A)$.
\end{claim}
\begin{proof}
We reach line 10 by Claim~\ref{claim:full}.  We then have
$\aut_G(\xx)_A^A\le M(A)$ because the relation $\aut_G(\xx)_A^A\le H(W)$
is an invariant of the process.
\end{proof}
Next we justify procedure {\sf Recompute $H(W)$}.
This is immediate from the observation 
\begin{equation}
      H(W_{\old}) = \bigcup_{\sigmabar} N\sigmabar
\end{equation}
where the union extends over $\sigmabar\in H(W_{\old})$.
So we can use strong Luks-reduction (over the orbits of
$N$ in $W_{\new}$) to compute $\aut_{H(W_{\old})}^{W_{\new}}(\xx)$.
But this group is $H(W_{\new})$ because $W_{\new}\supseteq W_{\old}$.

Finally we need to justify the complexity assertion.  This is where
Cor.~\ref{cor:affected} (``Affected Orbits Lemma'')
plays a critical role.

The \bwhile\ loop is executed at most $n$ times
(because $W$ strictly increases in each round;
we exit on line 02 when the ``beard'' stops growing),
so the dominant component of the complexity is in
recomputing $H(W)$.  We have reduced this to $\le k!$
instances of string $N$-isomorphism on the window
$W_{\new}$.   

By Cor.~\ref{cor:affected} (``Affected Orbits Lemma''),
each orbit of $N$ in $W_{\new}$ has length $\le n/k$.

We conclude that strong Luks reduction reduces the recomputation of
$H(W)$ to $\le n\cdot k!$ instances of String Isomorphism on windows
of size $\le n/k$, justifying the stated complexity estimate.
\end{proof}

Our procedure does more than stated in Theorem~\ref{thm:local-certificates}.
It also returns the set $W(A)$.  We summarize key properties of this
assignment.

\begin{proposition}
As in Theorem~\ref{thm:local-certificates}, let a ``test set'' be
a subset $A\subseteq\Gamma$ with $|A|=k$ elements where 
$\max\{8,2+\log_2 n\} < k \le m/10$.  
For all test sets $A$ we have
\begin{enumerate}[(i)]
\item $\Omega(A)\subseteq W(A)\subseteq \Omega$
\item $W(A)$ is invariant under $\aut_{G_A}(\xx)$ 
\item if $A$ is full then $W(A)=\aff(\aut_{G_A}^{W(A)}(\xx))$
\item if $A$ is full then $K(A)^A$ fixes all elements of
            $\Omega\setminus W(A)$
\item the assignment $A\mapsto W(A)$ is canonical.
\end{enumerate}
\end{proposition}
\begin{proof}  
Evident from the algorithm.
\end{proof}

We need to highlight one more fact about the structures we obtained.

\begin{notation}[Truncation of strings]
Let $\ast$ be a special symbol not in the alphabet $\Sigma$.  
For the string $\xx : \Omega\to\Sigma$ and ``window'' $W\subseteq\Omega$
we define the string $\xx^W : |omgea\to (\Sigma\cup\{\ast\})$ 
by setting $\xx^W(u)=\xx(u)$ for $u\in W$ and 
$\xx^W(u)=\ast$ for $u\in\Omega\setminus W$.
\end{notation}

\begin{notation}[Coloring of strings]
For the string $\xx :\Omega \to\Sigma$ and the ``test set''
$A\subseteq\Gamma$ we define the string
$\xx_A :  \Omega\to (\Sigma\times \{0,1\})$ by setting
$\xx_A(u)=(\xx(u),1)$ if $u\in\Omega(A)$ and
$\xx_A(u)=(\xx(u),0)$ if $u\notin\Omega(A)$.
\end{notation}

\begin{proposition}[Comparing local certificates]  \label{prop:compare}
For all test sets $A,A'\subseteq \Gamma$ with $|A|=|A'|=k$ and all
strings $\xx,\xx' :\Omega\to\Sigma$ we can compute
$\iso_G\left(\xx_A^{W(A)},(\xx')_{A'}^{W(A')}\right)$ by 
making $\le k!n^2$ calls to String Isomorphism
problems on domains of size $\le n/k$ 
and performing $k!\poly(n)$ computation.
\end{proposition}

\begin{proof}
Run procedure {\sf LocalCertificates} simultaneously
on $(\xx,A)$ and on $(\xx',A')$, maintaining the variable
$W$ for $(x,A)$ and the variable $W'$ for $(\xx',A')$.
Further maintain the set $Q=\iso_G(\xx_A^W,(\xx')_{A'}^{W'})$.
On line 01 we shall have $Q=G_A\sigma$ for any $\sigma\in G$
that takes $A$ to $A'$.

Change line 04 to ``recompute $H(W)$ and $Q$.''  \\
Here is the modified ``Recompute'' code.

\mn
{\sf Procedure Recompute $H(W)$ and $Q$}
\begin{tabbing} m \= m \= m \= m \= m \= mmmmmmmmmmmmmmmmm \= m \kill \\
04a  \>\>\> $N\leftarrow H(W_{\old})_{(A)}^A$ 
      \>\>\> (: kernel of $H(W_{\old})\to\sym(A)$ map :)\\
04b1 \>\>\> $L \leftarrow\emptyset$ 
           \>\>\> (: $L$ will collect elements of $H(W_{\new})$ :) \\
04b2 \>\>\> $R \leftarrow\emptyset$ 
           \>\>\> (: $R$ will collect elements of $Q_{\new}$ :) \\
04c0 \>\>\> fix $\pi_0\in Q_{\old}$ \\
04c1 \>\>\> \bfor\ $\sigmabar\in H(W_{\old})^A$   
           \>\>\> (: $H(W_{\old})^A=\alt(A)$ or $\sym(A)$ :)\\   
04d1 \>\>\>\>  select $\sigma\in H(W_{\old})$ such that 
                   $\sigma^A=\sigmabar$ 
          \>\> \qquad\qquad
              (: lifting $\sigmabar$ to $\Omega$ :)\\
04d2 \>\>\>\>  $\pi\leftarrow \sigma\pi_0$ \>\> (: $\pi\in Q_{\old}$ :) \\
04e1 \>\>\>\>  $L(\sigmabar)\leftarrow \aut_{N\sigma}^{W_{\new}}(\xx)$  \\
04e2 \>\>\>\>  $R(\sigmabar)\leftarrow 
               \iso_{N\pi}(\xx_A^{W_{\new}},(\xx')_{A'}^{W'_{\new}})$ 
      \>\>  (: performing strong Luks-reduction to $N$ :)\\
04f1 \>\>\>\>  $L \leftarrow L\cup L(\sigmabar)$  
      \>\>  (: collecting automorphisms :) \\
04f2 \>\>\>\>  $R \leftarrow R\cup R(\pibar)$  
      \>\>  (: collecting isomorphisms :) \\
04g \>\>\> \bendfor \\
04x \>\>\> \bif\ $R=\emptyset$ \bthen\ reject isomorphism, \bexit \\
04h \>\>\> \belse\ \breturn\ $H(W_{\new})\leftarrow L$ and $Q\leftarrow R$
\end{tabbing}
The analysis is analogous with the analysis of the 
 {\sf Recompute $H(W)$} routine.
\end{proof}

\begin{notation}[Sideburn]
Assume $A$ is full.  Let $\wtw(A)$ (the ``sideburn'') denote the set 
of those elements of $\Gamma$ that are affected by $K(A)^{\Gamma}$:
\begin{equation}
     \wtw(A) = \aff(K(A)^{\Gamma})
\end{equation}
\end{notation}
Clearly $A\subseteq \wtw(A)$.

Growing the ``sideburn'' is analogous to growing the ``beard'' except
we do not iterate ($K(A)$ already consists of ``global''
automorphisms).

Recall the definition of the \emph{minial degree} of a
permutation group (Def.~\ref{def:mindeg}).
\begin{proposition}   \label{prop:mindeg5}
Let $A\subseteq\Gamma$ be a full test set (subset with $|A|=k$
where $k > \max\{8, 2+\log_2 m\}$).
Then the minimal degree of $K(A)^{\Gamma}$ is at most 
$|\wtw(A)|$.
\end{proposition}
\begin{proof}
Let $L(A)\le\sym(\Gamma)$ denote the pointwise stabilizer of 
$\Gamma\setminus \wtw(A)$ in $K(A)^{\Gamma}$.  Then,
by the ``Unaffected Stabilizer Theorem,'' (Thm.~\ref{thm:unaffected}),
$L(A)^{A}\ge\alt(A)$.
\end{proof}



\subsection{Aggregating the local certificates}
\label{sec:aggregate}

We continue the notation of the previous section.

\begin{theorem}[AggregateCertificates]  \label{thm:aggregate}
Let $\vf : G\to \sym(\Gamma)$ be a giant representation,
where $G\le\sym(\Omega)$,\ $|\Omega|=n$, and $|\Gamma|=m$.
Let $\max\{8, 2+\log_2 n\}< k < m/10$.  
Then, at a multiplicative cost of $m^{O(k)}$, we can
either find a canonical colored $4/5$-partition of $\Gamma$
or find a canonically embedded $k$-ary relational
structure with relative symmetry defect $\ge 1/2$ on $\Gamma$,
or reduce the determination of $\iso_G(\xx,\yy)$ to
$n^{O(1)}$ instances of size $\le 2n/3$.
\end{theorem}

\begin{proof}
We describe the procedure, interspersed with the justification.

Run the {\sf LocalCertificates} routine for both
inputs $\xx,\yy$ and all test sets $A\in\binom{\Gamma}{k}$.  

Run the {\sf CompareLocalCertificates} routine for 
all pairs $((\xx,A), (\xx',A'))$ where $\xx$ is fixed,
$\xx'\in\{\xx,\yy\}$, and 
$A,A'\in\binom{\Gamma}{k}$ are test sets
(a total of $2\binom{m}{k}^2$ runs).

Let $F(\xx)$ be the subgroup generated by the groups $K(A)$ for all
full subsets $A\in\binom{\Gamma}{k}$ with reference to input string
$\xx$.  So $F(\xx)$, and with it $F(\xx)^{\Gamma}$, are canonically
associated with $\xx$.  In particular, if $F(\yy)$ is analogously
defined for $\yy$, then $F(\xx)^{\Gamma}$ is permutationally
isomorphic to $F(\yy)^{\Gamma}$, \ie, there exists a permutation
$\alpha\in\sym(\Gamma)$ such that 
$F(\yy)^{\Gamma} = \alpha^{-1}F(\xx)^{\Gamma}\alpha$.

Below we ignore $\yy$ and focus on $\xx$, omitting it from the notation,
so we write $F=F(\xx)$.  But our guide is the above consequence
of canonicity.

\begin{itemize}
\item[1] If there exists a full subset $A\in\binom{\Gamma}{k}$ such that
     $|\wtw(A)|\ge m/5$ then by individualizing each element of $A$
     we reduce the question to $N$-isomorphism at a multiplicative
     cost of $m^{O(k)}$.  
     The $N$-orbits on $\wtw(A)$ have length $\le m/k$, so the $N$-orbits
     on $\Gamma$ form a $4/5$-partition.  It follows
     by Cor.~\ref{cor:effect-on-tuples2} that each $N$-orbit
     on $\Omega$ has size $\le 4n/5$.  Process via Chain Rule, \bexit\ 
\item[2] (: Now $(\forall A)(|\wtw(A)|< m/5)$ :) 
\item[2a] If the nontrivial orbits (orbits of length $\ge 2$) 
     of $F^{\Gamma}$ cover at least $m/5$
     elements of $\Gamma$ and no orbit of $F^{\Gamma}$ has length 
     $>4m/5$ we found a colored $4/5$-partition of $\Gamma$, \bexit
\item[2b] (: $F^{\Gamma}$ has an orbit $C\subseteq\Gamma$
        of length $|C|>4m/5$.  Note that since $|C| >m/2$, this
        orbit is canonical. :)
\item[2b1] Assume $F^{C}$ is doubly transitive. \\
     {\bf Claim:}\  If $F^{C}$ is doubly transitive then it is a giant,
     \ie, $F^{C} \ge \alt(C)$.
     \begin{proof}[Proof of Claim]
     Assume $F^C$ is doubly transitive.
     If $F^C$ is not a giant, it follows that its minimal degree
     is $\ge |C|/4>m/5$ assuming $|C|\ge 217$ (Bochert's theorem,
     Thm.~\ref{thm:bochert}) for which $m\ge 272$ is sufficient. 
     But the minimal degree of $F^C$ is at most the
     minimal degree of $K(A)^{\Gamma}$ which is
     at most $|\wtw(A)|$ by Prop.~\ref{prop:mindeg5},
     a contradiction with the assumption that $|\wtw(A)|< m/5$.
     \end{proof} 
     (: So $F^{C}\ge\alt(C)$ :) \\
     Now apply Cor.~\ref{cor:topaction-col}
\item[2b2] (: $F^{\Gamma}$ is transitive but not doubly transitive :) \\
     Let $\xxx=(C;R_1,\dots,R_r)$ be the orbital configuration
     of $F^C$ (the $R_i$ are the orbits of $F^C$ on $C\times C$).
     This is a non-clique homogeneous coherent configuration, so 
      $3\le r\le m$. \\
     (: Warning: the numbering of the $R_i$ is not canonical;
        isomorphisms may permute the $R_i$. :) \\
     Let $R_1=\diag(C)$ be the diagonal \\
     (: so for $i\ge 2$ the constituents
     $X_i=(C,R_i)$ are nontrivial biregular digraphs :) \\
     Individualize one of the $X_i$ $(i\ge 2)$ 
     (: multiplicative cost $r-1 \le m-1$ :) \\
     \breturn\ $X_i$, \bexit \\
     (: Note: $X_i$ has relative symmetry defect $\ge 1/2$
     by Cor.~\ref{cor:reg-defect} 
     because $X_i$ is an irreflexive, biregular, nontrivial digraph. :)
\item[2c] Let $D\subseteq\Gamma$ be the set of fixed points of $F^{\Gamma}$.
     So in the remaining case we have
      $|D|\ge 4m/5$. 
     Note that in this case, if $A\subset D$ then $A$ is not full. 
     (In fact even if $A\cap D\neq\emptyset$ then $A$ is not full.)\\
     {\bf Claim} (Turning local asymmetry into global irregularity) \\
     In time $m^{O(k)}$ 
     we can construct a canonical $k$-ary relational structure
     on $D$ with symmetry defect (much) greater than $1/2$.  
     \begin{proof}
     We apply Prop.~\ref{prop:localguide} (Local guides).  To do so,
     we need to define the relevant categories.  Let $\xx_1$ and $\xx_2$
     (rather than $\xx$ and $\yy$) denote our two input strings.
     Let $D_i$ be the subset $D$ derived from input $\xx_i$.
     We apply Prop.~\ref{prop:localguide} with the assignment
     $\Omega_i\leftarrow D_i$ of variables.

    The objects of the category $\calL$ correspond to the pairs 
    $(A,i)$ where $A\in\binom{D_i}{k}$ is a test set.
    The set of morphisms $(A,i)\to (A',j)$ are the bijections
    $A\to A'$ corresponding to the set
    $\iso_G\left((\xx_i)_A^{W_i(A)},(\xx_j)_{A'}^{W_j(A')}\right)$
    for all $A,A'\in\binom{\Gamma}{k}$, where $W_i$ 
    corresponds to $W$ under input $\xx_i$.

The two abstract objects of category $\calC$ are
denoted $\xxx_1$ and $\xxx_2$.  The underlying set of $\xxx_i$
is $\Box(\xxx_i)=D_i$.  The morphisms are the bijections
$D_1\to D_2$ induced by the $G$-isomorphisms $\xx_1\to \xx_2$.

Our current assumption is that the objects in $\calL$ are
not full in our sense, \ie, $\aut(A,i)\le M_i(A)$ where
$M_i(A)\ngeq \alt(A)$.  In particular it follows that
the objects in $\calL$ are not full in the sense
of Prop.~\ref{prop:localguide}, \ie, $\aut(A,i)\neq\sym(A)$.

Thus the assumptions of Prop.~\ref{prop:localguide}
are satisfied.   The algorithm of Prop.~\ref{prop:localguide}
returns canonical $k$-ary relational structures on $D_i$ with strong 
symmetry defect $\ge |D_i|-k+1 > m/2$.
     \end{proof}

     Now \breturn\ this canonical $k$-ary relational structure, \bexit
\end{itemize}

This completes the procedure and the proof.
\end{proof}

\section{Effect of discovery of canonical structures}
\label{sec:structure-discovered}
Situation: We have a transitive group $G\le\sym(\Omega)$ 
of degree $n=|\Omega|$ and a giant representation
$\vf : G\to\sym(\Gamma)$ (\ie, $G^{\,\vf}\ge\alt(\Gamma)$).
Assume $m:=|\Gamma|\ge 10\log_2 n$.  Let
$\Phi$ be the set of standard blocks for $\vf$
(see the Main Structure Theorem, Thm.~\ref{thm:altquotient},
so $\Phi=\{B_T : T\in\binom{\Gamma}{t}\}$.  The $B_T$ partition 
$\Omega$ and form a system of imprimitivity for $G$.

In this section we study the effect of canonical structures
embedded in $\Gamma$.  

Both our group-theoretic partitioning algorithm
(AggregateCertificates, Theorem~\ref{thm:aggregate})
and our combinatorial partitioning algorithm
(the Extended Design Lemma, Theorem~\ref{thm:extended-design}) 
produce a canonical coloring of $\Gamma$
with an additional canonical structure on some of the 
color classes.  The additional structure can be an
equipartition or a Johnson scheme.
(We note that canonicity in each case is relative to
arbitrary choices previously made and correspondigly came
at a multiplicative cost.)

\subsection{Alignment of input strings, reduction of group}
\label{sec:align}

A common features of the categories of these types of structures is
that their $G^{\,\vf}$-isomorphisms are easy to find
($G^{\,\vf}$ is either $\sym(\Gamma)$ or $\alt(\Gamma)$).
(This is trivial in linear time for colored equipartitions, and
polynomial time for Johnson schemes).

We use these structures to align the input strings $\xx$ and $\yy$
and reduce the group $G$.

Let $\xxx(\zz)$ be the canonical structure associated with the 
input string $\zz\in\{\xx,\yy\}$.  Alignment means that 
$\xxx(\xx)=\xxx(\yy')$ for a $G$-shifted copy $\yy'$ of $\yy$.

\mn
{\sf Procedure Align}

\mn
Input: canonical structures $\xxx(\xx)$, $\xxx(\yy)$ on $\Gamma$ \\
Output: string $\yy'$, permutation $\sigma\in G$, and
        group $G_1\le G$ such that 
\begin{equation}   \label{eq:align}
  \iso_G(\xx,\yy)=\iso_{G_1}(\xx,\yy')\sigma \text{\quad and\quad}
        G_1^{\,\vf} =\aut(\xxx(\xx))
\end{equation}
\noindent
(: Note that it follows that $\xxx(\xx)=\xxx(\yy')$ :)

Additional output if $\xxx$ has a dominant color class $\Delta\subseteq\Gamma$
($|\Delta|> m/2$) and $\xxx$ involves an equipartition of $\Delta$
or a Johnson scheme on $\Delta$: reduced set $\Gamma'$ and
giant representation $G\to\sym(\Gamma')$ for recursive 
processing of the corresponding window $\Omega(\Delta)$.

\begin{enumerate}
\item    \label{item:aligned-reject}
If $\xxx(\xx)$ and $\xxx(\yy)$ are not $G^{\,\vf}$-isomorphic
then reject isomorphism, exit

\item   \label{item:aligned}
Else, let 
\begin{enumerate}[(i)]
 \item $\sigmabar\in\iso_{G^{\,\vf}}(\xxx(\xx),\xxx(\yy))$ \qquad\ 
               (: aligning in $\Gamma$ :)
 \item $\sigma\in\vf^{-1}(\sigmabar)$ \qquad\qquad\qquad\quad (: lifting :)
 \item $\yy'=\yy^{\sigma^{-1}}$ \qquad \qquad \qquad \qquad
               (: aligning the inputs :)
 \item $G_1=\vf^{-1}(\aut(\xxx(\xx)))$ \qquad\ \   (: reducing the group :)
\end{enumerate}
\noindent
(: Alignment as stated in Eq.~\eqref{eq:align} achieved :)

\item
Update: $\yy\leftarrow \yy'$, $G\leftarrow G_1$.

\item
(: Each of our structures has an underlying coloring -- possibly trivial :)

\sn
Let $(\Delta_1,\dots,\Delta_k)$ be the coloring of $\xxx(\xx)$
(the $\Delta_j$ are the color classes); so $\Gamma$ is the
disjoin union of the $\Delta_j$.

This coloring induces a canonical coloring of $\Phi=\binom{\Gamma}{t}$ as 
described in Lemma~\ref{lem:effect-on-tuples}; 
let $\Phi_1,\dots,\Phi_s$ be the color classes.
This coloring in turn lifts to a canonical coloring of $\Omega$ with
corresponding color classes $\Omega_1,\dots,\Omega_s$ 
where $\Omega_i=\bigcup_{T\in\Phi_i} B_T$.  
For $A\subseteq\Gamma$ recall the notation $\Phi(A)=\binom{A}{t}$ and 
$\Omega(A)=\bigcup_{T\in\Phi(A)}B_T$.

\item   \label{item:chain1}
Apply the Chain Rule to the color classes $\Omega_i$.

\item   \label{item:chain2}
If $(\exists j)(|\Delta_j|> m/2)$ (``dominant color'')
then start the application of the Chain Rule
with the window
$\Omega(\Delta_j)=\bigcup_{T\in\binom{\Delta_j}{t}} B_T$.
    
\item
While processing window $\Omega(\Delta_j)$
 \begin{enumerate}[(A)]
 \item if $\xxx$ gives a nontrivial equipartition of $\Delta_j$
       then let $\Gamma'$ be the set of blocks
 \item if $\Delta_j$ is the vertex set of a Johnson scheme
       $\jjj(m',t')$\ $(t'\ge 2)$ then identify $\Delta_j$ 
       with $\Delta_j=\binom{\Gamma'}{t'}$ where $|\Gamma'|=m'$.
\end{enumerate}
\item  let $\vf': G\to \sym(\Gamma')$ be the induced $G$-action on 
      $\Gamma'$\quad (: this is a giant representation :)
\item update: $\vf\leftarrow\vf'$,\ $\Gamma\leftarrow\Gamma'$

\end{enumerate}
{\sf end(procedure)}

\subsection{Cost analysis}
We are assuming that isomorphism of our canonical structures
$\xxx$ is testable in polynomial time (which is certainly true
for the types of structures considered), so Line~\ref{item:aligned}
is executed in polynomial (in $m$) time.

We need to examine the efficiency of the application of the Chain rule
in Lines~\eqref{item:chain1},~\eqref{item:chain2}.

We measure complexity in terms of the number of group operations.
We assume $G$ and a giant representation $\vf : G\to \sym(\Gamma)$
are given where $G\le\sym(\Omega)$ with $|\Omega|=n$ and
$|\Gamma|=m$.  Let $T(G,\vf)$ be the maximum cost over all input
strings for the pair $(G,\vf)$.  

We use the notation of Section~\ref{sec:cost-estimate}.
So $T_{\jh}(x,y)$ is the maximum of $T(G,\vf)$ over all 
$G$ and $\vf$ with the parameters $n\le x$ and $m\le y$.
Moreover, $T_{\jh}(x)$ is defined as $T_{\jh}(x)=T_{\jh}(x,x)$.
$T(x)$ is the upper bound for all groups $G$ of degree $n\le x$.
(Note that $n$ is the ``window size.'')

We are looking a function $T(x)$ that is ``nice''
in the sense that $\log\log T(x)/\log\log x$ is monotone nondecreasing
for sufficiently large $x$.  (For the function $\exp((\log x)^c)$,
this quantity is constant.) 

In analyzing the complexity, we need to take into account the
potentially quasipolynomial (in terms of $m$), say $q(m)$,
multiplicative cost of reaching our canonical structures $\xxx$: we
need to compare not one but $q(m)$ instances of $\xxx(\yy)$ with
$\xxx(\xx)$).  So the overall cost, including the application of
the Chain rule, will be
\begin{equation}    \label{eq:chain3}
        T(G,\vf) \le q(m)\sum_i T(|\Omega_i|)
\end{equation}

If $(\forall i)(|\Omega_i|\le 2n/3)$ then this yields (generously)
the inequality
\begin{equation}
        T(G,\vf) \le m\cdot q(m)T(2n/3),
\end{equation}
justifying Inequality~\eqref{eq:quasipoly1}.  (In fact, for ``nice''
functions as postulated, we obtain $T(G, \vf) \le q(m)(T(n/3)+T(2n/3))$.
But this gain of a factor of $m$ will make no difference.)

If $(\exists i)(|\Omega_i|>2n/3)$ then by 
Lemma~\ref{lem:effect-on-tuples}, for this $i=i_0$ we must have
$\Omega_{i_0}=\Omega(\Delta_j)$ where $|\Delta_j|>2m/3$.  The total
contribution of all other $\Omega_i$ to the right-hand side of
Eq.~\eqref{eq:chain3} is at most $q(m)T(n/3)$.  

Our progress on $\Omega(\Delta_j)$ is measured in terms of the
reduced $\Gamma$.  In the case of an equipartition, $\Gamma'$ is
the set of blocks of the partition, so $|\Gamma'|\le m/2$.
In case of a Johnson scheme $\jjj(m',t')$ $(t'\ge 2)$ with
vertex set $\Delta_j=\binom{\Gamma'}{t'}$, we have 
$m\ge |\Delta_j|=\binom{m'}{t'}\ge \binom{m'}{2}> (m'-1)^2/2$,
so $m' < 1+\sqrt{2m} < m/2$ (for $m\ge 12$).
So in each case we obtain the inequality
\begin{equation}
    T(G,\vf)\le q(m)(T(n/3)+T_{\jh}(n,m/2))
\end{equation}
justifying Eq.~(v) in Sec.~\ref{sec:cost-estimate} and
yielding the conclusion
\begin{equation}
    T(n)\le q(n)^{O(\log^2 n)}
\end{equation}
as in Eq.~\eqref{eq:cost4}.

\section{The Master Algorithm}
\label{sec:master}

The algorithm will refer to a polylogarithmic function
$\ell(x)$ to be specified later.

Whenever a subroutine in the algorithm exits and returns
a good color-partition of $\Omega$, the algorithm 
starts over (recursively).   If it returns
a structure such as a UPCC, we move to the next line.
If the subroutine returns isomorphism rejection,
that branch of the recursion terminates and the
algorithm backtracks.

\mn
{\sf Procedure String-Isomorphism}

\mn
Input: group $G\le\sym(\Omega)$, strings $\xx,\yy : \Omega\to\Sigma$

\mn
Output: $\iso_G(\xx,\yy)$


\begin{enumerate}
\item Apply {\sf Procedure Reduce-to-Johnson} 
            (Luks reductions, Sec.~\ref{sec:reducetojohnson}) \\
      (: The rest of this algorithm constitutes
          the {\sf ProcessJohnsonAction} routine
         announced in Sec.~\ref{sec:reducetojohnson}) 
\item (: $G$ is transitive, $G$-action $\ggg$ on blocks is Johnson group 
         isomorphic to $\sss_m$ or $\aaa_m$ :)\\
         set $\ell = (\log n)^3$ \\
      \bif\ $m\le \ell$ \bthen\ apply strong Luks reduction to reduce to
         kernel of the $G$-action on the blocks (brute force on
         small primitive group $\ggg$, multiplicative cost $\ell!$ :)
\item (: $G$-action on blocks is isomorphic to 
         $\sym(\Gamma)$ or $\alt(\Gamma)$,\ $|\Gamma|=m > \ell$ :) \\
         Let $\vf: G\to \sym(\Gamma)$ be a giant representation
         (inferred from $\ggg$)\\
         Let $N=\ker(\vf)$ and let 
         $\Phi=\{B_T\mid T\in\binom{\Gamma}{t}\}$ be the set of standard
           blocks (Thm.~\ref{thm:altquotient})
           (: the $B_T$ partition $\Omega$ and $G$ acts on $\Phi$ as
               $\sym^{(t)}(\Gamma)$ or $\alt^{(t)}(\Gamma)$ :)
\item \bif\ $G$ primitive (: \ie, $\Omega=\Phi$ :) \\
         \indent\quad \bif\ $t=1$ \bthen\ 
               find $\iso_G(\xx,\yy)$, \bexit \quad
              (: trivial case: $\Omega=\Gamma$, $G=\alt(\Omega)$; \\
         \indent\quad\quad
              isomorphism only depends on the multiplicity of each 
              letter in the strings $\xx,\yy$ :) 
\item    \indent\quad \belse\ (: $t\ge 2$ :) 
          view $\xx,\yy$ as edge-colored $t$-uniform hypergraphs 
           $\calH(\xx)$ and $\calH(\yy)$ \\
         \indent\quad\quad\quad on vertex set $\Gamma$ \\
         \indent\quad\quad
           \bif\ relative symmetry defect of $\calH(\xx)$ is $< 1/2$
          \bthen\ apply Cor.~\ref{cor:topaction-col}  
\item    \indent\quad\quad  \belse\
           (: now their relative symmetry defect is $\ge 1/2$ :) \\
         \indent\quad
           (: view these hypergraphs as $t$-ary relational structures :)\\
         \indent\quad\quad
         apply Extended Design Lemma (Theorem~\ref{thm:extended-design}) 
\item    \indent\quad (: canonical structure $\xxx$ on $\Gamma$ found: 
            colored equipartition or Johnson scheme :) \\
         \indent\quad\quad
          apply {\sf Procedure Align} to $\xxx$ (Sec.~\ref{sec:align})
\item  \belse\ (: $G$ imprimitive, \ie, $|\Phi|\le (1/2)|\Omega|$ :) \\  
         \indent\quad\quad
          apply {\sf AggregateCertificates} (Theorem~\ref{thm:aggregate}) \\
         \indent\quad
          (: Note: this is where our main 
             group-theoretic Divide-and-Conquer algorithm, \\
         \indent\quad\quad
             {\sf Procedure LocalCertificates} 
             (Theorem~\ref{thm:local-certificates}) is used :) 
\item    \indent\quad
           \bif\ {\sf AggregateCertificates} returns canonically embedded 
              $k$-ary relational structure on $\Gamma$ \\
         \indent\quad\quad
              with relative symmetry defect $\ge 1/2$ \bthen\ 
\item    \indent\quad\quad
         apply Extended Design Lemma (Theorem~\ref{thm:extended-design}) \\
         \indent\quad\quad
           $\xxx\leftarrow$ canonical structure on $\Gamma$ returned \\
         \indent\quad (: $\xxx$ is a colored equipartition of $\Gamma$
              or a Johnson scheme embedded in $\Gamma$ :) 
\item    \indent\quad
           \belse\ (: {\sf AggregateCertificates} returns 
              canonical colored equipartition on $\Gamma$ :)\\
         \indent\quad\quad
              $\xxx\leftarrow$ colored equipartition returned
\item    \indent\quad\quad
             apply {\sf Procedure Align} to $\xxx$ (Sec.~\ref{sec:align})



\end{enumerate}

The essence of the analysis is in the analysis of 
{\sf Procedure Align} given in Section~\ref{sec:align}.



\section{Concluding remarks}
\subsection{Dependence on the Classification of Finite Simple Groups}
\label{sec:CFSG}
As mentioned in the Introduction, the analysis of the
algorithm, as stated, depends on
the Classification of Finite Simple Groups (CFSG) via Cameron's
classification of large primitive permutation groups.  There is one
other instance in which we rely on CFSG; we employ ``Schreier's
Hypothesis'' in the proof of Lemma~\ref{lem:prim-topquotient}.

We are, however, able to considerably reduce the dependence of the
analysis on CFSG; we are able to do \emph{without Cameron's result} by
one more application of the {\sf Procedure UPCC Split-or-Johnson}
(Theorem~\ref{thm:UPCC}) and some 80-year-old group theory.

Cameron's result guaranteed that if $G$ acted as a large primitive
group $\ggg\le \sym(\Phi)$ on the set $\Phi$ of blocks of a minimal system
of imprimitivity (the blocks are maximal), then $\ggg$ was a Cameron
group, which in turn either had a transitive, imprimitive subgroup of
small index (so Luks reduction was applicable) or a Johnson group.
This reduction was done in {\sf Procedure Reduce-to-Johnson}
(Sec.~\ref{sec:reducetojohnson}).

We are able to replace this procedure by one that does not rely on
Cameron's result; we locate this Johnson group combinatorially.  Here
is an outline.

Let $k=|\Phi|$ be the number of blocks.

If $\ggg$ is uniprimitive (primitive but not doubly transitive)
then let $\xxx$ be the orbital configuration of $\ggg$, defined
as the coherent configuration on $\Phi$ where the color classes
are the orbitals of $\ggg$, \ie, the orbits of $\ggg$ on 
$\Phi\times\Phi$.  Now $\xxx$ is uniprimitive because $\ggg$
is uniprimitive, and the color classes are by definition
$G$-invariant.  Apply {\sf Procedure UPCC Split-or-Johnson}
(Theorem~\ref{thm:UPCC}) to $\xxx$.  The procedure either returns
a canonical colored $3/4$-partition of $\Phi$,
representing significant progress, or
returns a canonically embedded Johnson scheme $\jjj(m,t)$
on a subset $J$ of $\Phi$ of size $|J|=\binom{m}{t}\ge 3k/4$.  
After breaking up $\Phi$ via the Chain rule,
we shall be left with $J$ (Lemma~\ref{lem:effect-on-tuples}).
The $G$-action on $\jjj(m,t)$ is a subgroup of $\sss_m^{(t)}$ and
can be represented on the set $[m]$ which is much smaller than $\Phi$
($k\ge \binom{m}{2}$, so $m < 1+\sqrt{2k}$).
If this is a giant action (the image contains $\aaa_m$), we are
in the same situation as if we had used Cameron's theorem.
If the action is not giant, we recurse (find orbits and 
minimal block system for the action on $[m]$, etc.).

This completes the case when $\ggg$ is uniprimitive.

If $\ggg$ is not uniprimitive then $\ggg$ is doubly transitive.
Giants are the $t=1$ case of Johnson groups, so if $\ggg$ is a giant,
we are done.  So we may now assume that $\ggg$ is doubly transitive
but not a giant.  We could conclude now by applying strong Luks
reduction to the kernel of the $G\to\ggg$ epimorphism (brute force on
$\ggg$), with reference to an elementary result by Pyber~\cite{pyber}
that the order of $\ggg$ is quasipolynomially bounded (as a function
of $k$).  But we can make our algorithm even more efficient and the
analysis even more elementary by limiting the group theory used to 
an old reference.

Let $t$ be the degree of transitivity of $\ggg$, \ie, the greatest 
integer $t$ such that $\ggg$ is $t$-transitive (transitive on the set of
$k(k-1)\cdots(k-t+1)$ ordered $t$-tuples of elements of $\Phi$).
We use the following 1934 result by Wielandt~\cite{wielandt}
(cited in the Remarks after Thm.~9.7 in~\cite{wielandtbook}).

\begin{theorem}[Wielandt]
Let $\ggg\le \sss_k$ be doubly transitive and not a giant.
Then the degree of transitivity of $\ggg$ is $t < 3\ln k$.
\end{theorem}

We continue with the procedure.
Pick $S\subset\Phi$ with $|S|=t-1$.  Individualize the elements of 
$S$.  Now the group $\ggg_{(S)}$ (pointwise stabilizer of $S$) is
transitive but not doubly transitive in its action on $\Phi\setminus S$.  
If $\ggg_{(S)}$ is imprimitive on $\Phi\setminus S$ then we reduce 
$\Phi$ to the set $\Phi'$ corresponding to the blocks of imprimitivity
of $\ggg_{(S)}$; so we now have $k':=|\Phi'|\le k/2$ blocks, significant 
progress.  Otherwise, $\ggg_{(S)}$ is uniprimitive on $\Phi\setminus S$,
so we are back to the case already discussed.

\begin{remark}
Stronger bounds hold on the degree of transitivity.
Under CFSG, we have $t\le 5$, and in fact $t\le 3$ if $k\ge 25$.
Moreover, Wielandt~\cite{wielandt-schreier} (see \cite[Thm. 7.3A]{dixon})
has shown that assuming
only Schreier's Hypothesis, one can prove $t\le 7$.
So 6 individualizations (rather than $3 \ln k$
individualizations) suffice in the above argument if we are
willing to assume Schreier's Hypothesis.
(Note also that if we do encounter a 7-transitive group that is not
a giant, we shall have found an explicit counterexample to
Schreier's Hypothesis and thereby to CFSG, an impressive by-product.)
\end{remark}


\subsection{How easy is Graph Isomorphism?}  \label{sec:conclusions2}
The first theoretical evidence against the possibility of
NP-completeness of GI was the equivalence of existence and 
counting~\cite{poznan,mathon}, not observed in any NP-complete problem.  
The second, stronger evidence came from the early theory of interactive
proofs: graph isomorphism is in coAM, and therefore if GI is
NP-complete then the polynomial-time hierarchy collapses to the
second level (Goldreich--Micali--Wigderson 1987~\cite{gmw}).  Our result
provides a third piece of evidence: GI is not NP-complete unless
all of NP can be solved in quasipolynomial time.

A number of questions remain.  The first one is of course whether GI
is in P.  Such expectations should be tempered by the
status of the \emph{Group Isomorphism} problem\footnote{In
 complexity theory, the ``Group Isomorphism Problem'' 
refers to groups given by Cayley tables; in other words,
complexity is compared to the order of the group.  From
the point of view of applications, this complexity measure is
of little use; in computational group theory, groups are usually 
given in compact representations
(permutation groups, matrix groups given by lists of generators,
$p$-groups given by power commutator presentation, etc.).
But the fact remains that even in the unreasonably redundant
representation by Cayley tables, we are unable to solve the
problem is polynomial time.}:  
given two groups by their Cayley tables, are they isomorphic?  
It is easy to reduce this problem to GI.
In fact, Group Isomorphism seems much easier than GI; it can
trivially be solved in time $n^{O(\log n)}$ where $n$ is the order of
the group.  But in spite of considerable effort and the availability
of powerful algebraic machinery, Group Isomorphism is still not 
known to be in P.  We are not even able to decide Group
Isomorphism\footnote{A simple algorithm, proposed by Tim Gowers
on Dick Lipton's blog in November 2011, has a chance of running in 
$n^{O(\sqrt{\log n})}$.   Let the \emph{$k$-profile} of a finite group
$G$ be the function $f$ on isomorphism types of $k$-generated groups
where $f(H)$ counts those $k$-tuples of elements of $G$ that generate
a subgroup isomorphic to $H$.  For what $k$ do $k$-profiles discriminate
between nonisomorphic groups of order $n$?  It is known that
$k < (1/2)\sqrt{\log_2 n}$ is insufficient
for infinitely many values of $n$
(Glauberman, Grabowski~\cite{ggg}).
Whether some $k$ that is not much greater than $\sqrt{\log n}$
suffices is an open question that I think would deserve
attention.  The test case is $p$-groups of class 2; the
Glauberman--Grabowski examples belong to this class.} in 
time $n^{o(\log n)}$.

A closely related challenge that deserves attention is the
String Isomorphism problem on $n=p^k$ points, with respect to
the linear group GL$(k,p)$.  The order of this group is
about $p^{k^2}=n^{\log_p n}$; the question is, can this problem
be solved in time $p^{o(k^2)}$ (or perhaps even in $\poly(n)$ time).
I note that this problem can be encoded as a GI problem for
graphs with $\poly(n)$ vertices so if GI\,$\in$\,P then this problem
is in P as well.

The result of the present paper amplifies the significance of the
Group Isomorphism problem (and the challenge problem stated) as a
barrier to placing GI in P.  It is quite possible that the
intermediate status of GI (neither NP-complete, nor polynomial
time) will persist.

In fact, even putting GI in coNP faces the same obstacle: 
Group Isomorphism is not known to be in coNP.

\subsection{How hard is Graph Isomorphism?}  \label{sec:conclusions}
Paradoxically, from a structural complexity point of view, GI (still)
seems harder than factoring integers.  The decision version of
Factoring (given positive integers $x,y$, does $x$ have divisor
$d$ in the interval $2\le d\le y$?) is in NP\,$\cap$\,coNP
while the best we can say about GI is NP\,$\cap$\,coAM.  
Factoring can be solved in polynomial time on a quantum computer,
but no quantum advantage has yet been found for GI.  On the other
hand, apparently hard instances of factoring abound, whereas
we don't know how to construct hard instances of GI.
Could this be an indication that in structural complexity maybe
we are not asking the right questions?

Even more baffling is another complexity arena, where GI is provably
hard, on par with many NP-hard problems: relaxation hierarchies in
proof complexity theory (Lov\'asz--Schrijver, Sherali--Adams,
Sum-of-Squares hierarchies).  Building on the seminal paper by Cai,
Furer, and Immerman~\cite{cfi}, increasingly powerful hierarchies
have recently been shown to be unable to refute isomorphism of
graphs on sublinear levels~\cite{atserias,odonnell,snook}, showing
that GI tests based on these hierarchies necessarily have
exponential (even factorial) complexity.  However,
hard-to-distinguish CFI pairs of graphs and the related pairs of
which isomorphism is hard to refute in these hierarchies are
vertex-colored graphs with bounded color classes.  Testing
isomorphism of such pairs of graphs was shown to be in polynomial
time via the first application of group theory (1979/80) that used
hardly more than Lagrange's Theorem from group
theory~\cite{lasvegas, fhl}.  One lesson is that these hierarchies
have difficulty capturing the power of even the most naive
applications of group theory.  Given that hardness with respect to
these hierarchies can now be proved by reduction from GI, this
raises the question, in what sense these hierarchies indicate
hardness.

\subsection{Outlook}  \label{sec:conclusions4}
On the bright side, a number of GI-related questions may look a bit
more hopeful now.  While GI is complete over the isomorphism problems
of \emph{explicit structures}, there are interesting classes of
non-explicit structures where progress may be possible.  Two important
examples are \emph{equivalence of linear codes} and 
\emph{conjugacy (permutational equivalence) of permutation groups.}
The former easily reduces to the latter.
Both of these problems belong\footnote{To see that these problems
belong to coAM, one can adapt the GMW protocol~\cite{gmw}
by conjugating the group by a random permutation and choosing 
a uniform random set of $O(n)$ generators.}
to NP\,$\cap$\,coAM and therefore they
are not NP-complete unless the polynomial-time hierarchy collapses.
In spite of this complexity status, 
no moderately exponential ($\exp(n^{1-c})$) algorithm is 
known for either problem.  GI reduces to each of these 
problems~\cite{luks-dimacs}\footnote{Luks's reduction is
\href{http://mathforum.org/kb/thread.jspa?forumID=253\&threadID=561418\&messageID=1681072\#1681072}{explained 
by Miyazaki in a post on The Math Forum, Sep. 29, 1996}.}.
Regarding both problems, see also~\cite{bcgq, bcq}.  

The present paper does not address the question of 
\emph{canonical forms.}  Do graphs permit quasipolynomial-time 
computable canonical forms?

It would be of great interest to find stronger structural results to
better correspond to the ``local $\to$ global symmetry'' philosophy.  
This raises difficult mathematical questions that our algorithmic 
divide-and-conquer techniques bypass, but results of this flavor
could make the algorithm more elegant and more efficient.

Finally a more concrete question.  Let $\xxx=(V;\calR)$ be a
homogeneous coherent configuration with $n$ vertices.
Let $W\subseteq V$, $|W|\ge \alpha n$.  Suppose that the
induced configuration $\xxx[W]$ is a Johnson scheme.
Is there a constant $\alpha <1$ such that this 
implies that $\xxx$ itself is a Johnson scheme?

A result in this direction could be a step toward
an elementary characterization of the Cameron groups 
as the only primitive groups of large order,
or somewhat less ambitiously,
an elementary characterization of the Johnson groups 
as the only primitive groups of large order,
without an imprimitive subgroup of small index.
Steps toward these goals
have previously been made in~\cite{uniprimitive} for 
the case $|G|> \exp(n^{1/2+\epsilon})$ and in a remarkable recent 
paper by Sun and Wilmes~\cite{sun-wilmes} for the case 
$|G|> \exp(n^{1/3+\epsilon})$.

\subsection{Analyze this!}
The purpose of the present paper is
to give a guaranteed upper bound (worst-case analysis); it
does not contribute to practical solutions.  
It seems, for all practical purposes, the Graph Isomorphism
problem is solved; a suite of remarkably efficient 
programs is available (nauty, saucy, Bliss, conauto, Traces).
The article by McKay and Piperno~\cite{mckay-piperno} 
gives a detailed comparison of methods and performance.
Piperno's article~\cite{piperno} gives a detailed description
of \emph{Traces}, possibly the most successful program for
large, difficult graphs.

These algorithms provide ingenious shortcuts in backtrack search.
One of the most important questions facing the theorist
in this area is to analyze these algorithms.  While Miyazaki's
graphs provide hard cases for the early version of \emph{nauty},
the recent update overcomes that difficulty.   

The question is, does there exist an infinite family of pairs of
graphs on which these heuristic algorithms fail to perform
efficiently?  The search for such pairs might turn up
interesting families of graphs.  

Alternatively, can one prove strong worst-case upper bounds on the
performance of any of these algorithms?

The comparison charts in~\cite{mckay-piperno} seem to suggest that
we lack true benchmarks -- difficult classes of graphs on which
to compare the algorithms.   Encoding class-2 $p$-groups 
as graphs could provide quasipolynomially difficult examples,
but right now we have no guarantee that the heuristics could not be
tricked into much worse, (moderately?) exponential behavior.

\section*{Acknowledgments}
I am happy to acknowledge the inspiration gained from my
recent collaboration on the structure, automorphism group,
and isomorphism problem for highly regular combinatorial
structures 
with my student John Wilmes as well as with Xi Chen, Xiaorui Sun,
and Shang-Hua Teng~\cite{bw,bcstw,delsarte}.  The recent
breakthrough on primitive coherent configurations by Sun and
Wilmes~\cite{sun-wilmes} was particularly encouraging; at one point
during the weeks before the completion of the present work, it
served as a tool to breaking the decades-old $\exp(\wto(\sqrt{n}))$
barrier (see Remark~\ref{rem:sunwilmes}).

The most direct forerunner of this paper was
my joint work with Paolo Codenotti on hypergraph
isomorphism~\cite{codenotti}; that paper combined 
the group theory method with a web of combinatorial 
partitioning techniques.  In particular, I found an 
early version of the Design Lemma in the wake of that work.
(A much simpler observation is called ``Design Lemma''
in that paper.)  

Some of the group theory used in the present paper was inspired by
my joint work with P\'eter P\'al P\'alfy and Jan Saxl~\cite{saxl};
in particular, the rendering of a result of Feit and
Tits~\cite{feit-tits} in that paper turned out to be particularly
handy in the proof of the main group theoretic lemma of this paper
(``Unaffected Stabilizer Theorem,'' Theorem~\ref{thm:unaffected}).

I am grateful to three long-time friends who helped me verify
critical parts of this paper: P\'eter P\'al P\'alfy and L\'aszl\'o
Pyber the proof of various versions of the group-theoretic ``Main
structure theorem'' (Theorem~\ref{thm:altquotient}) that includes
the crucial ``Unaffected Stabilizer Theorem,'' and Gene Luks the
{\sf LocalCertificates} procedure
(Sec.~\ref{sec:localcertificates}), the core algorithm of the
paper.  Their comments helped improve the presentation, and, more
significantly, raised my confidence that these items actually work.
All other parts of the paper seem quite ``fault-tolerant,'' with
multiple solutions, and a bag of tricks to rely on, should any gaps
be found.  Naturally, any errors that may remain in these items (or
any other part of the paper) are my sole responsibility.

I wish to thank several colleagues, and especially Thomas Klimpel
and P\'eter P. P\'alfy, for their careful reading of parts of the
first arXiv version and pointing out a large number of typos and
some inaccuracies.  The present update
(second arXiv version) corrects these and adds an occasional
clarification.  New content has been added to
Section~\ref{sec:CFSG} (Dependence on CFSG).
The following sections have been updated for better exposition
but without any significant new content:
Sections~\ref{sec:classical} 
(Classical coherent configurations),
\ref{sec:designlemma} 
(Design Lemma), 
\ref{sec:localguide} 
(Local guides), 
\ref{sec:resilient} 
(Resilience of Johnson schemes), 
\ref{sec:localguides-johnson} 
(Local guides for the Johnson case),
\ref{lem:topfactor} 
(Simple quotients of subdirect products),
\ref{sec:aggregate} 
(Aggregation, item 2c: aggregating non-fullness certificates).


\begin{thebibliography}{abcdef}

\bibitem[AsS]{aschbacher}
            {\sc Michael Aschbacher and Leonard L. Scott:}
            Maximal subgroups of finite groups.
            \emph{J. Algebra} {\bf 92} (1985), 44--80.
\bibitem[AtM]{atserias} {\sc Albert Atserias and Elitza Maneva:}
            Graph Isomorphism, Sherali--Adams Relaxations and
            Indistinguishability in Counting Logics.
            \emph{SIAM J. Comp.} {\bf 42(1)}, 2013, 112--137.
\bibitem[Ba77]{poznan} {\sc L\'aszl\'o Babai:}
            On the isomorphism problem.
            Manuscript, 1977.  Cited in~\cite{mathon}
\bibitem[Ba79a]{lasvegas} {\sc L\'aszl\'o Babai:}
            Monte {C}arlo algorithms in graph isomorphism testing.
            Tech. Rep. 79--10, {D\' ep.} Math. et Stat., 
            Universit\'e de Montr\'eal, 1979 (pp. 42)
           \href{http://people.cs.uchicago.edu/~laci/lasvegas79.pdf}%
            {\tt http://people.cs.uchicago.edu/$\sim$laci/lasvegas79.pdf}
\bibitem[Ba79b]{toronto} {\sc L\'aszl\'o Babai:} Lectures on Graph Isomorphism.
         University of Toronto, Department of Computer Science.
         Mimeographed lecture notes, October 1979
\bibitem[Ba81]{uniprimitive} {\sc L\'aszl\'o Babai:}
        On the order of uniprimitive permutation groups.
        \emph{Annals of Math.} {\bf 113(3)} (1981) 553--568.
\bibitem[Ba83]{doktori} {\sc L\'aszl\'o Babai:} 
 {\em Permutation Groups, Coherent Configurations and Graph Isomorphism.}
       D.Sc. Thesis (Hungarian), Hungarian Academy of Sciences, April 1983.
\bibitem[Ba86]{chain} {\sc L\'aszl\'o Babai:}
       On the length of subgroup chains in the symmetric group.
       {\em Communications in Algebra} {\bf 14} (1986) 1729--1736.
\bibitem[Ba--]{int} {\sc L\'aszl\'o Babai:}
       Coset intersection in moderately exponential time.
       Manuscript, 2008.
       {\tt http://people.cs.uchicago.edu/$\sim$laci/int.pdf}
\bibitem[BaB]{beals}
        {\sc L\'aszl\'o Babai and Robert Beals:}
        A polynomial-time theory of black box groups I.
        In: \emph{Groups St Andrews 1997 in Bath, I}
        (C.\,M.\ Campbell \emph{et al.}, eds.), pp.~30--64,
        London Math. Soc.\ Lecture Note Ser.\ Vol.\,260,
        Cambridge U. Press, 1999.
\bibitem[BaBS]{bbs}
        {\sc L\'aszl\'o Babai, Robert Beals, \'Akos Seress:}
        Polynomial-Time Theory of Mtrix Groups (Extended Abstract).
        \emph{In: Proc. 41st ACM STOC}, 2009, pp. 55--64.
\bibitem[BaCaP]{bcp} 
        {\sc L\'aszl\'o\,Babai, Peter\,J.\,Cameron, P\'eter\,P.\,P\'alfy:}
On the orders of primitive groups with restricted
nonabelian composition factors.
{\em J.\,Algebra} {\bf 79} (1982) 161--168.
\bibitem[BaCh+]{bcstw} {\sc L\'aszl\'o Babai, Xi Chen, Xiaorui Sun,
           Shang-Hua Teng, John Wilmes:}
        Faster Canonical Forms For Strongly Regular Graphs.
        \emph{In: 54th IEEE FOCS}, 2013, pp.~157--166.
\bibitem[BaCo]{codenotti}
       {\sc L\'aszl\'o Babai, Paolo Codenotti:}
    Isomorphism of hypergraphs of low rank in moderately exponential time.
        \emph{In: Proc. 49th IEEE FOCS}, 2008, pp.~667--676. 
\bibitem[BaCGQ]{bcgq}
     {\sc L\'aszl\'o Babai, Paolo Codenotti, Joshua A. Grochow, 
    Youming Qiao:}
    Code Equivalence and Group Isomorphism.
    \emph{In: Proc. 22nd Ann. Symp. on Discrete Algorithms (SODA'11)}, 
    ACM-SIAM, 2011, pp. 1395--1408.
\bibitem[BaCQ]{bcq}
    {\sc L\'aszl\'o Babai, Paolo Codenotti, Youming Qiao:}
    Polynomial-time Isomorphism Test for Groups with no
    Abelian Normal Subgroups (Extended Abstract).
    \emph{In: Proc. 39th Internat. Colloq. on Automata, 
     Languages and Programming (ICALP'12)}, 
     Springer LNCS 7391, 2012, pp. 51--62.
\bibitem[BaKL]{bkl} {\sc L\'aszl\'o Babai, William M. Kantor, Eugene M. Luks:}
        Computational complexity and the classification of finite simple
        groups.
       \emph{In: Proc. 24th IEEE FOCS}, 1983, pp.~162--171.
\bibitem[BaL]{canonical} {\sc L\'aszl\'o Babai, Eugene M. Luks:}
         Canonical labeling of graphs.
         \emph{In: Proc. 15th ACM STOC}, 1983, pp.~171--183.
\bibitem[BaLS]{nc} {\sc L\'aszl\'o Babai, Eugene M. Luks, \'Akos Seress:}
         Permutation groups in NC.
         \emph{In: Proc. 19th ACM STOC}, 1987, pp.~409--420.
\bibitem[BaPS]{saxl} {\sc L\'aszl\'o Babai, P\'eter P. P\'alfy, Jan Saxl:}
         On the number of $p$-regular elements in finite simple groups.
         \emph{LMS J. Comput. and Math.}, {\bf 12} (2009) 82--119.
\bibitem[BaS]{transitivity} 
         {\sc L\'aszl\'o Babai, \'Akos Seress:}
         On the degree of transitivity of permutation groups: a short proof.
         \emph{J. Combinatorial Theory---A} {\bf 45} (1987) 310--315.
\bibitem[BaW1]{bw} {\sc L\'aszl\'o Babai, John Wilmes:}
         Quasipolynomial-time canonical form for {S}teiner designs.
         \emph{In: Proc. 45th ACM STOC}, 2013, pp.~261--270.
\bibitem[BaW2]{delsarte} {\sc L\'aszl\'o Babai, John Wilmes:}
         Asymptotic Delsarte cliques in distance-regular graphs.
         \emph{J. Algebraic Combinatorics,} to appear.
         See \href{http://arxiv.org/abs/1503.02746}{arXiv:1503.02746}
\bibitem[Bo89]{bochert89} 
         {\sc Alfred Bochert:}
         \"Uber die {Z}ahl verschiedener {W}erthe, 
         die eine Funktion gegebener Buchstaben durch Vertauschung
         derselben erlangen kann.
         \emph{Math. Ann.} {\bf 33} (1889) 584--590.
\bibitem[Bo97]{bochert97} 
         {\sc Alfred Bochert:}
         \"Uber die {C}lasse der transitiven {S}ubstitutionengruppen {II}.
         \emph{Math. Ann.} {\bf 49} (1897) 133--144.
\bibitem[CaiFI]{cfi} {\sc {Jin-Yi} Cai, Martin F\"{u}rer, Neil Immerman:}
       An optimal lower bound on the number of variables for graph
       identification. \emph{Combinatorica} \textbf{12} (1992) 389--410.
\bibitem[Cam1]{cameron} {\sc Peter\,J. Cameron:}
  Finite permutation groups and finite simple groups,
  \emph{Bull. London Math Soc.} {\bf 13} (1981) 1--22.
\bibitem[Cam2]{cameronblog} {\sc Peter\,J. Cameron:}
         The symmetric group, 12. \emph{Blog article}, 21/04/2011,
  \href{http://cameroncounts.wordpress.com/2011/04/21/the-symmetric-group-12/}{http://cameroncounts.wordpress.com/2011/04/21/the-symmetric-group-12/}
\bibitem[CST]{cst}
        {\sc Xi Chen, Xiaorui Sun, Shang-Hua Teng:}
        Multi-stage design for quasipolynomial-time
        isomorphism testing of {S}teiner 2-systems.
        \emph{In: Proc. 45th ACM STOC}, 2013, pp.~271--280.
\bibitem[DiM]{dixon} {\sc John D. Dixon, Brian Mortimer:}
        \emph{Permutation Groups.}
        Springer Grad. Texts in Math. vol. 163, 1996 
\bibitem[FeT]{feit-tits} 
        {\sc Walter Feit and Jacques Tits:}
        Projective representations of minimum degree of group extensions.
        \emph{Canad. J. Math.} {\bf 30} (1978) 1092--1102.
\bibitem[FuHL]{fhl} {\sc Merrick Furst, John Hopcroft, Eugene Luks:}
       Polynomial-time algorithms for permutation groups. 
       \emph{In: Proc. 21st IEEE FOCS}, 1980, pp.~36--41.
\bibitem[GlG]{ggg} {\sc George Glauberman and {\L}ukasz Grabowski:}
        Groups with identical $k$-profiles.
        Manuscript, 2015.
\bibitem[GoMW]{gmw} {\sc Oded Goldreich, Silvio Micali, and
        Avi Wigderson:}
        Proofs that yield nothing but their validity and
        a methodology of cryptographic protocol design.
        \emph{In: Proc. 27th IEEE FOCS}, 1986, pp.~174--187.
\bibitem[HoS]{holt} {\sc Derek F. Holt, Mark J. Stather:}
        Computing chief series and the soluble radical of
        a matrix group over a finite field.
      \emph{LMS J. Computation and Mathematics} {\bf 11} (2008) pp. 223--251.
\bibitem[ImL]{immerman} {\sc Neil Immerman, Eric S. Lander:} 
       Describing graphs: a first-order approach to graph canonization.
       \emph{In: Complexity Theory Retrospective ---
       in honor of Juris Hartmanis on the occasion of his 
       60th birthday, July 5, 1988} (Alan Selman, ed.),
       Springer 1990, pp. 59--81.
\bibitem[Jor]{jordan}
       {\sc Camille Jordan:}
       \emph{Trait\'e des substitutions et des equations alg\'ebriques.}
       Gauthier--Villars, 1870. (Reprinted 1957, Paris, Albert Blanchard)
\bibitem[KlL]{kleidman}
       {\sc Peter Kleidman and Martin Liebeck:}
       \emph{The Subgroup Structure of the Finite Classical Groups.}
       London Math. Soc. Lecture Note Ser.\ Vol.\,129,
       Cambridge Univ. Press, 1990.
\bibitem[Kn]{knuth} {\sc Donald\,E. Knuth:}
       Efficient representation of perm groups.
\emph{Combinatorica} {\bf 11} (1991) 57--68.
\bibitem[Lie83]{liebeck}
         {\sc Martin W. Liebeck:}
          On graphs whose full automorphism group is an
          alternating group or a finite classical group.
          \emph{Proc. London Math. Soc. (3)} {\bf 47} (1983) 337-362
\bibitem[LiePS]{liebeck-praeger-saxl}
         {\sc Martin W. Liebeck, Cheryl E. Praeger, Jan Saxl:}
         On the O'Nan--Scott theorem for finite primitive
         permutation groups. 
         \emph{J. Austral. Math. Soc. (A)} {\bf 44} (1988) 389--396
\bibitem[Lu82]{luks-bded} {\sc Eugene M. Luks:} 
        Isomorphism of graphs of bounded valence can be tested in polynomial
        time.
        \emph{J. Comput. Syst. Sci.} {\bf 25(1)} (1982) 42--65.
\bibitem[Lu87]{luks-comp} {\sc Eugene\,M. Luks:}
        Computing the composition factors of a permutation group 
        in polynomial time.
        \emph{Combinatorica} {\bf 7} (1987) 87--99.
\bibitem[Lu93]{luks-dimacs} {\sc Eugene\,M. Luks:}
        Permutation groups and polynomial-time computation.
        \emph{In: Groups and Computation,} 
        DIMACS Ser. in Discr. Math. and Theor. Computer Sci.
        {\bf 11} (1993) 139--175.
\bibitem[Lu99]{luks-hyp} {\sc Eugene\,M. Luks:}
        Hypergraph Isomorphism and Structural Equivalence of 
        Boolean Functions. 
        \emph{In: 31st ACM STOC}, 1999, pp. 652-658.
\bibitem[Mar]{maroti}
        {\sc Attila Mar\'oti:}
        On the orders of primitive groups.
        \emph{J. Algebra} {\bf 258(2)} (2002) 631--640.
\bibitem[Mat]{mathon}
        {\sc Rudi Mathon:}
         A note on the graph isomorphism counting problem.
         \emph{Info. Proc. Lett.} {\bf 8} pp.~131--132.
\bibitem[McP]{mckay-piperno}
        {\sc Brendan D. McKay and Adolfo Piperno:}
        Practical Graph Isomoprhism, II.
        \href{http://arxiv.org/abs/1301.1493}{\tt arXiv:1301.1493}, 2013.
\bibitem[Me]{meier}
        {\sf Ulrich Meierfrankefeld:}
        Non-finitary locally finite simple groups.
        \emph{In: Finite and Locally Finite Groups.}
        B. Hartley et al., eds., Kluwer 1995, pp.~189--212.
\bibitem[Mi]{miyazaki} {\sc Takunari Miyazaki:}
        Luks's reduction of Graph isomorphism to 
        code equivalence.  Comment on The Math Forum,
        Sep. 29, 1996.
  \href{http://mathforum.org/kb/thread.jspa?forumID=253\&threadID=561418\&messageID=1681072\#1681072}{\tt http://mathforum.org/kb/thread.jspa?forumID=253\&threadID=561418\& messageID=1681072\#1681072}
\bibitem[OWWZ]{odonnell} 
       {\sc Ryan O'Donnell, John Wright, Chenggang Wu, Yuan Zhou:}
        Hardness of robust graph isomorphism, Lasserre gaps,
        and asymmetry of random graphs.
        \emph{In: Proc. 25th ACM--SIAM Symp. Disr. Alg. (SODA'14)},
        2014, pp. 1659--1677.
\bibitem[Pi]{piperno} {\sc Adolfo Piperno:}
        Search Space Contraction in Canonical Labeling of Graphs.
        \href{http://arxiv.org/abs/0804.4881}{\tt arXiv:0804.4881},
        2008, v2 2011.
\bibitem[Py]{pyber} {\sc L\'aszl\'o Pyber:}
        On the orders of doubly transitive permutation groups,
        elementary estimates.
        \emph{J. Combinatorial Theory, Ser A} {\bf 62(2)} (1993) 361--366.
\bibitem[Sc]{sc} {\sc Leonard L. Scott:}  
        Representations in characteristic $p$.
        \emph{In: The Santa Cruz Conference on Finite Groups}, 1980,
        Amer. Math. Soc., pp. 319--322.
\bibitem[Se]{seress-book} {\sc \'Akos Seress:}
       \emph{Permutation Group Algorithms.}
        Cambridge Univ. Press, 2003
\bibitem[Si1]{sims1} {\sc Charles\,C. Sims:}
        Computation with Permutation Groups.
        \emph{In: Proc.~2$^{nd}$ Symp. Symb. Algeb. Manip.}
        (S.\,R. Petrick,ed.),
        ACM, New York, 1971, pp. 23--28.
\bibitem[Si2]{sims2} {\sc Charles\,C. Sims:}
        Some group theoretic algorithms.
       \emph{In: Lecture Notes in Math.} Vol. 697, Springer, 1978, pp. 108-124.
\bibitem[SnSC]{snook} {\sc Aaron Snook, Grant Schoenebeck, Paolo Codenotti:}
         Graph Isomorphism and the Lasserre Hierarchy.
         \href{http:arxiv.org/abs/1401.0758}{\tt arXiv:1401.0758}
\bibitem[Sp]{spielman}{\sc Daniel \,A. Spielman:}
        Faster Isomorphism Testing of Strongly regular Graphs.
        In: \emph{Proc. 28th ACM STOC}, 1996, pp.~576--584.
\bibitem[SuW]{sun-wilmes}
       {\sc Xiaorui Sun and John Wilmes:}
       Faster canonical forms for primitive coherent configurations.
       \emph{In: Proc. 47th STOC}, 2015, pp. 693--702.
\bibitem[We]{weisfeiler-book} {\sc Boris Weisfeiler} (ed.):
        \emph{On Construction and Identification of Graphs.}
        Springer Lect. Notes in Math. Vol 558, 1976.
\bibitem[WeL]{weisfeiler-leman} {\sc Boris Weisfeiler, Andrei A. Leman:}
        A reduction of a graph to a canonical form and an algebra
        arising during this reduction.
        \emph{Nauchno-Technicheskaya Informatsiya} {\bf 9} (1968) 12--16.
\bibitem[Wi1]{wielandt} {\sc Helmut Wielandt:}
        Absch\"atzungen f\"ur den Grad einer Permutationsgruppe
        von vorgeschriebenem Transitivit\"atsgrad.
        Dissertation, Berlin, 1934. \emph{Schriften Math. Seminars
        Inst. Angew. Math. Univ. Berlin} {\bf 2} (1934) 151--174.
\bibitem[Wi2]{wielandt-schreier} {\sc Helmut Wielandt:}
        \"Uber den Transitivit\"atsgrad von Permutationsgruppen.
        \emph{Math. Z.} {\bf 74} (1960) 297--298.
\bibitem[Wi3]{wielandtbook} {\sc Helmut Wielandt:}
         \emph{Finite Permutation Groups}.
         Acad. Press, New York 1964.
\bibitem[ZKT]{zemlyachenko}
        {\sc Viktor N. Zemlyachenko, Nikolai M. Korneenko, Regina I.
        Tyshkevich:}  Graph isomorphism problem.
        \emph{Zapiski Nauchnykh Seminarov {LOMI}} {\bf 118} (1982) 
        83--158, 215.
\end{thebibliography}
\end{document}